\documentclass[manuscript,nonacm]{acmart}

\AtBeginDocument{
  \providecommand\BibTeX{{
    \normalfont B\kern-0.5em{\scshape i\kern-0.25em b}\kern-0.8em\TeX}}}

\setcopyright{acmcopyright}
\copyrightyear{2021}
\acmYear{2021}
\acmDOI{10.1145/1122445.1122456}

\acmBooktitle{(under review)}
\acmPrice{15.00}
\acmISBN{978-1-4503-XXXX-X/18/06}

\usepackage{comment}
\usepackage{dsbda}

\newcommand{\summarizeAdjustedMethod}{Parallel\-Hash\-Messaging\-Algorithm}

\usepackage[utf8]{inputenc}

\usepackage{tabularx,booktabs}
\usepackage{colortbl}
\usepackage{multirow}
\usepackage{dblfloatfix}
\usepackage{url}
\usepackage[ruled,linesnumbered,noend]{algorithm2e}
\usepackage{subcaption}
\usepackage[capitalise,nameinlink]{cleveref}
\usepackage{listings}
\usepackage{bm}

\usepackage{xifthen}
\usepackage{pdfpages}

\newcommand{\enquote}[1]{``#1''}

\crefname{section}{Sect.}{Sect.}
\Crefname{section}{Section}{Sections}
\crefname{figure}{Fig.}{Fig.}
\Crefname{figure}{Figure}{Figures}
\crefname{example}{Ex.}{Ex.}
\crefname{example}{Example}{Examples}
\crefname{algorithm}{Alg.}{Alg.}
\Crefname{algorithm}{Algorithm}{Algorithms}
\crefname{lstlisting}{Listing}{Listings}

\newcommand{\images}{figures/}

\newcommand{\model}{FLUID}

\newcommand{\vertexHashIndex}{Vertex\-Update\-Hash\-Index}
\newcommand{\summarizeMethod}{Parallel\-Summarize}

\newcommand{\addSchema}{ADD-SG} 
\newcommand{\modifySchema}{MOD-SG} 
\newcommand{\deleteSchema}{DEL-SG} 
\newcommand{\addInstance}{ADD-PE} 
\newcommand{\modifyInstance}{MOD-PE} 
\newcommand{\deleteInstance}{DEL-PE}

\newcolumntype{L}[1]{>{\Centering\arraybackslash}p{#1}}
\newcolumntype{Y}[2]{*{#1}{>{\hsize=#2\hsize\centering\arraybackslash}X}}
\mathchardef\mhyphen="2D

\mathchardef\<="3C
\mathchardef\>="3E
\newcommand{\id}{\mathrm{id}}
\newcommand{\idrel}{\mathrm{id_{rel}}}

\newcommand{\PAY}{\mathrm{PAY}}
\newcommand{\PC}{\mathrm{PC}}
\newcommand{\OC}{\mathrm{OC}}

\newcommand{\OCtype}{\OC_{\mathrm{type}}}
\newcommand{\PCrel}{\PC_{\mathrm{rel}}}

\newcommand{\nsout}[1]{\Gamma^+(#1)}

\newcommand{\simS}{{\sim_s}}
\newcommand{\simP}{{\sim_p}}
\newcommand{\simO}{{\sim_o}}

\newcommand{\Vvs}{V_{\mathrm{vs}}}
\newcommand{\Evs}{E_{\mathrm{vs}}}
\newcommand{\Vpe}{V_{\mathrm{pe}}}
\newcommand{\Epe}{E_{\mathrm{pe}}}
\newcommand{\SigmaKey}{\Sigma_{\text{Key}}}
\newcommand{\SigmaValue}{\Sigma_{\text{Value}}}

\newcommand{\GLPG}{G_\text{LPG}}
\newcommand{\GRDF}{G_\text{RDF}}

\newcommand{\SGincr}{SG_{\mathrm{incr}}}
\newcommand{\SGbatch}{SG_{\mathrm{batch}}}
\newcommand{\idprev}{id_{\mathrm{prev}}}
\newcommand{\vsprev}{vs_{\mathrm{prev}}}
\newcommand{\GDB}{GDB}

\begin{document}
\lefthyphenmin=4
\righthyphenmin=4
\sloppy

\title[A Parallel Algorithm for Graph Summarization, Incremental Summarization, and Iterative $k$-Bisimulation]{A Parallel Algorithm for Graph Summarization and its Extension to Incremental Summarization and Iterative $k$-Bisimulation}

%\begin{comment}
\title
[Parallel Algorithm for Graph Summaries and two Extensions to Incremental Summarization and $k$-Bisimulation]
{Time and Memory Efficient Parallel Algorithm for Structural Graph Summaries and two Extensions to Incremental Summarization and $k$-Bisimulation for Long $k$-Chaining}
%\end{comment}

\author{Till Blume}
\orcid{0000-0001-6970-9489}
\email{till.blume@de.ey.com}
\affiliation{
  \institution{ Ernst \& Young GmbH WPG – R\&D}
  \city{Berlin}
  \country{Germany}
}
\author{Jannik Rau}
\orcid{0000-0001-7764-6131}
\email{jannik.rau@uni-ulm.de}
\affiliation{
  \institution{Data Science and Big Data Analytics Group, Ulm University}
  \city{Ulm}
  \country{Germany}
}
\author{Ansgar Scherp}
\orcid{0000-0002-2653-9245}
\email{ansgar.scherp@uni-ulm.de}
\affiliation{
  \institution{Data Science and Big Data Analytics Group, Ulm University}
  \city{Ulm}
  \country{Germany}
}

\author{David Richerby}
\orcid{https://orcid.org/0000-0003-1062-8451}
\email{david.richerby@essex.ac.uk}
\affiliation{
  \institution{University of Essex}
  \city{Colchester}
  \country{UK}
}

\addtocontents{toc}{\protect\setcounter{tocdepth}{-1}}

\begin{abstract}
Most graph summarization algorithms are tailored to a specific graph summary model and were designed for a one-time batch computation of a summary. 
We developed a flexible parallel algorithm for graph summarization based on vertex-centric programming and parameterized message passing.
The base algorithm supports infinitely many structural graph summary models defined in a formal language.
An extension of the parallel base algorithm allows incremental graph summarization, \ie updating the summary as the graph evolves.\footnote{
\textbf{Note on extension of this paper to prior work:}
This paper is based on our earlier publication of the incremental graph summarization algorithm~\cite{DBLP:conf/cikm/BlumeRS20} published at CIKM 2020. 
In this paper, we make several significant extensions: 

First, we provide a formal proof of correctness for the incremental update algorithm, demonstrating that the incremental algorithm is guaranteed to produce the same summary as the batch-based summarization algorithm.

Second, although the incremental algorithm for graph summarization published at CIKM 2020 supports $k$-bisimulation for values of $k>1$, it requires nested data structures for the message passing.
Thus, in this paper we extend the base summarization algorithm by an efficient hash-based messaging mechanism to support scalable iterative computation of graph summaries based on $k$-bisimulation for arbitrary~$k$.
This new variant generates unique hashes at each iteration of a bisimulation computation and passes these to neighboring vertices instead of full summaries, a technique due to Schätzle \etal{}~\cite{DBLP:conf/sigmod/SchatzleNLP13}. 
Thus, we call this new algorithm the parallel hash-based graph summarization algorithm.

Third, we have significantly extended the experimental analysis of graph summary computation, on both synthetic and real-world graph datasets. 
We use datasets ranging from tens of millions of edges up to one billion edges, a significant increase over the original CIKM 2020 paper. 
Two new experiments are conducted for the incremental algorithm.
One experiment investigates the effect of parallelizing the incremental algorithm versus the batch-based algorithm on multiple compute cores. 
The second measures the memory overhead of the \vertexHashIndex{} data structure in the incremental algorithm.
Finally, a fourth experiment is added to the paper investigating the scalability of the new parallel hash-based graph summarization algorithm on datasets with up to one billion edges.
%
%Note, as we support the open science movement, we made the  additional experiments available in an arXiv  preprint~\cite{DBLP:journals/corr/abs-2111-12493}.
%After acceptance of the manuscript to the journal, we plan to update the arXiv paper with an explicit statement how to cite the journal.
%This will increase visibility and the number of citations to the journal article.
}
In this paper, we prove that the incremental algorithm is correct and show that updates are performed in time $\mathcal{O}(\Delta \cdot d^k)$, where $\Delta$ is the number of additions, deletions, and modifications to the input graph, $d$ the maximum degree, and $k$ is the maximum distance in the subgraphs considered. 
Although the iterative algorithm supports values of $k>1$, it requires nested data structures for the message passing that are memory-inefficient.
Thus, we extended the base summarization algorithm by a hash-based messaging mechanism to support a scalable iterative computation of graph summarizations based on $k$-bisimulation for arbitrary~$k$.
We empirically evaluate the performance of our algorithms using benchmark and real-world datasets.
The incremental algorithm almost always outperforms the batch computation.
We observe in our experiments that the incremental algorithm is faster even in cases when $50\%$ of the graph database changes from one version to the next.
The incremental computation requires a three-layered hash index, which has a low memory overhead of only $8\%$ ($\pm 1\%$) compared to its batch-based version.
Finally, the incremental summarization algorithm outperforms the batch algorithm even when fewer cores are available.
The iterative parallel $k$-bisimulation algorithm computes summaries on graphs with over $10$M edges within seconds.
We show that the algorithm scales and processes graphs of $100+\,$M edges within a few minutes while having a moderate memory consumption below of $150$~GB.
For the largest BSBM1B dataset with 1~billion edges, it computes $k=10$ bisimulation in under an hour, \ie we need only four to five minutes per iteration.
\end{abstract}

\maketitle

\section{Introduction}
\label{sec:introduction}

Structural graph summaries are condensed representations of graphs that preserve selected characteristics of the original graph~\cite{DBLP:journals/vldb/CebiricGKKMTZ19,DBLP:journals/csur/LiuSDK18}.
Different attempts have been made to classify existing graph summarization approaches~\cite{DBLP:journals/csur/LiuSDK18,DBLP:journals/corr/abs-2004-14794,DBLP:journals/pvldb/KhanBB17,DBLP:journals/vldb/CebiricGKKMTZ19}.
Structural graph summaries precisely capture specific structural features of the data graph~\cite{DBLP:journals/vldb/CebiricGKKMTZ19}.
Early examples of structural graph summaries include Goldman and Widom's DataGuides~\cite{DBLP:conf/vldb/GoldmanW97}, the representative objects of Nestorov \etal~\cite{DBLP:conf/icde/NestorovUWC97} and Milo and Suciu's T-indexes~\cite{DBLP:conf/icdt/MiloS99}.
Structural graph summaries are defined using various mathematical notations such as quotient graphs \citep{DBLP:conf/semweb/CebiricGM17} or $k$-bisimulation~\cite{DBLP:conf/sigmod/SchatzleNLP13,DBLP:conf/icde/KaushikSBG02}.
Other graph summary approaches include statistical summarization~\cite{DBLP:journals/vldb/CebiricGKKMTZ19}, based on frequent pattern mining, and neighborhood-preserving compression~\cite{DBLP:conf/www/ShinG0R19,DBLP:conf/kdd/KoKS20,tiptap2021}, which summarizes unlabeled graphs into supergraphs that approximate the number of neighbors.
We focus on structural graph summaries, as they are used in many real-life applications~\cite{DBLP:conf/sigmod/FanLLTWW11} such as cardinality estimations in graph databases~\cite{DBLP:conf/www/StefanoniMK18,DBLP:conf/icde/NeumannM11},
data exploration~\cite{lodex2015,loupe2015,DBLP:conf/semWeb/PietrigaGADCGM18,DBLP:conf/esws/SpahiuPPRM16a}, data visualization~\cite{DBLP:journals/vldb/GoasdoueGM20},
vocabulary term recommendations~\cite{DBLP:conf/esws/SchaibleGS16}, related entity retrieval~\cite{DBLP:conf/www/CiglanNH12}, and query answering in data search~\cite{DBLP:conf/kcap/GottronSKP13}.
However, we do not focus on a specific application: rather, we produce a general framework that can compute many different summaries, suited to many different tasks.

A general property of structural graph summaries is that they partition the set of vertices in a graph based on equivalences of subgraphs~\cite{DBLP:journals/tcs/BlumeRS21}.
To determine subgraph equivalences, only structural features are used, such as specific combinations of labels.
An example of a data graph is shown in the top-left of \cref{fig:schema-evolution}. 
The graph contains two subgraphs $G1$ and~$G2$, which originate from distinct sources $X$ and~$Y$.
Both graphs contain vertices labeled \textsf{Book}, \textsf{Subject}, and \textsf{Person}, and edges labeled \textsf{topic} and \textsf{author}.
One can define a graph summary (SG) using an equivalence relation~$\sim$ that summarizes vertices that have the same labels and are connected to vertices with the same labels by edges with the same labels.
If two equivalent (sub)graphs are found in the data graph, only a single new subgraph in SG is created that preserves the information about the equivalent (sub)graphs from the data graph.
For instance, since $v1$ and~$v7$ are equivalent under~$\sim$, they are summarized by the same subgraph $vs1$ in SG, starting with the its root vertex~$r1$. 
This graph summary SG is shown in the top right of \cref{fig:schema-evolution}.
It preserves the information about the combinations of graph labels (\textsf{Book}$-$\textsf{topic}$\rightarrow$\textsf{Subject} and \textsf{Book}$-$\textsf{author}$\rightarrow$\textsf{Person}) found in the data graph (GDB) at time $t$.
In addition, a so-called payload~\cite{DBLP:conf/kcap/GottronSKP13} can be added to the~SG (not shown in the figure).
This could be, \eg the number of equivalent subgraphs found (here, two) or a list of the subgraph sources (here, $X$ and~$Y$).
For different tasks, different payloads are used~\cite{DBLP:journals/tcs/BlumeRS21}.
Structural graph summaries are often an order of magnitude smaller than the input graph~\cite{DBLP:journals/vldb/CebiricGKKMTZ19}.
This allows tasks to be executed faster on the summary instead of the original graph.

\begin{figure*}[!t]
\centering
\includegraphics[trim={0cm 0cm 0cm 0.5cm},
clip,width=0.98\linewidth]
{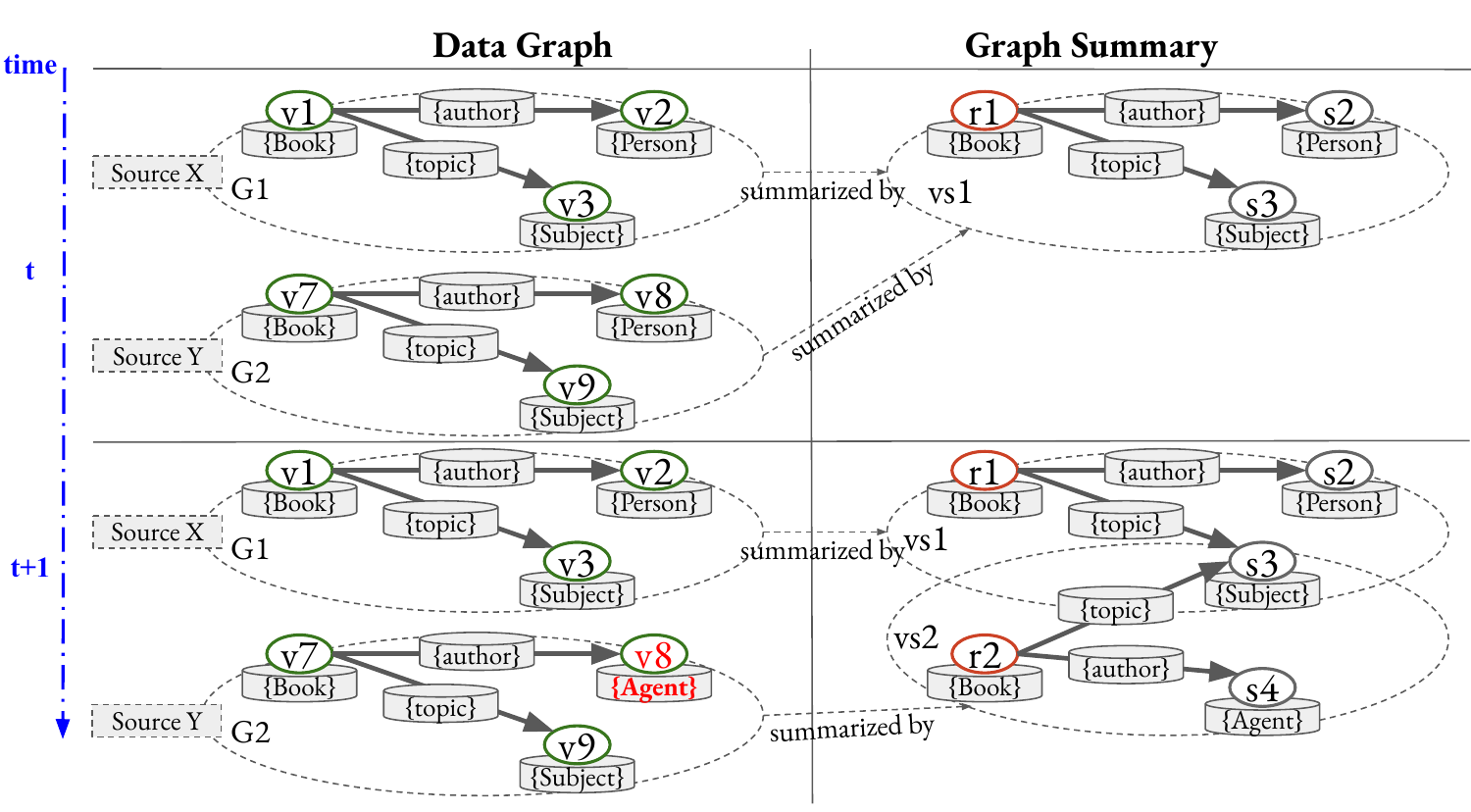}
\caption{\label{fig:schema-evolution}\label{fig:schema-example}Evolving graph database (GDB) from points in time $t$ to $t+1$ and part of its graph summary (SG) computed based on the SchemEX graph summary model~\cite{DBLP:journals/ws/KonrathGSS12}.
}
\end{figure*}

The data graph shown in the top-left of \cref{fig:schema-evolution} is observed at a time~$t$.
At time~$t + 1$, shown in the bottom-left of \cref{fig:schema-evolution}, the data graph has changed.
While vertex $v8$ was labeled \texttt{Person} at time~$t$, it is now labeled \texttt{Agent}.
This means, the structural graph summary needs to be updated.
Since $v1$ and $v7$ are no longer equivalent under $\sim$ they are now summarized by different subgraphs in the graph summary.
When the data graph evolves over time, it is often prohibitively expensive to recompute the structural graph summary from scratch, especially when only few changes occur in relation to the overall size of the data graph.
Thus, an algorithm is needed that can efficiently update previously computed structural graph summaries, including additions, modifications, and deletions of vertices and edges.
Graph summarization algorithms do not consider modifications and deletions~\cite{DBLP:journals/ws/KonrathGSS12,DBLP:conf/edbt/BaaziziLCGS17,DBLP:conf/icde/HegewaldNW06,DBLP:conf/edbt/GoasdoueGM19}, which we can expect to occur in graphs that evolve over time.
Our algorithm can update the graph summary without requiring a change log, \ie a list of changed vertices and edges from one version to the other of the GDB.
Especially for graphs on the Semantic Web, updates often do not provide a reliable change log~\cite{DBLP:conf/esws/KaferAUOH13,DBLP:conf/semWeb/DividinoGS15}.
This could be overcome by storing local copies of the data graphs~\cite{DBLP:conf/semWeb/DividinoGS15} but this is impractical for huge and distributed graphs like on the Semantic Web~\cite{DBLP:conf/semWeb/DividinoGS15,DBLP:journals/semweb/RietveldBHS17} and defeats the purpose of graph summarization to produce a significantly smaller representation.

\subsection{Contributions}
In~\cite{DBLP:conf/cikm/BlumeRS20}, we proposed a parallel algorithm for graph summarization, which uses message-passing in a vertex-centric approach, to compute and merge partial results of the graph summary~\cite{DBLP:conf/sigmod/MalewiczABDHLC10,DBLP:journals/semweb/StutzSB16}.
The algorithm is designed following a formal model that defines structural graph summaries using equivalence relations~\cite{DBLP:journals/tcs/BlumeRS21}. 
This base algorithm allows us to process large graphs in parallel and in a distributed computing environment.

To incrementally update a graph summary as the underlying graph changes over time (see \cref{fig:schema-evolution}) we extended the base algorithm and introduce a new, complex data structure called the \vertexHashIndex{}~\cite{DBLP:conf/cikm/BlumeRS20}.
The \vertexHashIndex{} is a three-layered hash data structure that stores links between vertices in the input graph, the graph summary, and the payload.
It contains the information needed to update the graph summary when changes occur.
This allows us to update only those parts of the graph summary that need to change.
The incremental algorithm automatically detects changes in the data graph including deletions and modifications of graph vertices.
To the best of our knowledge, there exists no other incremental structural graph summarization algorithm for ``truly'' evolving graphs, \ie one that automatically detects and supports additions, deletions, and modifications.
The only algorithms applicable for such evolving graphs are approaches that we call incremental subgraph indices and they typically require a change log as input~\cite{DBLP:conf/sigmod/FanLLTWW11, DBLP:journals/pvldb/MinPPGIH21} (see discussion in the related work in \cref{sec:related-work-incremental}).
In~\cite{DBLP:conf/cikm/BlumeRS20}, we also showed that all FLUID graph summaries can be updated in $\mathcal{O}(\Delta \cdot d^k)$, where $\Delta$~is the number of additions, deletions, and modifications in the input graph, $d$~is the maximum degree of the input graph, and $k$~is the maximum distance in the subgraphs considered.
Detecting changes, however, takes time linear in the size of the input graph.
If a change log is provided, this detection phase can be omitted.

In this paper, we provide a formal proof of correctness for the incremental update algorithm, demonstrating that the incremental algorithm is guaranteed to produce the same summary as the batch-based summarization algorithm.
Furthermore, although the incremental algorithm for graph summarization~\cite{DBLP:conf/cikm/BlumeRS20} supports $k$-bisimulation for values of $k>1$, it requires nested data structures for the message passing.
Thus, we extend in this paper the base summarization algorithm by an efficient hash-based messaging mechanism to support scalable iterative computation of graph summaries based on $k$-bisimulation for arbitrary~$k$.
This new variant generates unique hashes at each iteration of a bisimulation computation and passes these to neighboring vertices instead of full summaries, a technique due to Schätzle \etal{}~\cite{DBLP:conf/sigmod/SchatzleNLP13}. 
Thus, we call this new algorithm the parallel hash-based graph summarization algorithm.

We empirically evaluate the performance of our base algorithm, our incremental algorithm, and our parallel hash-messaging  algorithm on benchmark datasets and real-world datasets using representative summary models.
For the base algorithm and the incremental algorithm, we use the three summary models Attribute Collection~\cite{DBLP:conf/dexaw/CampinasPCDT12}, Class Collection~\cite{DBLP:conf/dexaw/CampinasPCDT12}, and SchemEX~\cite{DBLP:journals/ws/KonrathGSS12}.
For the parallel hash-messaging algorithm, we experiment with the graph summary model defined by the $k$-bisimulation of Schätzle \etal~\cite{DBLP:conf/sigmod/SchatzleNLP13}.
In our first experiment, originally published in~\cite{DBLP:conf/cikm/BlumeRS20}, we show that the incremental summarization algorithm almost always outperforms its batch counterpart.
We observe that the incremental algorithm is faster, even when about $50\%$ of the graph database changes from one version to the next~\cite{DBLP:conf/cikm/BlumeRS20}. 
In this first experiment, we use a fixed number of cores.

In this paper, we add three additional experiments:
Our second experiment systematically investigates the influence of the number of available cores on the performance of the algorithms.
The incremental summarization algorithm outperforms the batch summarization algorithm even when using fewer cores, even though the batch computation benefits more from parallelization than the incremental algorithm.
In a third experiment, we measure the memory consumption of our incremental algorithm and particularly the memory overhead that is induced by the \vertexHashIndex{}.
Using real-world datasets and four cores, we found that the incremental algorithm is on average $5$ to $44$~times faster while only producing a memory overhead of $8\%$ ($\pm 1\%$).
Finally, we empirically evaluate the hash messaging extension of our base algorithm by computing a $10$-bisimulation.
The smaller data graphs with $10+\,$M edges can be processed within seconds and the larger data graphs with $100+\,$M edges in less than five minutes.
To demonstrate the scalability of our approach, we additionally use a synthetic graph with $1\,$B edges. 
Our parallel hash-messaging algorithm computes the $10$-bisimulation on a synthetic graph with $1\,$B edges in less than an hour. The memory requirement is only five times as much as for a graph of around $100\,$M edges, despite being an order of magnitude larger.

\subsection{Outline}
We define the concepts of a graph database, a graph summary model, $k$-bisimulation and our data structure for representing graph summaries in \cref{sec:foundations}.
We also formally define the two graph models, namely the Resource Description Framework~\cite{w3c-rdf} and Labeled Property Graphs~\cite{DBLP:series/synthesis/2021Hogan}, and show that they can be converted to each other.
In \cref{sec:graph-summarization}, we first describe our parallel graph summarizing algorithm.
Subsequently, we first extend the base algorithm for incrementally updating graph summaries.
We proof correctness of our incremental algorithm and analyze its complexity.
Finally, we introduce the extension of the base algorithm using the hash messaging system.
This allows us to efficiently compute graph summaries based on $k$-bisimulation for large graphs and considering neighbors up to a distance of at least $k = 10$.
In \cref{sec:incremental-apparatus}, we describe the experimental apparatus, \ie the datasets, graph summary models, and our test system.
We evaluate our algorithm in extensive experiments.
First, we compare the performance of the incremental algorithm and the batch algorithm in \cref{sec:experiment-1} and discuss our findings w.r.t.\@ our complexity analysis.
In \cref{sec:experiment-2}, we evaluate the scalability of our parallelization approach using a different number of cores.
In \cref{sec:experiment-3}, we evaluate the memory consumption of our incremental algorithm and particularly the memory overhead introduced by the \vertexHashIndex{} and discuss the results.
We demonstrate the scalability of the hash messaging algorithm on large synthetic and real-world graphs, which are reported and discussed in \cref{sec:experiment-4}. 
Finally, we discuss related work in \cref{sec:related-work}, before we conclude.

\section{Graph Models and Graph Summarization}
\label{sec:foundations}
There are two major competing frameworks to model graph data: the Resource Description Framework (RDF) and Labeled Property Graphs (LPGs). 
RDF graphs are directed, edge-labeled multi-graphs. 
Intuitively, RDF represents all information as edges, \eg properties of vertices are edges in the graph and types of vertices are linked via the property \texttt{rdf:type}.
RDF supports special kinds of vertices, \ie blank nodes and literals~\cite{DBLP:series/synthesis/2021Hogan}.
LPGs do not have special vertices.  
LPGs allow vertex labels (types) and allow property information to be attached to both vertices and edges, in the form of key--value pairs~\cite{DBLP:series/synthesis/2021Hogan}.
In Sections~\cref{sec:graph-models-rdf} and~\cref{sec:graph-models-lpg}, we formally define these two graph models and discuss their differences.
For a more extensive discussion, see~\cite{DBLP:series/synthesis/2021Hogan}.
It is common knowledge that each can be transformed into the other~-- for example, both can be translated into relational databases~\cite{DBLP:series/synthesis/2021Hogan}. 
In \cref{sec:transformation}, we give an example transformation between RDF graphs and LPGs, which is also used in our experimental evaluation to transform RDF datasets.
This allows us to use the term ``data graph'', which subsumes RDF graphs and LPGs.
Without loss of generality, we will focus on LPGs for the remainder of this paper.
Subsequently, we define the concept of graph summarization in \cref{sec:graph-summary-model} and describe the key concepts of the formal language FLUID used to flexibly define different graph summary models~\cite{DBLP:journals/tcs/BlumeRS21}.
In \cref{sec:intro-bisimulation}, we introduce the concept of bisimulation and its relation to graph summarization.
Finally, we define the data structure for FLUID's structural graph summaries and analyze its data complexity.

\begin{figure}[t!]
    \centering
\begin{minipage}[b]{0.45\linewidth}
    \centering
    \small
   \begin{align*}  
  \{&(v_1,\texttt{rdf:type},\texttt{Proceedings}),\\
    &(v_2,\texttt{rdf:type},\texttt{Person}),\\
    &(v_1,\texttt{author},v_2),\\
    &(v_1,\texttt{title},\text{\enquote{Graph Database}}),\\
    &(v_2,\texttt{name},\text{\enquote{Max Power}}) \}
	\end{align*}  
\end{minipage}
\vspace{1em}
\begin{minipage}[b]{0.45\linewidth}
    \centering 
    \includegraphics[scale=.6,trim={8cm 10.2cm 5cm 3.4cm}, clip=true]{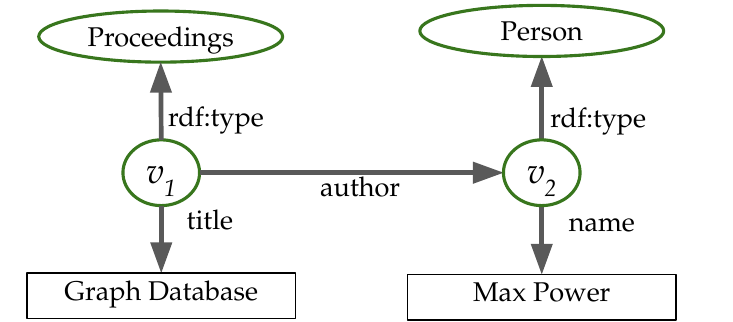}
\end{minipage}
\caption{\label{fig:rdf-example-introduction} Simple RDF graph to demonstrate the relation between a set of triples (top) and its visualization as a graph (bottom).}
\end{figure}

\subsection{RDF Graphs}
\label{sec:graph-model-rdf}
\label{sec:graph-models-rdf}
An RDF graph~\cite{w3c-rdf} is a set of triples $(s, p, o)$, with a subject~$s$, predicate~$p$, and object~$o$.
Each triple denotes a directed edge from the subject vertex~$s$ to the object vertex~$o$, the edge being labeled with the predicate~$p$.
RDF graphs distinguish three kinds of vertices: referents, represented by International Resource Identifiers (IRIs)~\cite{rfc-iri}; blank nodes; and literals. 

Formally, an RDF graph is defined as $G \subseteq (V_I\cup V_{B}) \times P \times (V_I \cup V_{B} \cup L)$, where $V_I$ denotes the set of IRIs, 
$P \subseteq V_I$ the set of predicates,
$V_{B}$ the set of blank nodes, 
and $L$ the set of literals (represented by strings)~\cite{w3c-rdf}.
IRIs conceptually correspond to real-world entities and are globally unique: an IRI may be included in more than one RDF graph, but this corresponds to stating different facts about the same real-world entity. 
In contrast, blank nodes are only locally defined, within the scope of a specific RDF graph, to serve special data modeling tasks.
Through skolemization, blank nodes can be turned into Skolem IRIs, which are globally unique~\cite[Section~3.5]{w3c-rdf}. 
Literal vertices are finite strings of characters from a finite alphabet such as Unicode~\cite{w3c-rdf}.
Thus, two literal vertices which are term-equal (\ie are the same string)~\cite[Section~3.3]{w3c-rdf} are the same vertex.
IRIs, blank nodes, and literals have different roles in RDF, but this distinction is not relevant in this work and we treat them equally.

Predicates $p \in P$ act as edge labels. 
The RDF standard includes the predicate \texttt{rdf:type}, which is used to simulate vertex labels: the triple $(s, \texttt{rdf:type}, o)$ denotes the vertex~$s$ having label~$o$. 
Such vertex labels are called RDF types and we write $V_C \subseteq V_I$ for the set of all RDF types~\cite{DBLP:journals/tcs/BlumeRS21}.
This indirect representation of vertex labels is a design decision of RDF and contrasts with the direct use of vertex labels in labeled property graphs~\cite{DBLP:series/synthesis/2021Hogan}.
Edges in an RDF graph that are not labeled with \texttt{rdf:type} are called RDF properties.
An example RDF graph is shown in \cref{fig:rdf-example-introduction}, as a set of triples (top) and as a graph (bottom).
The vertex $v_1$ has the type set $\lbrace \text{\texttt{Proceedings}}\rbrace$ and the property set $ \lbrace \text{\texttt{author, title}}\rbrace$.
The vertex $v_2$ has the type set $\lbrace \text{\texttt{Person}}\rbrace$ and the property set $\lbrace \text{\texttt{name}}\rbrace$.
Also, $v_1$ has predicates $\lbrace \texttt{author,}$ $\texttt{title,}$ $\texttt{rdf:type}\rbrace$ and outgoing neighbors $\nsout{v_1} = \lbrace \texttt{Proceedings}, v_2, \text{\enquote{Graph Database}}\rbrace$.

A special characteristic of RDF graphs is their support for semantic labels, which allows the inference of implicit information. 
Such semantic labels are from ontologies, where semantic relationships between types and properties are denoted, \eg with predicates from the RDF Schema vocabulary.
RDF Schema (RDFS) and its entailment rules are standardized by the W3C~\cite{w3c-rdf-schema}.
A comprehensive overview of these rules is presented in~\cite{DBLP:books/daglib/0028543}.
For example, the semantics of a triple $(p, \texttt{rdfs:subPropertyOf}, p')\in G$ is that, for any subject vertex~$s$ with $(s,p,o)\in G$, we can infer the existence of the additional triple $(s,p',o)$.
This means that, when using RDF Schema inference, each vertex may have more types and properties in its type set and property set, respectively. 
We briefly discuss the role of RDFS inferencing in the context of graph summarization in Section~\ref{sec:graph-summary-model}.
A more exhaustive discussion is found in \cite{DBLP:journals/tcs/BlumeRS21} and in \cite{DBLP:journals/vldb/GoasdoueGM20}.
Finally, RDF supports the definition of named graphs~\cite{w3c-rdf}.
Following Harth \etal~\cite{DBLP:conf/semweb/HarthUHD07}, we formalize named graphs by extending each triple $(s,p,o)$ to a tuple $((s, p, o), d)$ or a quad $(s, p, o, d)$, where $d$ denotes the name of the data source from which the triple originated~\cite{w3c-rdf-dataset-semantics}.

\begin{table*}[t!]
{\renewcommand\arraystretch{1.0}
\setlength{\tabcolsep}{6pt}
\small  
\caption{\label{tab:notation} Overview of the most important graph notations. The first group of symbols is used to define labeled property graph databases, the second group of symbols to define the graph summary, and the third group comprised general graph notations.}
\begin{tabularx}{\textwidth}{p{.2cm} p{5cm} X}
\toprule
& \textbf{Symbol} & \textbf{Explanation} \\
\midrule
& GDB & Labeled Property Graph database $(V,E,\mathcal{G}, \ell_G)$\\
& $V$, $E$ & Vertices and edges of a GDB \\
& $\mathcal{G}$ & Multi-set $\{G_1, \ldots, G_n\}$ of graphs in GDB \\
& $G \in \mathcal{G}$ & LP Graph $G = (V_G, E_G, \ell_V, \ell_E, \ell_P)$\\ 
& $V_G \subseteq V$ & Vertices appearing in graph $G$ \\
& $E_G \subseteq E$ & Edges appearing in graph $G$ \\
& $\ell_V\colon V \rightarrow \mathcal{P}(\Sigma_V)$ & Labeling function for all vertices in GDB\\
& $\ell_E\colon E \rightarrow \mathcal{P}(\Sigma_E)$ & Labeling function for all edges in GDB\\
& $\ell_P\colon V \cup E \rightarrow \mathcal{P}(\SigmaKey \times \SigmaValue)$ & Property function for all vertices and edges in GDB\\
& $\ell_{G}\colon G \rightarrow \Sigma_G$ & Labeling function for all graphs in GDB\\
\midrule
& $SG$ & $SG$ is a graph summary $SG = (\Vvs \cup \Vpe, \Evs \cup \Epe, \ell_V, \ell_E, \ell_P)$\\
& $vs \subseteq SG$ & $vs$ is a vertex summary, which is a subgraph of $SG$ \\
& $\Vvs$ & Vertices containing schema information in $SG$ \\
& $\Vpe$ & Vertices containing payload information in $SG$ \\
& $\Evs$ & Edges containing schema information in $SG$ \\
& $\Epe$ & Edges connecting payload information in $SG$ \\
\midrule
& $\Gamma\colon V \rightarrow \mathcal{P}(V)$ & The set of $v$'s neighbors \\ 
& $\Gamma^+\colon V \rightarrow \mathcal{P}(V)$ & The set of $v$'s outgoing neighbors\\
& $\Gamma^-\colon V \rightarrow \mathcal{P}(V)$ & The set of $v$'s incoming neighbors\\
& $d, d^+, d^-$ &  Degree, outdegree, and indegree in G\\
\bottomrule
\end{tabularx}
}
\end{table*}
 
\begin{figure}[t]
    \centering
\begin{minipage}[b]{0.45\linewidth}
    \centering
    \small
   \begin{align*}  
    &\ell_V(v_1) \rightarrow \{\texttt{Proceedings}\}\\
    &\ell_V(v_2) \rightarrow \{\texttt{Person}\}\\
    &\ell_E((v_1,v_2)) \rightarrow \{\texttt{author}\}\\
    &\ell_E((v_1,v_3)) \rightarrow \{\texttt{title}\}\\
    &\ell_E((v_2,v_4)) \rightarrow \{\texttt{name}\}\\
    &\ell_P(v_3) \rightarrow \{(\texttt{literal}, \text{Graph Database})\}\\
    &\ell_P(v_4) \rightarrow \{(\texttt{literal}, \text{Max Power})\}\\
	\end{align*}  
\end{minipage}
\vspace{1em}
\begin{subfigure}[b]{0.45\linewidth}
    \centering 
    \includegraphics[width=1\columnwidth,trim={8cm 7.5cm 5cm 4.75cm}, clip=true]{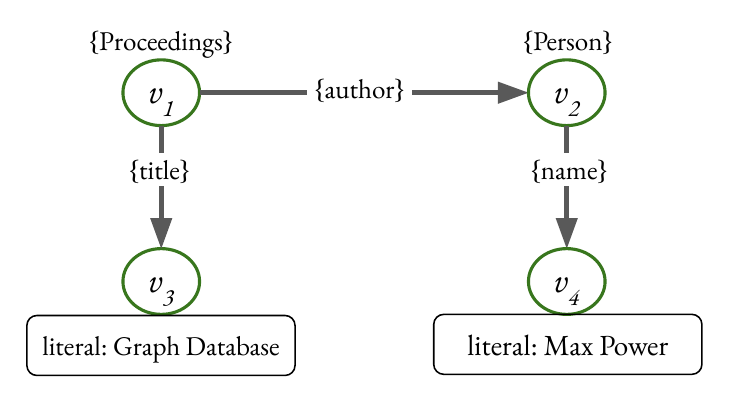}
\end{subfigure}
\caption{\label{fig:lpg-example-introduction} Labeled Property Graph to demonstrate the relation between the label functions of vertices and edges (top) and its visualization as graph (bottom).}
\end{figure}

\subsection{Labeled Property Graphs}
\label{sec:graph-model-lpg}
\label{sec:graph-models-lpg} 

There are a variety of formal definitions of labeled property graphs (LPGs), which emphasize different aspects~\cite{DBLP:conf/www/CiglanNH12,DBLP:journals/csur/AnglesABHRV17}.
In general, LPGs are defined as graphs in which vertices and edges are labeled and can have key--value properties. 
Thus, we define an LPG as $\GLPG = (V,E,\ell_V,\ell_E,\ell_P)$ with vertices $V$ and edges~$E$.
Furthermore, we define finite alphabets 
$\Sigma_V$ (vertex labels), 
$\Sigma_E$ (edge labels),
$\SigmaKey$ (property keys), and $\SigmaValue$ (property values).

We define labeling functions $\ell_V$ and~$\ell_E$ for vertices and edges.
The first function $\ell_V\colon V \rightarrow \mathcal{P}(\Sigma_V)$ maps each vertex to zero or more labels from the finite alphabet~$\Sigma_V$.
The second function $\ell_E\colon E \rightarrow \mathcal{P}(\Sigma_E)$ maps each edge to zero or more labels from the finite alphabet~$\Sigma_E$.
Furthermore, $\ell_P\colon V \cup E \rightarrow \mathcal{P}(\SigmaKey \times \SigmaValue)$ maps each vertex~$v$ and each edge~$e$ to zero or more key--value pairs $\SigmaKey \times \SigmaValue$, \ie the properties.

In our LPG definition, edges are directed, \ie $(v, w) \neq (w, v)$ for all distinct $v, w \in V$.
We can represent the undirected edge~$vw$ by the pair of directed edges $(v,w)$ and $(w,v)$.
In a directed graph, we might have edges $(v,w)$ and $(w,v)$ and, in this case, we might have $\ell_E(v,w)\neq \ell_E(w,v)$.
However, in an undirected graph (where all edges are undirected), we will always have $\ell_E(v,w) = \ell_E(w,v)$.

Vertex labels are often used to represent type information and edge labels to denote specific relationships between vertices.
Information about a specific relationship can be added via key--value properties. 
Key--value properties can also be added to vertices. 
This can be used to state a property of a vertex~$v$ as an alternative to creating a new vertex~$w$ and an edge relating $w$ to~$v$~\cite{DBLP:series/synthesis/2021Hogan}. 
For example, in \cref{fig:lpg-example-introduction}, we represent the same graph from \cref{fig:rdf-example-introduction} as LPG.
We simulate RDF literals by creating vertices $v_3, v_4$ and labeling them with their literal values as key--value properties. 
This way, except for the \texttt{rdf:type} predicates, \cref{fig:rdf-example-introduction} and \cref{fig:lpg-example-introduction} have the same edges. 
Alternatively, one could add the title \enquote{Graph Database} and the name \enquote{Max Power} as a key--value properties to vertex~$v_1$ and~$v_2$, respectively.

Notably, the definition of LPGs deliberately allows for vertices that have no edges, though it is not possible to have an edge without two vertices as its endpoints. 
Vertices can be labeled and can have key--value properties.
Thus, they can add valuable information to a graph even without incoming or outgoing edges. 

LPGs can be used in multi-modal databases, as all kinds of information~-- including complete documents~-- can be attached as key--value properties.
In particular, provenance information can be attached as key--value properties to all vertices and edges.
However, when working with physically and logically distributed graphs, such as the Semantic Web, precise statements about the origin of certain information are desirable. 
To formalize this, we introduce the notion of a Labeled Property Graph Databases, which we refer to as GDB. 

We define a GDB as a multiset of labeled property graphs, which may share vertices and edges.
A multiset of graphs allows data replication, which is important when graph databases are distributed such as in the Semantic Web.
Formally, we define a graph database $GDB = (V, E, \mathcal{G}, \ell_G)$, where $V$~is the set of vertices, $E \subseteq V \times V$ is the set of edges, $\mathcal{G} = \{G_{1}, \ldots, G_{n} \}$ is a multiset of LPGs, and $\ell_G$ is a labeling function for graphs $G \in \mathcal{G}$.
Each labeled property graph $G \in \mathcal{G}$ is a tuple $G = (V_G, E_G, \ell_V, \ell_E, \ell_P)$ with $V_G \subseteq V$ and $E_G \subseteq E$.
We define the labeling function $\ell_G\colon G \rightarrow \Sigma_G$ to map each instance of each graph $G \in \mathcal{G}$ to a single label from the finite alphabet $\Sigma_G$.
This graph label is the name of the graph. 
The sets $V$ and~$E$ are not multisets: all vertices and edges are uniquely identified within the GDB, \eg by IRIs on the Web.
The individual graphs $G \in \mathcal{G}$ are not necessarily connected, \ie they may contain multiple disjoint components.
Two graphs $G_i$ and $G_j$ with $i \neq j$ can have common vertices, \ie it is possible that $V_{G_i}\cap V_{G_j}\neq \emptyset$.
This is a design decision that allows vertices and edges in a graph $G$ to be implicitly labeled with the graph label of $\ell_G(G)$. 
If vertices and edges appear in multiple graphs, they consequently have multiple graph labels. 

In principle, a GDB can be used exactly like an LPG, as presented above. 
In this case, vertex labels denote types, edge labels classify relationships, and key--value properties are used for information only relevant to a single vertex or edge.
When we use a GDB, we label the different graphs contained in the multiset $\mathcal{G}$. 
Following our implied semantics, labeling $G$ is equivalent to adding the property \enquote{graph name} to all vertices and edges contained in this graph $G$. 
Thus, GDBs offer a shortcut to defining  labeling and/or naming in LPGs to ease notation when general statements about a set of vertices and edges are made. 
Note, though, that an LPG with a \enquote{graph name} property is not formally equivalent to an GDB, as it allows inconsistencies that cannot occur in an GDB. 
For example, in an GDB, for every edge $(x,y)$ in, say, $G_1$, the vertices $x$ and~$y$ must also be in~$G_1$. 
However, an LPG could give $x$, $y$ and $(x,y)$ arbitrary values of the graph name property.

\subsection{Transformation between RDF Graphs and LPGs}
\label{sec:transformation}

After having defined RDF graphs and LPGs, we show that they can be transformed into each other.
We give an example of such a transformation, which is adequate for our purposes, while acknowledging that other transformations exist.

Regarding the transformation from RDF to LPG, consider the RDF graph $\GRDF \subseteq (V_I\cup V_{B}) \times P \times (V_I \cup V_{B} \cup L)$. 
We show how to convert it to an equivalent LPG $\GLPG = (V,E,\ell_V,\ell_E,\ell_P)$. The intuition is that RDF types are represented in $\GLPG$ as vertex labels, predicates become edge labels, and literal vertices in $\GRDF$ have their string values attached as  key--value pairs. 
The vertex set~$V$ of $\GLPG$ is the following subset of $\GRDF$'s vertices: all vertices that appear in subject position in triples; along with all vertices that appear in object position, except RDF types. 
This gives $V\subseteq V_I\cup V_B$. The edge set~$E$ of $\GLPG$ is the set of all pairs $(s,o)$ such that $o\in V$ and $\GRDF$ contains a triple $(s,p,o)$ for some~$p$. 
Note that, if a triple $(s,p,o)$ exists, then $s\in V$ by construction, but there may be triples $(s,p,o)$ where $o\in V_C$ and, therefore, $o\notin V$ of the LPG.
It remains to define the labeling functions $\ell_V$, $\ell_E$, and~$\ell_P$. For each vertex $v\in V$, we set $\ell_V(v)$ to be the set of $v$'s RDF types in $\GRDF$. 
For each edge $(s,o)\in E$, we set $\ell_E(s,o) = \{p\in P\mid (s,p,o)\in \GRDF\}$. 
Finally, for each $s\in L$, we set $\ell_P(s)$ to be the set that contains the key--value property $(\texttt{literal}, s)$. 
The reader may verify that $\Sigma_V = V_C$, $\SigmaValue = L$, and $\Sigma_E \cup \SigmaKey = P \setminus \{\texttt{rdf:type}\}$.

The mapping from LPGs to RDF is similar. To translate an LPG into an RDF graph, vertex labels become RDF types, and edge labels become predicates. 
Each vertex with the \texttt{literal} key--value property is transformed to an RDF literal.
Any other key--value property $(k,\mathrm{val})$ of vertex~$s$ becomes a triple $(s,k,v)$, where $v\in L$ is a literal vertex representing~$\mathrm{val}$. 
RDF does not support properties of relationships.
If our LPG contains an edge $(x,y)$ with key--value properties, we add a new vertex~$z$ (following~\cite{DBLP:series/synthesis/2021Hogan}), 
replace the edge $(x,y)$ with $(x,z)$ and $(z,y)$ and add the key--value properties to~$z$.
This is done by adding edges to $z$ with the labels from the keys and literals with the values.
Subsequently, we translate the resulting graph to RDF as described above.

\extended{We have demonstrated that we can transform LPGs and RDF graphs into each other.
LPGs are supported by modern graph databases and are popular in industry and industry related-research~\cite{DBLP:conf/sigmod/TianXZPTSKAP20,orientdb}.
For example, Apache Spark GraphX~\cite{spark:graphx} and Pregel~\cite{DBLP:conf/sigmod/MalewiczABDHLC10} operate on LPGs rather than on RDF graphs.
Since we can transform RDF graphs into LPGs, we can use existing graph processing frameworks for RDF datasets.}

\subsection{Formal Language to Define Structural Graph Summary Models}
\label{sec:graph-summary-model}
\label{sec:gsm}

Structural graph summarization is the task of finding a condensed representation $SG$ (short for ``summary graph'') of an input graph $G$ such that selected characteristics of the original graph are preserved in $SG$~\cite{DBLP:journals/tcs/BlumeRS21,DBLP:journals/vldb/CebiricGKKMTZ19}.
Intuitively, structural graph summarization means that we can conduct specific tasks~-- \eg counting the vertices with a specific type label~-- on the structural graph summary $SG$ instead of~$G$.
The fundamental idea of structural graph summaries is that the task can be completed much faster on the graph summary than on the original graph.

To compute a structural \textbf{graph summary} $SG$ for a given data graph~$G$, we partition the data graph into disjoint sets of vertices. 
We partition the vertices based on equivalence of the subgraphs around them~\cite{DBLP:journals/tcs/BlumeRS21}. 
Equivalence relations describe any graph partitioning in a formal way.
We call the respective subgraphs containing the information necessary to determine the equivalence of two vertices the \textbf{schema structure} of the vertices. 
Which \textbf{features} of the input graph are considered in determining equivalence of schema structures is defined by the \textbf{graph summary model}. 
For different tasks, different features of the summarized vertices are of interest, \eg the number of summarized vertices for cardinality computation or the data source 
for data search.
This information about the summarized vertices is called the \textbf{payload}.
Formally, a structural graph summary model is a $3$-tuple of a data graph $G$, an equivalence relation~$\sim$, and a set of payload elements $\PAY$.

\begin{definition}
\label{def:index}
A \textbf{structural graph summary model} is a tuple $(G, \sim, \PAY)$, where $G$~is the data graph, $\sim \subseteq V \times V$ is an equivalence relation over the vertices in~$G$, and $\PAY$ is a set of payload elements.
Hence, the equivalence classes in~$\sim$ define a partition of the vertices in the graph~$G$.
\end{definition}

A simple example of a graph summary model would be label equality, \ie two vertices are considered equivalent iff they have the same set of labels.
Dependent on the application, one might want to summarize a graph w.r.t. different graph summary models.

\subsubsection{Simple and Complex Schema Elements}

We use our formal language FLUID to define graph summary models as equivalence relations. We summarize FLUID here; for a detailed definition, see~\cite{DBLP:journals/tcs/BlumeRS21}.
The language defines schema elements and parameterizations, which specify different equivalence relations~$\sim$.
One parameterization, the \textit{chaining parameterization}, is of special interest in the context of this work, as it enables summary models such as $k$-bisimulation~\cite{DBLP:journals/vldb/CebiricGKKMTZ19,DBLP:journals/tcs/BlumeRS21,DBLP:conf/sigmod/SchatzleNLP13,DBLP:conf/icde/KaushikSBG02}.
Chaining is described in detail below, while the other five parameterizations are briefly summarized. 
The basic building blocks of summary models in FLUID are \textbf{Simple Schema Elements}~(SSE) and \textbf{Complex Schema Elements}~(CSE).
The three Simple Schema Elements summarize vertices~$v$ based on $\ell_V(v)$, $\ell_E(v,w)$, and/or neighboring vertex identifiers $w$ with $w \in \Gamma(v)$.

\begin{definition}[Simple Schema Elements; from~\cite{DBLP:journals/tcs/BlumeRS21}]
\label{def:sse}
The three simple schema elements are:
\begin{enumerate}
    \item \textbf{Object Cluster} (OC) compares types and vertex identifiers of all neighboring vertices: two vertices $v$ and~$v'$ are equivalent iff $\ell_V(v) = \ell_V(v')$ and $\Gamma(v) = \Gamma(v')$.

\item \textbf{Predicate Cluster} (PC) compares labels: $v$ and~$v'$ are equivalent iff
(i) their \emph{vertex} label sets are both empty or both non-empty and (ii) they have the same labels on their outgoing edges: specifically, $\{\ell(v,w)\mid w\in\Gamma(v)\} = \{\ell_E(v',w')\mid w'\in\Gamma(v')\}$.\footnote{The definition is more intuitive for RDF graphs as noted in~\cite{DBLP:journals/tcs/BlumeRS21}, as only condition~(ii) is needed. When an LPG is translated to an RDF graph, vertex labels are implemented as edges with property \texttt{rdf:type}, so property~(i) becomes ``$v$~has an edge with property \texttt{rdf:type} iff $v'$~does'', and this is already covered by condition~(ii).}

\item \textbf{Predicate--Object Cluster} (POC) combines PC and OC: $v$ and~$v'$ are equivalent iff, $\ell_V(v) = \ell_V(v')$, $\Gamma(v) = \Gamma(v')$, and $\ell_E(v,w) = \ell_E(v',w)$ for all $w\in \Gamma(v)$.
\end{enumerate}
\end{definition}

Observe that SSEs only consider local information about a vertex, \ie its neighbors, the labels of its outgoing edges, or the combination of the two.
FLUID provides a Complex Schema Element (CSE)~\cite{DBLP:journals/tcs/BlumeRS21} to extend this by one step: CSEs allow vertices to be summarized based on their neighbors' neighbors and their neighbors' outgoing edges.
CSEs can be nested to define summaries using vertices at any chosen distance from $v$ and~$v'\!$, and a common pattern of nesting, which generalizes $k$-bisimulation, is implemented by the chaining parameterization (Definition~\ref{def:chaining-parameterization}).

CSEs define a new equivalence relation~$\sim$ by using a tuple of three existing equivalence relations, \ie $CSE := (\sim^{s}$, $\sim^{p}$, $\sim^{o})$.
The equivalence relation~$\sim^{s}$ defines the local schema structure of the vertex~$v$.
$\sim^{o}$~defines the local schema structure of neighbors $w \in \Gamma(v)$.
Intuitively, $\sim^{p}$ defines how the local schema structures of $v$ and~$w$ are connected. 

\begin{definition}[Complex Schema Element; from~\cite{DBLP:journals/tcs/BlumeRS21}]
\label{def:cse}
A Complex Schema Element consists of three equivalence relations and is defined as $CSE := (\sim^s, \sim^p, \sim^o)$.
Two vertices $v, v'$ are considered equivalent, iff 
\begin{gather}
    v \sim^s v', \\
    {
    \forall w\in \Gamma(v): \; \exists w'\in\Gamma(v') \text{ with } \ell_E(v,w) \sim^p \ell_E(v',w') \text{ and } w \sim^o w', \text{ and vice versa.}
    }
\end{gather}
\end{definition}

An example of a CSE, which not only takes local information into account, is given by $(T, id, PC) = (T, id, (T, id, T))$.
It considers vertices as equivalent iff they have the same outgoing edge labels and have neighbors with the same outgoing edge labels.

Introducing the identity relation~$id = \{(v,v) \mid v \in V\}$ and tautology relation~$T = V \times V$, we can represent the three SSEs as CSEs~\cite{DBLP:journals/tcs/BlumeRS21}.
\begin{align*}
    & OC = (T, T, id), \\
    & PC = (T, id, T), \\
    & POC = (T, id, id).
\end{align*}

\subsubsection{Parameterizations} 
One can specialize FLUID's simple and complex schema elements using parameterizations~\cite{DBLP:journals/tcs/BlumeRS21}.
As mentioned before, the chaining parameterization is of special interest, as it enables computing $k$-bisimulations of a graph~$G$ and hence increases the considered neighborhood for determining vertex equivalence~\cite{DBLP:journals/tcs/BlumeRS21}.
The chaining parameterization has a parameter~$k$ that recursively applies CSEs to depth~$k$; the resulting CSE is denoted by $CSE_k$. 

\begin{definition}[Chaining parameterization (from~\cite{DBLP:journals/tcs/BlumeRS21})]
\label{def:chaining-parameterization}
The chaining parameterization $cp(CSE, k)$ takes a complex schema element~$CSE := (\sim^s, \sim^p, \sim^o)$ and a chaining parameter~$k \in \mathbb{N}_{>0}$ and returns an equivalence relation~$CSE_k$ that corresponds to recursively applying CSE to a distance of $k$~hops.
$CSE_k$ is defined inductively as follows:
\begin{align*}
    & CSE_1 = (\sim^s, \sim^p, \sim^o), \\
    & CSE_{k+1} = (\sim^s, \sim^p, CSE_k).
\end{align*}
\end{definition}

The remaining parameterizations are, briefly, as follows.
The label parameterization restricts the edges considered for the summaries to edges with labels defined in a given set~$P_l$ an ignores those with labels not in~$P_l$.
The label parameterization can be used, \eg to consider types in RDF graphs, as they are represented as vertex identifiers and attached to vertices with edges labeled \texttt{rdf:type}.
As vertex types are commonly used, \ie vertex labels in our GDB definition, we write $\OCtype{}$ for OC with the label parameterization $P_l = \lbrace \text{\texttt{rdf:type}} \rbrace$, \ie the \textbf{type cluster}.
Analogously, the \textbf{property cluster} $\PCrel$ is defined as the label parameterized $\PC$ with $P_l = P \setminus \lbrace \text{\texttt{rdf:type}} \rbrace$.
The set parameterization has as parameter a set of labels or vertex identifiers~$S$.
It forces, in addition to the equivalence of vertex and/or edge labels, that all labels are also contained in~$S$.
The direction parameterization allows us to consider only outgoing edges, incoming edges, or both.
The inference parameterization enables ontology reasoning such as RDFS (see \cref{sec:graph-model-rdf}) using a vocabulary graph.
The vocabulary graph stores all hierarchical dependencies between vertex labels (types) and edge labels (properties) denoted by ontologies present in the graph database.
The instance parameterization allows vertices to be merged when they are labeled as equivalent.
The latter parameterization on vertices is commonly known as \texttt{owl:sameAs} inference~\cite{DBLP:books/daglib/0028543,DBLP:journals/tcs/BlumeRS21}.

Note that FLUID's schema elements and parameterizations for inference do not stipulate when inference actually happens.
In the context of graph summarization, the inference is either conducted before summarization or after summarization~\cite{DBLP:journals/tcs/BlumeRS21}.
Generally, these two approaches are equivalent~\cite{DBLP:conf/semweb/LiebigVOM15,DBLP:journals/vldb/GoasdoueGM20,DBLP:conf/aaai/GlimmKT17}, but inference on the graph summary may require multiple iterations over the graph summary~\cite{DBLP:journals/tcs/BlumeRS21}.
In this work, we do not further consider the parameterizations for RDFS and \texttt{owl:sameAs} inference and we assume that inference has been performed before conducting the summarization.
For the interested reader, we discuss practical implications of design choices for inference in the empirical analysis of~\citep{DBLP:journals/tcs/BlumeRS21} and evaluate structural graph summaries using RDF Schema inference in~\cite{DBLP:journals/rpjdi/ScherpB21}.
In \citep{DBLP:journals/vldb/GoasdoueGM20}, the condition to a so-called shortcut to inferencing on the graph summary is discussed and empirically evaluated as well. 

\subsection{Bisimulation}
\label{sec:intro-bisimulation}

Bisimulation originates in labeled transition systems~\cite{Bisimulation:2009}, which can be though of as edge-labeled graphs.
The vertices of the graph are the states of the transition system and the directed edge $(u,v)$ with label~$\ell$ corresponds to a transition of type~$\ell$ from state~$u$ to state~$v$.
Bisimulation defines an equivalence relation on states such that equivalent states have transitions of the same types, to equivalent states.
Forward bisimulation is defined using outgoing edges and backward bisimulation is defined using incoming edges~\cite{DBLP:journals/vldb/CebiricGKKMTZ19,DBLP:journals/tcs/BlumeRS21}.
More formally, we can define forward bisimulation as follows. All vertices are $0$-bisimilar. For $k>1$, $u$ and~$v$ are $k$-bisimilar if, and only if,
\begin{enumerate}
\item they are $(k-1)$-bisimilar, and
\item for every edge $(u,u')$ with label~$\ell$, there is an edge $(v,v')$, also with label~$\ell$, such that $u'$ and~$v'$ are $(k-1)$-bisimilar.
\end{enumerate}
Backward $k$-bisimulation is defined in the same way but considering incoming edges $(u',u)$ and $(v',v)$.
Variants of (backward- or forward-) bisimulation may incorporate vertex labels by requiring that $0$-bisimilar vertices have the same labels.
As such, bisimulation is understood in this work as a form of graph summarization.
In fact, the formal language FLUID introduced in \cref{sec:graph-summary-model} for defining graph summarization models can also be used to define graph summaries based on bisimulation.
We will show this in \cref{sec:iterative-bisimulation}.

Efficient algorithms for bisimulation have been developed by Paige and Tarjan~\cite{DBLP:journals/siamcomp/PaigeT87}, Kaushik \etal~\cite{DBLP:conf/icde/KaushikSBG02} and Sch\"atzle \etal~\cite{DBLP:conf/sigmod/SchatzleNLP13}, and others.
Milo and Suciu's T-index summaries~\cite{DBLP:conf/icdt/MiloS99} and the work of Kaushik \etal are based on backward $k$-bisimulations.
Conversely, the Extended Property Paths of Consens \etal~\cite{DBLP:journals/pvldb/ConsensFKP15}, the SemSets model of Ciglan \etal~\cite{DBLP:conf/www/CiglanNH12} and the work of Sch\"atzle \etal~\cite{DBLP:conf/sigmod/SchatzleNLP13} are based on forward $k$-bisimulation.
Other graph summaries based on bisimulation include~\cite{DBLP:journals/ws/KonrathGSS12,DBLP:journals/vldb/GoasdoueGM20,DBLP:journals/tkde/TranLR13}.

\subsection{Data Structure for Representing Graph Summaries}
\label{sec:data-structure}

Let $GDB=(V,E,\mathcal{G}, \ell_G)$ be a graph database with label functions $\ell_V$, $\ell_E$, and~$\ell_P$ and let $\sim$ be an equivalence relation over~$V$.

\begin{definition}
The \textbf{graph summary} with respect to the model $(GDB, \sim, \PAY)$ is a labeled graph $SG = (\Vvs\cup\Vpe, \Evs\cup\Epe, \ell_V, \ell_E, \ell_P)$, where $\Evs\subseteq\Vvs\times\Vvs$ and $\Epe\subseteq\Vvs\times\Vpe$.
Here, the subscripts \enquote{$\mathrm{vs}$} and \enquote{$\mathrm{pe}$} denote \enquote{vertex summary} and \enquote{payload elements}.
\end{definition}

The subgraph $VS = (\Vvs,\Evs, \ell_V,\ell_E,\ell_P)\subseteq SG$ contains the \emph{schema} information about the GDB according to the model used for~$\sim$, as introduced in \cref{sec:graph-summary-model} (upper half of \cref{fig:complex-rdf-graph-with-payload}).
Thus, the vertices $\Vvs$ and edges $\Evs$ are those shown in the graph summary $SG$ in \cref{fig:schema-evolution}.
The subgraph $PG = (\Vvs\cup\Vpe,\Epe, \ell_V,\ell_E,\ell_P)\subseteq SG$ connects the schema to the payload, \ie each edge $(v,w)\in \Epe$ connects a vertex~$v$ in $VS$ to a vertex~$w$ (the payload element) that contains $v$'s payload information (lower half of  \cref{fig:complex-rdf-graph-with-payload}).
The summary graph $SG$ is the union of the vertex summaries $VS$ and their payload $PG$. 

\begin{figure}[!t]
    \centering
    \includegraphics[width=0.5\linewidth,trim={1cm 8cm 12cm 4cm}, clip=true]
    {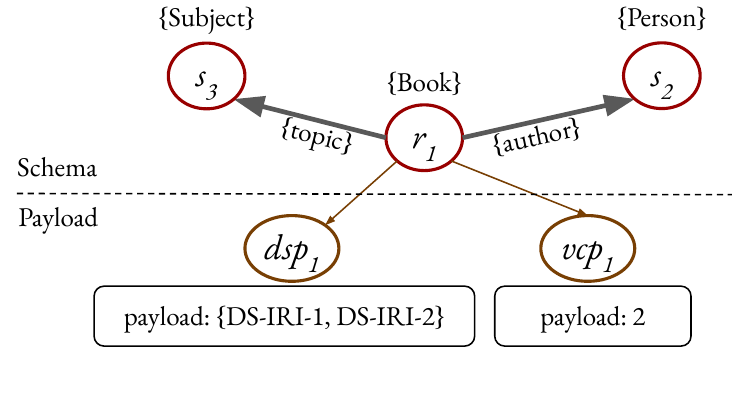}
          %     Example_RDF_Graph_as_SchemEx_with_Payload.pdf
          
\caption{\label{fig:complex-rdf-graph-with-payload} Example structural graph summary with its schema structure of summarized vertices (top) and payload information (bottom). Schema structure is represented using vertex summaries identified by a single primary vertex, \eg $r_1$, and zero or more secondary vertices, \eg $s_2$ and $s_3$.}
\end{figure}

Each vertex in $\Vvs$ has as its identifier a pair $(C,R)$, where $R$~is an equivalence relation over~$V$, the vertices of the GDB being summarized, and $C$~is one of $R$'s equivalence classes.
The edges in~$\Evs$ correspond to the edges in the GDB through which the equivalence relations are defined.
We further divide~$\Vvs$ into \emph{primary vertices}, which are equivalence classes of~$\sim$, and \emph{secondary vertices}, which are equivalence classes of the relations from which $\sim$~is defined.

\paragraph{Example:}
Consider the GDB at time~$t$ in \cref{fig:schema-evolution}. The graph summary $SG$ is defined using the SchemEX graph summary model.
SchemEX summarizes vertices that have the same types (vertex labels) and that have the same set of properties (outgoing edge labels) that are connecting a neighbor that has the same types (vertex labels). 
SchemEX is defined as $CSE(\OCtype,\idrel,\OCtype)$, where $\idrel$ is the identity relation $\id{}$ excluding the \texttt{rdf:type} predicate, and $\OCtype$ (the \enquote{same type} equivalence) has classes $\{v1,v7\}$, $\{v2,v8\}$ and $\{v3,v9\}$. 
Writing $\sim$ for the SchemEX equivalence relation, the identifiers of the vertices in $SG$ are $r1 = (\{v1,v7\}, \sim)$, $s2 = (\{v2,v8\},$ $\OCtype)$ and $s3 = (\{v3,v9\}, \OCtype)$. 
Thus, the vertices $r1$ and~$r2$ are primary, while $s2$ and~$s3$ are secondary.

\paragraph{Inductive Definition of Vertex Summaries.}
For each $v\in V$ (of the GDB) and equivalence relation~$\sim$ (defined by simple or complex schema elements), we define the local \textbf{vertex summary} $vs_{\sim}(v)$ by induction on the structure of the schema elements.
Distinct vertex summaries may share vertices, which compresses the graph summary, since data is reused.

To serve as \emph{base cases for the inductive definition} of vertex summaries, we define equivalence relations $\id = \{(v,v)\mid v\in V\}$ and $T = V\times V$.
For any vertex $v\in V$, $vs_\id(v)$ is the graph with the single vertex~$(\{v\},\id)$, which is the primary vertex, and no edges; similarly, $vs_T(v)$ has a single vertex $(V,T)$ (which is primary) and no edges.
Note that $vs_T(v)$ is identical for every $v\in V$, \ie all vertices are summarized by the same vertex summary, but $vs_\id(v)$ is distinct for every $v\in V$, \ie each vertex summary summarizes one vertex.

For the \emph{inductive step}, we define the vertex summaries for CSEs.
This implicitly includes the simple schema elements OC, PC, and POC. 
They are equivalent to the CSEs $(T,T,\id)$, $(T,\id,T)$ and $(T,\id,\id)$, respectively, but they are implemented separately for efficiency.
Given an equivalence relation~$R$ on a set~$X$ and some $x\in X$, we write $[x]_R$ for the equivalence class of~$R$ that contains~$x$.
Now, let $\sim$ be the equivalence relation defined by the CSE $(\simS,\simP,\simO)$ and let $v\in V$.
Let $\Gamma^o = \{[w]_\simO\mid w\in\Gamma(v)\}$ be the set of $\simO$-equivalence classes of $v$'s neighbors.
The primary vertex of $vs_{\sim}(v)$ is $([v]_\simS,\simS)$.
For each equivalence class $\gamma\in\Gamma^o$, $vs_{\sim}(v)$ has a subgraph
$vs_\simO(w_\gamma)$, where $w_\gamma$~is an arbitrary vertex in~$\gamma$.
Now, let $B^p_\gamma = \{[(v,w)]_\simP\mid (v,w)\in E\text{ and } w\in \gamma\}$; \ie if $v$~has neighbors in $\simO$-class~$\gamma$, then $B^p_\gamma$ is the set of $\simP$-equivalence classes of the edges linking $v$ with a vertex $w\in \gamma$.
For each class $\beta\in B^p_\gamma$, $vs_{\sim}(v)$ contains an edge labeled~$\beta$ from its primary vertex to the primary vertex of $vs_{\simO}(w_\gamma)$. Note that this may introduce parallel edges into the vertex summary, \ie there may be multiple edges between a pair of vertices, albeit with different labels.

\begin{theorem}
\label{theorem:data-complexity}
Let $GDB=(V,E,\mathcal{G}, \ell_G)$ be a graph database with maximum degree at most $d>1$, and let $\sim$~be an equivalence relation on~$V$ defined by nesting CSEs to depth~$k$.
For every $v\in V$, $vs_{\sim}(v)$ is a tree (possibly with parallel edges) with $\mathcal{O}(d^k)$ vertices.
\end{theorem}
\begin{proof}
    That $vs_{\sim}(v)$ is a tree follows from the definition: the base cases are one-vertex trees and the inductive steps cannot create cycles. Any vertex in~$V$ has at most $d$ neighbors, so is adjacent to at most $d$~equivalence classes.  Therefore, no vertex in $vs_{\sim}(v)$ has degree more than~$d$.  $vs_{\sim}(v)$ has depth~$k$, so it contains at most $\sum_{i=0}^k d^i = \mathcal{O}(d^k)$ vertices.
\end{proof}

In \cref{theorem:data-complexity}, we show that the size of a single vertex summary can be bounded by a function of the maximum degree in the input graph and the chaining parameter~$k$. 
Thus, in principle, a single vertex summary in a graph summary may be bigger than the original GDB, but this requires the use of highly nested CSEs on small GDBs, which is unlikely in practice.

\subsection{Summary}
We defined structural graph summaries as equivalence relations over data graphs. 
As data graphs, we work with Labeled Property Graphs (LPGs) and Resource Description Framework (RDF) graphs. 
We described simple transformations between LPGs into RDF graphs. 
Thus, LPGs and RDF can be used interchangeably and without loss of generality. 
Furthermore, we introduced the main concepts of the  graph summary model FLUID, which allows us to define equivalence relations for structural graph summaries.
Finally, we defined summary graphs and the vertex summaries they are built from. These can represent all structural graph summaries defined with FLUID and analyzed its data complexity.
In the next section, we define our algorithm to compute and update FLUID graph summaries. 

\section{Our Graph Summarization Algorithms}
\label{sec:graph-summarization}

Our base algorithm described in~\cref{sec:parallel-algorithm} is designed to allow parallel computation of graph summaries in a distributed system architecture~\cite{DBLP:conf/cikm/BlumeRS20}.
The algorithm is based on the idea of Tarjan's two-phase algorithm for the set union problem~\cite{DBLP:journals/jacm/TarjanL84}.
We implement the make-set phase following a vertex-centric programming model~\cite{DBLP:conf/sigmod/MalewiczABDHLC10,DBLP:journals/semweb/StutzSB16}.
As vertex-centric programming model, we use the message sending and merging of Pregel~\cite{DBLP:conf/sigmod/MalewiczABDHLC10}, which is inherently iterative, synchronous, and deterministic~\cite{DBLP:phd/dnb/Erb20}
The programming model achieves parallelism very similar to MapReduce~\cite{DBLP:journals/cacm/DeanG10}, while it is specifically aiming for graph data~\cite{DBLP:phd/dnb/Erb20}.

In \cref{sec:incremental-algorithm}, we describe the extension of this base algorithm to incrementally compute updates of summaries for evolving graphs.
Our incremental algorithm can automatically detect changes in the data graph, \ie it works without being provided with a change log.
This detection of changes takes time linear in the size of the input graph.
Finally, we present the extension of the base algorithm to a hash-based messaging approach in \cref{sec:iterative-bisimulation}, which allows us to efficiently compute $k$-bisimulations for larger $k$.

\subsection{Parallel Algorithm for Graph Summarization}
\label{sec:parallel-algorithm}

\begin{figure*}[t!]
\centering
\includegraphics[trim={0cm 7cm 0cm 1cm},clip,width=.98\linewidth]{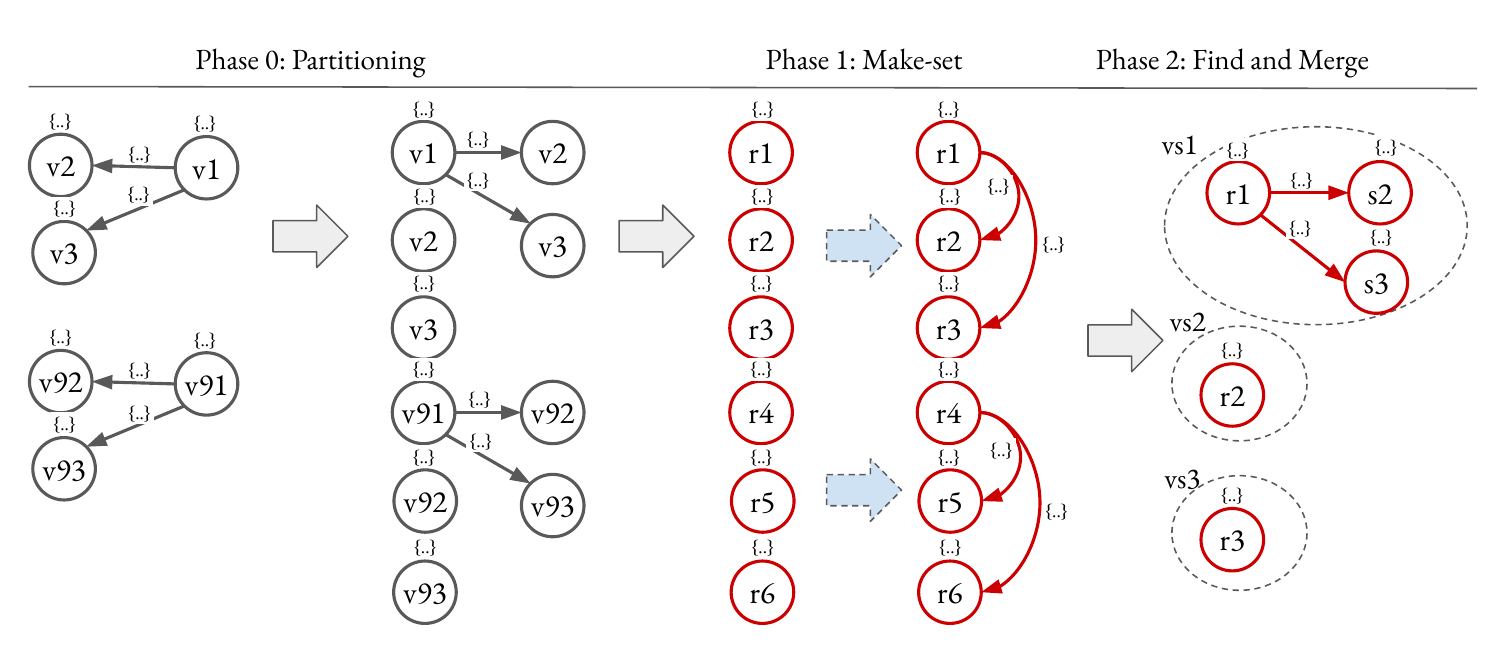}
\caption{\label{fig:find-and-merge} 
Phase~0: Distribute parts of the data graph to the computing nodes. Phase~1: iteratively compute vertex summaries $vs$. Phase~2: merge equivalent vertex summaries, \ie remove redundant information provided by r$4$, r$5$, and r$6$.
Note $\{\ldots\}$ denotes an arbitrary set of labels to the edges and vertices, respectively. 
}
\end{figure*}

Our parallel algorithm for graph summarization is inspired by the two-phase approach of Tarjan's algorithm for the set union problem~\cite{DBLP:journals/jacm/TarjanL84}, as shown in \cref{fig:find-and-merge}. 
To achieve high parallelism, we partition the graph in such a way that all vertices with their label set and references to their outgoing edges including their label are in a single partition.
Thus, the equivalence defined by simple schema elements can be computed for each vertex independently. This partition is computed as Phase~0 of our algorithm.
In Phase 1 (make-set), we compute, for each vertex~$v$ the corresponding vertex summary~$vs$. 
In Phase 2 (find and merge), we find the vertex summaries that have the same schema structure and merge them.
This relates to the find-set operation in Tarjan's algorithm.
When all vertex summaries with the same schema structure are merged, we have partitioned the GDB.

\paragraph{Example Definition of a Graph Summary Model}
To construct a graph summary, we first define an equivalence relation~$\sim$, \ie we define the schema structure we want to capture.
Recall from \cref{sec:data-structure} that 
SchemEX~\cite{DBLP:journals/ws/KonrathGSS12} is defined as the CSE 
${\sim_{\text{SchemEX}}} := (\OCtype{}, \idrel, \OCtype{})$.
To compute the SchemEX graph summary, we need to extract from each vertex~$v$ in the graph database GDB the information needed to define the SchemEX equivalence relation.
This means all vertex labels (types) and all edge labels (properties).
The complex schema element also requires us to compare the vertex labels of all neighbors $w \in \Gamma(v)$.
Thus, we have to exchange information between vertices, sending the label sets of vertices $w$ to all vertices $v$, where there is an edge $(v,w) \in E$.
Finally, we combine this neighbor information with the local vertex information into a vertex summary~$vs(v)$ for each vertex~$v$ in the graph database.
This, completes the first phase as described above.
Note that, when we apply the $k$-chaining parameter, we must repeat this step of exchanging and combining information between vertices up to $k$~times.

\subsubsection{Parallel Algorithm}
We support the parallel computation of all graph summary models defined in \cref{sec:graph-summary-model}.
This is achieved by using a parameterized implementation of the simple and complex schema elements.
The pseudocode of the summarization algorithm \textsc{\summarizeMethod{}} is presented in \cref{algo:summarization}.
In \cref{algo:extraction}, the extraction of the schema for each vertex~$v$ in the graph database begins.
In parallel, for each~$v$, the local vertex schemata are extracted as defined by the simple schema elements of the graph summary model provided as input.
This simple schema extraction is applied using both the $\sim^s$ and $\sim^o$ equivalence relations (see \cref{algo:extract-subject,algo:extract-object}) of the graph summary model.

The locally computed vertex schema is exchanged between vertices to construct the complex schema information, as defined by the graph summary model.
In \cref{algo:signal}, each vertex~$v$ receives the schema (according to the object equivalence relation~$\sim^o$) of all its neighbors.
Likewise, \cref{algo:collect}, collects neighbors' schemata and constructs the data structure defined in \cref{sec:data-structure}.
When we use the $k$-chaining parameterization, this step of sending and aggregating information is done $k$~times (\cref{algo:k-bisimulation,algo:k-bisimulation2,algo:k-bisimulation-end})
and a vertex accesses the schema information from vertices up to distance~$k$. 
\cref{algo:payload} extracts the payload information from the vertex~$v$.
Example payload functions are counting the number of vertices (increasing a counter) or memorizing the source label of~$v$.
The final vertex summary $vs$ and payload vertex~$pe$ are stored in a centralized managed data structure (\cref{algo:merge}), \eg a graph database, where the \textsc{FindAndMerge} phase is implemented.
This means that we compare the vertex summaries and, when two vertices $v$ and $v'$ are found to have the same vertex summary~$vs$, this summary is stored only once.
The payloads of $v$ and~$v'$ are merged, \eg the number of summarized vertices is added or the source lists are concatenated.

The  direction parameterization only changes how the graph is traversed but not the algorithm itself, so it is not shown.
The label and set parameterizations are omitted as they require only a lookup in the corresponding parameter set.
The instance parameterization is a pre-processing step, \ie all vertices connected by an edge with a specific label, \eg \texttt{owl:sameAs}~\cite{w3c-owl}, are merged in \cref{algo:instance-param}.
Following Liebig \etal~\cite{DBLP:conf/semweb/LiebigVOM15}, the inference parameterization is a post-processing step, since inference on a single vertex summary $vs$ is equivalent to inference on all summarized subgraphs.

The \textsc{\summarizeMethod{}} and \textsc{FindAndMerge} functions allow a graph summary $SG$ to be passed as a parameter.
Passing an empty graph summary~$SG$ corresponds to batch computation; for incremental computation, the previous graph summary is passed.

\begin{algorithm}[t]
 \small
 \SetAlgoLined
 \SetFuncSty{textsc}
 \SetKwProg{Fn}{function}{}{end}
 \SetKwProg{ForAllParallel}{forall}{ do in parallel}{end}
 \SetKwInput{Input}{Input}
 \SetKw{Result}{returns}
 \SetKwFunction{summarize}{\summarizeMethod{}}
 \SetKwFunction{extractVertex}{ExtractVertexSchema}
 \SetKwFunction{extractEdge}{ExtractEdgeSchema}
 \SetKwFunction{extractPayload}{ExtractPayload}
 \SetKwFunction{index}{FindAndMerge}
 \SetKwFunction{initialize}{LoadOrCreateSummaryGraph}
 \SetKwData{secondaryIndex}{$SG$} 
 \SetKwFunction{getAll}{GetSummarizedVertices}
 \SetKwFunction{delete}{RemoveLink}
 \SetKwFunction{deleteElement}{RemoveElement}
 \SetKwFunction{getLinks}{GetLinks}
 \SetKwFunction{getLink}{GetLink}
 \SetKwFunction{add}{Add}

\Fn{\summarize{$GDB, SG, (\sim_s,\sim_p,\sim_o)_k$}}{
 \Input{graph database GDB = $(V, E, \mathcal{G}, \ell_G)$}
 \Input{graph summary $SG$ (can be empty)}
 \Input{equivalence relation ${\sim} = (\sim_s,\sim_p,\sim_o)_k$}

  \Result{graph summary $SG$}
  \BlankLine

\If{Instance Parameter $\mathcal{S}$ is used}{\label{algo:instance-param}
 \ForAll{$(v,w) \in E$}{
    \If{$\ell_E(v,w) \in \mathcal{S}$}{
     \textsc{MergeVertices($v,w$)}\;
    }
  }
}

 \ForAllParallel{$v \in V$}{\label{algo:extraction}
   $vs \gets $ \extractVertex{$v, E, \sim_s$}\;\label{algo:extract-subject}
   $tmp \gets $ \extractVertex{$v, E, \sim_o$}\;\label{algo:extract-object}
   \tcc{send information relevant for $\sim_o$ to incoming neighbors}
   \ForAll{$w \in V:\ (w,v) \in E$}{\label{algo:signal}
     $w$.\textsc{Inbox} $\gets$ ($tmp$, $1$)\;
   }
   \tcc{merge information of neighbors and construct complex vertex summaries}
   \ForAll{($tmp\_vs', r) \in v$.\textsc{Inbox}}{\label{algo:collect}
    $t \gets$ \extractEdge{$(v, w), \sim_p$}\;\label{algo:extract-predicate}
    $vs$.\textsc{Neighbor}$_w \gets (t, tmp\_vs')$\;
    \tcc{$k$-chaining repeats send and merge $k$-times}
    \If{$r < k$}{\label{algo:k-bisimulation}
      \ForAll{$w \in V:\ (w,v) \in E$\label{algo:k-bisimulation2}}{
         $w$.\textsc{Inbox} $\gets$ ($vs$, $r+1$)\;\label{algo:k-bisimulation-end}
      }
    }
   }
   $pe \gets$ \extractPayload{$v$}\;\label{algo:payload}
   $SG \gets $ \index{$SG, vs, pe, v$}\;\label{algo:merge}
 }

 \textbf{return} $SG$\;
 }
 \caption{\label{algo:summarization}Parameterized Parallel Graph Summarization}
\end{algorithm}

\subsubsection{Complexity of Parallel Summarization}
\label{sec:complexity-graph-summarization}
We briefly discuss our parallel algorithm's complexity.
Phase~$1$ partitions the set of $n$ vertices using $n$ make-set operations.
Phase~$2$ uses some number $m \leq n$ of find operations.
The worst-case complexity of this computation is proven to be $\mathcal{O}(n + m \cdot \alpha(m + n, n))$, where $\alpha$ is the functional inverse of Ackermann's function~\cite{DBLP:journals/jacm/TarjanL84}.
It is generally accepted that in practice $\alpha \leq 4$ holds true~\cite{DBLP:journals/jacm/TarjanL84}.
The inverse-Ackermann performance of Tarjan's algorithm is asymptotically optimal.
Tarjan's algorithm establishes an essentially linear lower bound on summary computation.

A detailed analysis of the complexity of graph summarization using our formal language FLUID can be found in~\cite{DBLP:journals/tcs/BlumeRS21}.
In summary, the analysis concludes that only the chaining and inference parameterizations affect the worst-case complexity. 
The inference parameterization applied on the graph database leads to a worst-case complexity of $\mathcal{O}(n^2)$.
We implement the inference parameterization as a post-processing step so this has no impact on \cref{algo:summarization}.
We consider the chaining parameterization in detail in \cref{sec:incremental:complexity}.

Typical payload functions extract information from a single vertex, \eg counting or storing the source of a data graph~\cite{DBLP:journals/vldb/CebiricGKKMTZ19}.
These functions run in time $\mathcal{O}(1)$ since only a single vertex is needed to extract the payload.
When we find and merge the vertex summaries, the payload is merged in time $\mathcal{O}(1)$ as well.

\subsection{Incremental Algorithm for Summarization over Evolving Graph}
\label{sec:incremental-algorithm}

For the incremental algorithm, we adapt the find and merge phase of the batch-based parallel algorithm.
In the batch algorithm, all vertex summaries~$vs$ are computed, found, and merged.
In the incremental algorithm, only vertex summaries of vertices with changed information are found and merged.
This avoids unnecessary operations.
However, if no change log is available, for each vertex~$v$ in the graph database the make-set operation needs to be executed, \ie the new vertex summary~$vs$ needs to be extracted.
When a change log is provided, vertex summaries are extracted only for changed vertices.

There are six changes in a graph database that could require updates in a structural graph summary:
a new vertex is observed with a new schema (\addSchema{}), a new vertex is observed with a known schema (\addInstance{}), a known vertex is observed with a changed schema (\modifySchema{}), a known vertex is observed with changed \enquote{payload-relevant information} (\modifyInstance{}), a vertex with its schema and payload information no longer exist (\deleteInstance{}), and no more vertices with a specific schema structure exist (\deleteSchema{}).

\begin{figure}[!b]
\centering
\includegraphics[trim={0.35cm 13.40cm 11.23cm 0cm},clip,width=0.6\linewidth]{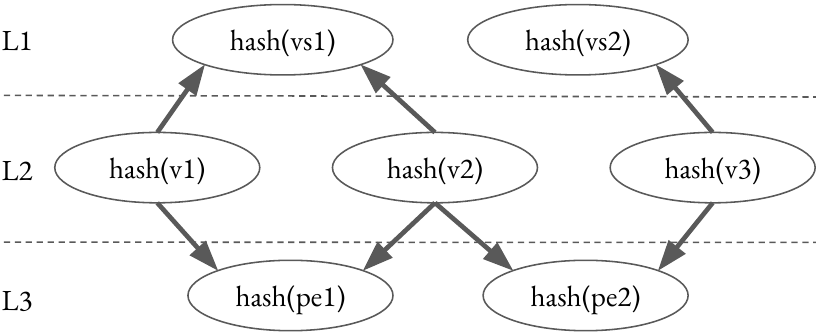}
\caption{The \vertexHashIndex{} is a three-layered data structure. L1 is a unique hash index for the vertex summaries $vs$, L2 is a unique hash index for vertices $v$, and L3 is a unique hash index for payload vertices $pe$.}
\label{fig:secondary-index}
\end{figure}

To check if vertices have changed, we developed the data structure called \vertexHashIndex{} illustrated in Figure~\ref{fig:secondary-index}.
This three-layered index allows us to trace links between vertices~$v$ in the graph database, and vertex summaries $vs$ and payload elements $pe$ in the graph summary.
Intuitively, in the find and merge phase, we look in the \vertexHashIndex{} to see if the vertex summary and payload element for a vertex contain information that requires an update to the graph summary.
If no update is required, we can skip the vertex.
When accessing the \vertexHashIndex{} is faster than actually finding and merging vertex summaries, this decreases the computation time.
We implement the \vertexHashIndex{} as a three-layered unique hash index with cross-links between the layers.
Since only hashes and cross-links are stored, updates on the secondary data structure are faster than updates on the graph summary.
The unique hash indices ensure that there is at most one entry for all referenced vertices in each layer. 
However, maintaining data consistency in the \vertexHashIndex{} is the primary problem that is solved by our algorithm.

\subsubsection{Incremental Algorithm}
The incremental graph summarization algorithm contains two extensions to the parallel graph summarization algorithm, \cref{algo:summarization}.
First, we replace the \textsc{FindAndMerge} function with an incremental version called \textsc{IncrementalFindAndMerge}, which is shown in \cref{algo:incremental}.
This handles additions and modifications. 
Second, we add a loop to handle deletions from the graph database (see \cref{algo:dele-agorithm}).

\begin{algorithm}[t]
 \small
 \SetAlgoLined
 \SetFuncSty{textsc}
 \SetKwProg{Fn}{function}{}{end}
 \SetKwProg{ForAllParallel}{forall}{ do in parallel}{end}
 \SetKwInput{Input}{Input}
 \SetKw{Result}{returns}
 \SetKwFunction{index}{IncrementalFindAndMerge}
 \SetKwFunction{retrieve}{GetElement}
 \SetKwFunction{contains}{ContainsLink}
 \SetKwFunction{containsElement}{ContainsElement}
 \SetKwFunction{delete}{RemoveLink}
 \SetKwFunction{deleteElement}{RemoveElement}
 \SetKwFunction{getLink}{GetLink}
 \SetKwFunction{getLinks}{GetLinks}

 \SetKwFunction{addLink}{AddLink}
 \SetKwFunction{addElement}{AddElement}
 \SetKwFunction{updatePayload}{UpdatePayload}
 \SetKwFunction{hash}{Hash}

 \SetKwData{secondaryIndex}{\texttt{VHI}} 

\Fn{\index{$SG, vs, pe, v$}}{
  \Input{graph summary $SG$}
  \Input{vertex summary $vs$}
  \Input{payload element $pe$}
  \Input{vertex $v$}
 
 \Result{updated graph summary $SG$}
 \BlankLine
 \tcc{\secondaryIndex is the \vertexHashIndex{}}
  \If{\secondaryIndex.\contains{$v$}}{\label{algo:lookup-secondary-index}
    \tcc{Modifications are handled as delete + add operation}
    $\idprev{} \gets$ \secondaryIndex.\getLink{$v$}\;\label{algo:get-secondary-index}
   	$\vsprev{} \gets SG$.\retrieve{$\idprev{}$}\;\label{algo:get-graph-summary}
    	\If{$\vsprev{} \not= vs$}{\label{algo:modifications}
    		\secondaryIndex.\delete{$v$}\; \label{algo:delete-link}
    		\If{$|$\secondaryIndex.\getLinks{$\vsprev{}$}$| \leq 0$}{$SG$.\deleteElement{$\vsprev{}$}\;\label{algo:delete-element}
    		}

    	}
    }

    \If{\textbf{not} \secondaryIndex.\contains{$v$}}{\label{algo:not-known}
        \secondaryIndex.\addLink{$v$, \hash{$vs$}}\;\label{algo:add-link}
    }
	\eIf{$SG$.\containsElement{$vs$}}{ 
	 $SG$.\updatePayload{$vs$, $pe$}\;\label{algo:update-payload}
	}{
	 $SG$.\addElement{$vs, pe$}\;\label{algo:add-element}
	}
	\textbf{return} $SG$\;
 }
 \caption{Incremental FindAndMerge.}
 \label{algo:incremental}
\end{algorithm}

\cref{algo:lookup-secondary-index} of \cref{algo:incremental} checks if the vertex~$v$ is already in the \vertexHashIndex{}. 
If it is, we retrieve its existing vertex summary $\vsprev{}$ (\cref{algo:get-secondary-index,algo:get-graph-summary}). 
If the current vertex summary $vs$ differs from $\vsprev{}$, then $v$'s schema has changed (\modifySchema{}), so we delete the link between $v$ and~$\vsprev{}$ in the \vertexHashIndex{} (\cref{algo:delete-link}). 
If we have just deleted the last link to $\vsprev{}$, then $\vsprev{}$ no longer summarizes any vertex and it is deleted from $SG$ in \cref{algo:delete-element} (\deleteSchema{}).
At this point (\cref{algo:not-known}), there are two reasons that $v$~might not be in the \vertexHashIndex{}: it could be a new vertex that was not in the previous version of the GDB (\addInstance{} or \addSchema{}) or it could be a vertex whose schema has changed since the previous version and we deleted it at \cref{algo:delete-link} (\modifySchema{}). 
In either case, we add a link from $v$ to its new summary~$vs$ at \cref{algo:add-link}.
After the \vertexHashIndex{} is updated, we update the graph summary $SG$.
If $vs$ is already in $SG$, we update the payload element $pe$ of $vs$ in \cref{algo:update-payload} (\addInstance{}).
Thus, we found and merged a vertex summary.
This reflects the case that the payload has changed (\modifyInstance{}), \eg the source graph label has changed.
If $vs$ does not yet exist in $SG$, we add the vertex summary $vs$ to the graph summary $SG$ (\addSchema{}).

After completing phases 1 and~2 of our incremental graph summarization algorithm, we handle deletions (\deleteInstance{}).
The pseudocode for deletions is presented in \cref{algo:dele-agorithm}.
All vertices that are no longer in the GDB are deleted from the \vertexHashIndex{} (\cref{algo:delete}).
Analogously to the deletion described above, deleting entries in the \vertexHashIndex{} can trigger the deletion of a vertex summary $vs$ in the graph summary $SG$ (\deleteSchema{}).
If a vertex summary no longer summarizes vertices in the GDB, it is deleted (\cref{algo:delete-2}).
Note that deleting a vertex summary means deleting at least the primary vertex and its edges. 
The secondary vertices may still be part of other vertex summaries. 

\begin{algorithm}[t]
 \small
 \SetAlgoLined
 \SetFuncSty{textsc}
 \SetKwProg{Fn}{function}{}{end}
 \SetKwProg{ForAllParallel}{forall}{ do in parallel}{end}
 \SetKwInput{Input}{Input}
 \SetKw{Result}{returns}
 \SetKwFunction{summarize}{\summarizeMethod{}}
 \SetKwFunction{extractVertex}{ExtractSimpleVertexSchema}
 \SetKwFunction{extractEdge}{ExtractSimpleEdgeSchema}
 \SetKwFunction{extractPayload}{ExtractPayload}
 \SetKwFunction{index}{FindAndMerge}
 \SetKwFunction{initialize}{LoadOrCreateSummaryGraph}
 \SetKwData{secondaryIndex}{$SG$} 
 \SetKwFunction{getAll}{GetSummarizedVertices}
 \SetKwFunction{delete}{RemoveLink}
 \SetKwFunction{deleteElement}{RemoveElement}
 \SetKwFunction{getLinks}{GetLinks}
 \SetKwFunction{getLink}{GetLink}
 \SetKwFunction{add}{Add}

\tcc{Assumption: all vertices not in the current GDB got deleted}
\tcc{Incremental $SG$ updates for deletions.}
\ForAll{$v \in$ \secondaryIndex.\getAll{} $: v \not\in V$}{\label{algo:delete}
    $vs \gets $ \secondaryIndex.\getLink{$v$}\;
    \secondaryIndex.\delete{$v$}\;
    \If{$|$\secondaryIndex.\getLinks{$vs$}$| \leq 0$ }{
        \tcc{Delete vertex summaries that no longer summarize GDB vertices.}
        $SG$.\deleteElement{$vs$}\;\label{algo:delete-2}

    }
}
 \caption{\label{algo:dele-agorithm}Delete all vertices not in the current version of the GDB.}
\end{algorithm}

\subsubsection{Proof of Correctness of Incremental Algorithm}

\newcommand{\Batch}{\operatorname{Batch}}
\newcommand{\Incremental}{\operatorname{Incr}}
We now prove correctness of our incremental algorithm. Fix a summary model~$\sim$. We use \cref{algo:summarization} with the incremental find and merge implementation of \cref{algo:incremental} and the deletion routine of \cref{algo:dele-agorithm} to define the function $\Incremental(G,SG) = \textsc{\summarizeMethod{}}( G, SG, \sim )$, corresponding to incremental summarization of the graph~$G$ given an existing summary~$SG$. 
We also define the function $\Batch(G) = \textsc{\summarizeMethod{}}(G, \emptyset, \sim)$, corresponding to batch computation of a summary of~$G$ without a pre-existing summary.
Our intended application of the following theorem is that $\GDB_1$ and $\GDB_2$ are snapshots of the same database at different times but the proof is applicable for arbitrary GDBs.

\begin{theorem}
Fix a graph summary model~$\sim$. Let $V$ be a set of vertices and let $\GDB_1$ and $\GDB_2$ be any two GDBs. Then $\Incremental(\GDB_2, \Batch(\GDB_1)) = \Batch(\GDB_2)$.
\label{thm:blah}
\end{theorem}
\begin{proof}
  For $i\in\{1, 2\}$, let $\GDB_i = (V_i, E_i, \mathcal{G}_i, \ell_{G,i})$ and
  let $BSG_i = \Batch(\GDB_i)$. Let $ISG = \Incremental(\GDB_2, BSG_1)$. Note that $BSG_1$ contains summaries for exactly all the vertices of $\GDB_1$, \ie exactly all the vertices in~$V_1$.

  We must show that $ISG=BSG_2$. To do this, we consider separately the vertices in $V_2\setminus V_1$ (new vertices, via \addSchema{} and \addInstance{}), vertices in $V_2\cap V_1$ (unchanged vertices and those modified via \modifySchema{} and \modifyInstance{}) and vertices in $V_1\setminus V_2$ (those deleted via \deleteSchema{} and \deleteInstance{}). We describe how the computations $\Incremental(\GDB_2, BSG_1)$ and $\Batch(\GDB_2)$ process vertices in these three classes.

  Consider a vertex $v\in V_2\setminus V_1$. Since $v\notin V_1$, $v$~is not in the \vertexHashIndex{} when we begin to compute $\Incremental(\GDB_2, BSG_1)$. We add it and its summary~$vs$ to the \vertexHashIndex{} (\cref{algo:add-link}). In the case of \addInstance{}, $vs$~is already contained in $BSG_1$ so we just update its payload in $ISG$ (\cref{algo:update-payload}); for \addSchema{}, $vs$ is not in $BSG_1$ so we add it to $ISG$ (\cref{algo:add-element}). When computing $\Batch(\GDB_2)$, we compute the same summary~$vs$. If $v$~is the first vertex we have seen with this summary, we add $vs$ to $BSG_2$ with the appropriate payload and link $v$ to it in the \vertexHashIndex{}; otherwise, $vs$~is already in $BSG_2$ and we update its payload and link $v$ to it.

  Now, consider a vertex $v\in V_1\cap V2$. When we begin to compute $\Incremental(\GDB_2, BSG_1)$, $v$~is summarized in $BSG_1$ so it is contained in the \vertexHashIndex{}. We retrieve its existing vertex summary $\vsprev{}$ (\cref{algo:get-secondary-index,algo:get-graph-summary}) and compare this to the new vertex summary~$vs$. If the vertex summary has changed (\modifySchema{}), we disconnect $v$ from $\vsprev$ in the \vertexHashIndex{}, delete $\vsprev$ from $ISG$ if it no longer summarizes any vertices and then link $v$ to~$vs$ (Lines \ref{algo:modifications}--\ref{algo:add-link}). 
  When we compute $\Batch(\GDB_2)$, $v$~is processed in the same way as in the previous case: if the summary $vs$ is already present in $BSG_2$, its payload is updated; 
  if $vs$ is not already present, it is added.

  Finally, consider a vertex $v\in V_1\setminus V_2$. Again, $v$ is summarized in $BSG_1$ so it is contained in the \vertexHashIndex{} when \cref{algo:summarization} begins. $v$~is not processed by \cref{algo:summarization}, because $v\notin V_2$, so it remains in the \vertexHashIndex{} when \cref{algo:summarization} completes. However, we then run \cref{algo:dele-agorithm}. This removes the link in the \vertexHashIndex{} from $v$ to its summary $vs$ (\deleteInstance{}) and, if $vs$ no longer summarizes any vertices, we delete $vs$ from $ISG$, too (\deleteSchema{}). When computing $\Batch(\GDB_2)$, $v$~is not processed because it is not in~$V_2$. Therefore, $v$~is not summarized in $BSG_2$.

  Thus, we have constructed summary graphs $ISG$ and $BSG_2$. Every vertex in $V_2$ has the same summary in $ISG$ as it does in $BSG_2$ and every vertex not in~$V_2$ is not summarized in either $ISG$ or $BSG_2$. Therefore, $ISG=BSG_2$, as claimed. 
\end{proof}

The following corollary shows that, if we incrementally compute a sequence of summaries over an evolving graph, the resulting summaries are the ones we would obtain by just running the batch algorithm at each version of the evolving graph. 

\begin{corollary}
Let $\GDB_1, \GDB_2, \dots$ be GDBs. 
Let $SG_1 = \Batch(\GDB_1)$ and, for all $k>1$, let $SG_k = \Incremental(\GDB_k,SG_{k-1})$. 
Then, for all $k\geq 1$, $SG_k = \Batch(\GDB_k)$.
\end{corollary}
\begin{proof}
    The case $k=1$ holds by hypothesis. Suppose that $SG_k = \Batch(\GDB_k)$. Then,
    \begin{align*}
        SG_{k+1} &= \Incremental(\GDB_{k+1},SG_k)\\
                 &= \Incremental(\GDB_{k+1},\Batch(\GDB_k))\\
                 &= \Batch(\GDB_{k+1})\,,
    \end{align*}
    where the first equality is defined in the corollary's statement, the second is by the inductive hypothesis and the third is by \cref{thm:blah}. 
\end{proof}

\subsubsection{Complexity of Incremental Summarization}
\label{sec:incremental:complexity}

In the following, we analyze the update complexity of all possible changes in the data graph w.r.t.\ the number of operations on the vertex summary, \ie adding and/or removing vertices and/or edges.
We first discuss \addSchema{}, \modifySchema{}, and \deleteSchema{} as they require an update to the vertex summary and possible cascading updates to other vertex summaries.
Then, we discuss updates on the \vertexHashIndex{}, which are common to all six changes, and we briefly discuss payload changes.

\paragraph{Graph Summary Updates.} 
Observed vertices with a new schema (\addSchema{}) require at least $1$ and at most $d^{k} + 1$ new vertices to be added to $SG$.
As discussed in \cref{sec:data-structure}, reusing vertices and edges in the graph summary reduces the number of add operations.
However, since it is a new vertex summary, at least one new vertex needs to be added, \ie the primary vertex.
Furthermore, up to $d^{k}$ edges are to be added to $SG$, in the same way.
Deleting all vertices $v$ summarized by a vertex summary $vs$ (\deleteSchema{}) also requires deleting $vs$ from SG.
\deleteSchema{} is the counterpart to \addSchema{}, \ie we have to revert all operations.
Thus, \deleteSchema{} has the same complexity as  \addSchema{}.
When we observe a vertex $v$ with vertex summary $vs'$ at time $t$, but already summarized $v$ with a different vertex summary $vs$ at time $t-1$, we have to modify the graph summary SG.
Transforming a vertex summary $vs$ to another vertex summary $vs'$ means in the worst case deleting all vertices and edges in $vs$ and adding all vertices and edges in $vs'$.
This occurs when the vertex summaries $vs$ and $vs'$ have no vertices in common, so the schema of $v$ has entirely changed from $t-1$ to $t$.
Thus, modifications to $vs'$ are in the worst case $d^{k} + 1$ added vertices, $d^{k}$ added edges, $d^{k} + 1$ deleted vertices, and $d^{k}$ deleted edges.
In the best case, $vs'$ already exists in $SG$ and no updates to the graph summary are needed.

\paragraph{Cascading Updates.}
When complex schema elements are used, updates to the vertex summary $vs$ of a vertex~$v$ can require updates to the vertex summaries of any neighboring vertex~$w$.
For each incoming edge to~$v$, up to $d^-$ vertices need an update.
Complex schema elements correspond to a chaining parameterization of $k=1$.
For arbitrary $k \in \mathbb{N}$, updating one vertex summary~$vs$ requires up to $(d^-)^k$ additional updates.
Therefore, the complexity of \addSchema{}, \deleteSchema{}, and \modifySchema{} is $\mathcal{O}(d^{k})$ for a single vertex update.
Since $k$ is fixed before computing the index, the only variable factor depending on the data is the maximum degree $d$ of the vertices in the GDB.

\paragraph{Updating the \vertexHashIndex{}.}
All six changes require an update to the \vertexHashIndex{}.
A summary model defined using an equivalence relation~$\sim$ partitions vertices of the GDB into equivalence classes $[v]_\sim$, \ie the vertex summaries (see \cref{sec:graph-summary-model}).
For each equivalence class $[v]_\sim$, there is exactly one entry $hash(vs)$ stored in the L1 layer of the \vertexHashIndex{}.
For each vertex $v$ in the GDB, there is a $hash(v)$ stored in L2, which links to exactly one hash in L1.
Thus, \addSchema{}, \deleteSchema{}, and \modifySchema{} require two operations on the \vertexHashIndex{}.
The remaining three changes \addInstance{}, \deleteInstance{}, and \modifyInstance{} require no updates on the vertex summaries $vs$, but require up to two updates on the \vertexHashIndex{}.
Suppose we observe a new vertex~$v$ that is summarized by an existing vertex summary $vs \subseteq SG$ (see \addInstance{}).
No updates on $SG$ are needed (the vertex's schema is already known) and only a single operation on the \vertexHashIndex{} is required to add $v$ to L2 and connect that vertex to L1.
Now, suppose a vertex~$v$ which is summarized by a vertex summary $vs$ is deleted from the GDB but there are other vertices $v'$ in the GDB that are summarized by~$vs$ (\deleteInstance{}).
In this case, no update on $SG$ is required and one operation on the \vertexHashIndex{} is required to
delete $hash(v)$ from L2.
In the case that the vertex $v$ is observed at time $t$ with vertex summary $vs$ and same vertex is already summarized by $vs$ at the previous time $t-1$ (\modifyInstance{}),
no update on $SG$ and no update to the \vertexHashIndex{} is required.

\paragraph{Payload Updates.}
All six changes possibly require an update to the payload.
As discussed above, different payloads are used to implement different tasks. 
Thus, the number of updates depends on what is stored as payload.
In principle, payload updates could be arbitrarily complex; however, payloads that are used in practice can be updated in constant time.
For example, for data search, the source label is stored in payload elements in the graph summary.
Links to these payload elements are stored in L3 of the \vertexHashIndex{}.
In this example, we only update payload elements if a source label changed.
This requires at most two updates to the \vertexHashIndex{}.

\paragraph{Overall Complexity.}
Any change in a GDB with maximum degree~$d$ requires at most $\mathcal{O}(d^{k})$ update operations on the graph summary, when the equivalence relation~$\sim$ is defined using a chaining parameter of~$k$.
Thus, the overall complexity of incrementally computing and updating the graph summary with $\Delta$ changes on the GDB is bounded by $\mathcal{O}(n + \Delta \cdot d^{k})$, where the GDB has $n$~vertices and maximum degree~$d$, and the chaining parameter is~$k$.

\subsection{Hash-based Messaging System for Long $k$-Chaining}
\label{sec:iterative-bisimulation}

The chaining parameterization enables computing $k$-bisimulations of a graph~$G$ and hence increases the considered neighborhood for determining vertex equivalence~\cite{DBLP:journals/tcs/BlumeRS21}.
The chaining parameterization has a parameter~$k$ that recursively applies CSEs to depth~$k$; the resulting CSE is denoted by $CSE_k$. 
As real-world graphs have some heterogeneity, there will be few $k$-bisimilar vertices for larger~$k$ resulting in larger graph summaries~\cite{DBLP:journals/vldb/CebiricGKKMTZ19}.
Our nested data structure of a vertex summary $v$ (see \cref{sec:data-structure}) and how the base algorithm operates on it (see \cref{sec:parallel-algorithm}) makes it flexible and powerful.
However, in particular for long $k$ chains, the base algorithm is expensive in terms of memory consumption. 
We extend our base algorithm towards a scalable algorithm for iterative computation of long $k$ chains. 

\subsubsection{Iterative Parallel Hash-Messaging Algorithm}

The nested data structure of a vertex summary $v$ 
is flexible and powerful. 
A vertex summary of a vertex $v$ contains a map of neighbors, which stores $(key, value)$ pairs, where $key$ is the edge label connecting $v$ to its neighbor $w$, and $value$ is the vertex summary of the neighbor $w$.
In the case of a $k >1$ bisimulation, the neighbors map becomes heavily nested.
For each neighbor of $v$, the map stores the vertex summary of a neighbor $w$.
Furthermore, the stored vertex summary of $w$ contains a map of neighbors of $w$, which potentially contains additional vertex summaries of the neighbors' neighbors of $w$, and so on.
Storing and signaling such nested structures quickly generates high memory requirements for $k>1$,\footnote{\url{https://spark.apache.org/docs/latest/tuning.html}} even for relatively small graphs.

Below, we show how the base algorithm can be extended to a scalable iterative algorithm for $k$-bisimulation by tuning it to avoid such nested data structures.
The key idea is based on Schätzle \etal's use of vertex identifiers derived from hashes instead of maps~\cite{DBLP:conf/sigmod/SchatzleNLP13}.
Every vertex has an identifier $id_{\sim_s}$ and $id_{\sim_o}$.
The algorithm's inputs are the graph database, equivalence relations $\sim_s$, $\sim_p$ and~$\sim_o$ and an integer~$k$, and it computes equivalence with respect to the graph summary model $(\sim_s, \sim_p, \sim_o)_k$. 
When the algorithm terminates, two vertices $v$ and~$v'$ are equivalent iff $v.id_{\sim_s} = v'.id_{\sim_s}$.  

\begin{algorithm}[h]
\small
\SetAlgoLined
\SetNoFillComment
\SetFuncSty{textsc}
\SetKwProg{Fn}{function}{}{end}
\SetKwProg{ForAllParallel}{forall}{ do in parallel}{end}
\SetKwInput{Input}{Input}
\SetKw{Result}{returns}
\SetKwFunction{summarizeAdjusted}{\summarizeAdjustedMethod{}}
\SetKwFunction{extractVertex}{VertexSchema}
\SetKwFunction{extractEdge}{EdgeSchema}
\SetKwFunction{signalMessages}{SignalMessages}
\SetKwFunction{mergeMessages}{MergeMessages}
\SetKwFunction{index}{FindAndMerge}
\SetKwFunction{hash}{Hash}
\caption{Iterative Parallel Hash-Messaging Algorithm for $k$-Bisimulation}
\label{alg:brs-adjusted}
\Fn{\summarizeAdjusted{$GDB, \sim_s,\sim_p,\sim_o,k$}}{
\Input{graph database GDB = $(V, E, \mathcal{G}, \ell_G)$}
\Input{equivalence relations $\sim_s$, $\sim_p$, $\sim_o$}
\Input{integer $k$}

\Result{graph summary $SG$}
\BlankLine
\tcc{Initialization}
\ForAllParallel{$v \in V$}{\label{algo:initialization}
$vs_{\sim_s} \gets $ \extractVertex{$v, E, \sim_s$}\; \label{algo:extract-subject-adjusted}
$vs_{\sim_o} \gets $ \extractVertex{$v, E, \sim_o$}\; \label{algo:extract-object-adjusted}
$v.id_{\sim_s} \gets vs_{\sim_s}.ID$\; \label{algo:initialize-id-subject}
$v.id_{\sim_o} \gets vs_{\sim_o}.ID$\; \label{algo:initialize-id-object}
\label{algo:intialization-end}}
\vspace{2ex}
\If{k > 1}{\label{algo:signal-all}
    \tcc{Signal initial messages. Update $v.id_{\sim_s}$ and $v.id_{\sim_o}$}
    \ForAllParallel{$v \in V$}{\label{algo:initial-signal}
        \signalMessages{$t_w, v.id_{\sim_s}, v.id_{\sim_o}$}\; \label{algo:signal-t-s-o-initial}
        $arr_{id_{\sim_s}} \gets $ \mergeMessages{$(t_w, id_{\sim_s})$ }\;
        \label{algo:merge-t-s-initial}
        $arr_{id_{\sim_o}} \gets $ \mergeMessages{$(t_w, id_{\sim_o})$ }\; \label{algo:merge-t-o}
        $v.id_{\sim_s} \gets $ \hash{$arr_{id_{\sim_s}}$}\; \label{algo:update-id-s-initial}
        $v.id_{\sim_o} \gets $ \hash{$arr_{id_{\sim_o}}$}\; \label{algo:update-id-o-intitial}
    }\label{algo:initial-signal-end}
    
    \vspace{2ex}
    
    \tcc{Signal $k-2$ times. Do not include $t_w$ when updating $v.id_{\sim_o}$.}
    \For{$i = 2$ to $k - 1$}{\label{algo:signal2}
        \ForAllParallel{$v \in V$}{\label{algo:k-2_signals}
            \signalMessages{$t_w, v.id_{\sim_s}, v.id_{\sim_o}$}\; \label{algo:signal-t-s-o}
            $arr_{id_{\sim_s}} \gets $ \mergeMessages{$(t_w, id_{\sim_s})$}\; \label{algo:merge-t-s}
            $arr_{id_{\sim_o}} \gets $ \mergeMessages{$id_{\sim_o}$}\; \label{algo:merge-o}
            $v.id_{\sim_s} \gets $ \hash{$arr_{id_{\sim_s}}$}\; \label{algo:update-id-s}
            $v.id_{\sim_o} \gets $ \hash{$arr_{id_{\sim_o}}$}\; \label{algo:update-id-o}
        }
    }\label{algo:signal-end}
    
    \vspace{2ex}
    
    \tcc{Signal final messages. Update $v.id_{\sim_s}$}
    \ForAllParallel{$v \in V$}{\label{algo:final-signal}
        \signalMessages{$v.id_{\sim_o}$}\; \label{algo:signal-o}
        $arr_{id_{\sim_o}} \gets $ \mergeMessages{$id_{\sim_o}$}\; \label{algo:merge-o-final}
            $v.id_{\sim_s} \gets $ \hash{$arr_{id_{\sim_o}}$}\; \label{algo:update-id-s-final}
    }\label{algo:final-signal-end}
}\label{algo:signal-all-end}
\tcc{If $k = 1$: Only signal $(t_w, v.id_{\sim_o})$ and return}
\Else{\label{algo:signal-only-objId}
    \ForAllParallel{$v \in V$}{\label{algo:signal-k-is-1}
        \signalMessages{$t_w, v.id_{\sim_o}$}\; \label{algo:only-signal-t-o}
        $arr_{id_{\sim_o}} \gets $ \mergeMessages{$(t_w, id_{\sim_o})$}\; \label{algo:merge-only-t-o}
        $v.id_{\sim_s} \gets $ \hash{$arr_{id_{\sim_o}}$}\; \label{algo:update-id-only-t-o}
    }
}\label{algo:signal-only-objId-end}
        \vspace{2ex}
$SG \gets $ \index{$SG, V$}\;\label{algo:adjusted-merge}
\Return $SG$\;
}
\end{algorithm}

\Cref{alg:brs-adjusted} shows the pseudocode for the hash-messaging algorithm to efficiently compute $k$-bisimulation.
In contrast to the base algorithm, it operates on vertex summaries only in the initialization step (Lines~\ref{algo:initialization} to \ref{algo:intialization-end}), which computes the local information of a vertex w.r.t. $\sim_s$ (\cref{algo:extract-subject-adjusted}) and $\sim_o$ (\cref{algo:extract-object-adjusted}) of the respective graph summary model.
After computing the local information, an $ID$ function is used to derive a numerical value from the initial vertex summaries, resulting in identifier values $id_{\sim_s}$ (\cref{algo:initialize-id-subject}) and $id_{\sim_o}$ (\cref{algo:initialize-id-object}).
The $ID$ function calculates a numerical value based on the vertex summaries' content, such that if two vertex summaries~$vs_1$ and~$vs_2$ have the same content, the $ID$ function outputs the same values.
Afterwards, if $k=1$ (\cref{algo:signal-only-objId} to \cref{algo:signal-only-objId-end}), every vertex~$v$ sends, to each in-neighbor~$w$, its $id_{\sim_o}$ value along with the simple edge schema~$t_w$, corresponding to the labels of the edge $(w,v)$ (\cref{algo:only-signal-t-o}).
Subsequently, each vertex merges the messages it receives into an array containing tuples with the received information $(t_w, id_{\sim_o})$ (\cref{algo:merge-only-t-o}).
This is handled by the function \textsc{MergeMessages}, which uses a hash table to eliminate any duplicates. 
Finally, the vertex's $id_{\sim_s}$ value is updated by hashing the array (\cref{algo:update-id-only-t-o}).
\textsc{Hash} uses an order independent hash function, \ie it computes the same value for the arrays $\text{[1,2]}$ and~$\text{[2,1]}$.

Moreover, whenever the algorithm updates the $v.id_{\sim_s}$ value for a vertex $v$, it first adds the old $v.id_{\sim_s}$ value of $v$ so that this information is included in the new value.
If $k > 1$, the algorithm operates differently (\cref{algo:signal-all} to \cref{algo:signal-all-end}).
In the first iteration (lines \ref{algo:initial-signal}-- \ref{algo:initial-signal-end}) every vertex $v$ sends its $id_{\sim_s}$ and $id_{\sim_o}$ values and the simple edge schema $t_w$ to each in-neighbor~$w$ (\cref{algo:signal-t-s-o-initial}).
Subsequently, each vertex merges the incoming messages into (1) an array containing tuples with the received information $(t_w, id_{\sim_s})$ (\cref{algo:merge-t-s-initial}) and (2) an array containing tuples with the received information $(t_w, id_{\sim_o})$ (\cref{algo:merge-t-o}).
Finally, the $id_{\sim_s}$ value (\cref{algo:update-id-s-initial}) and the $id_{\sim_o}$ value (\cref{algo:update-id-o-intitial}) of~$v$ are updated by hashing the corresponding arrays.
In the next block (lines \ref{algo:signal2}--\ref{algo:signal-end}), the algorithm performs the same steps but excludes $t_w$ when merging messages for $id_{\sim_o}$ (\cref{algo:merge-o}).
Including $t_w$ is not necessary, as the $id_\sim{o}$ values already depend on the corresponding simple edge schema, which was added in the first iteration.
As a result, in the final iteration (lines \ref{algo:final-signal}--\ref{algo:final-signal-end}), the only missing information to determine equivalence between two vertices $v$ and~$v'$ is stored in their neighbors' $id_{\sim_o}$ values.
Hence, each vertex first signals this information to its neighbors (\cref{algo:signal-o}), then merges the incoming information (\cref{algo:merge-o-final}) and compute their final $id_{\sim_s}$ value (\cref{algo:update-id-s-final}).
At the end of execution, vertices with the same $id_{\sim_s}$ value are merged together (\cref{algo:adjusted-merge}), which are resembling the equivalence classes.

\paragraph{Example}
We illustrate how the algorithm operates using the forward $k$-bisimulation of Schätzle  \etal{}~\cite{DBLP:conf/sigmod/SchatzleNLP13} applied on the graph in Figure~\ref{fig:lpg-example-introduction} for $k=2$.
The forward bisimulation summary model of Schätzle~\etal{}\@~\cite{DBLP:conf/sigmod/SchatzleNLP13} incorporates edge labels into the definition of $k$-bisimulation.
Using the chaining parameterization (\cref{sec:graph-summary-model}), the edge-labeled forward $k$-bisimulation of Schätzle  \etal{}~\cite{DBLP:conf/sigmod/SchatzleNLP13} is given by
\begin{align}
\label{gsm:schaetzle}
    \begin{aligned}
        CSE_{\text{Schätzle}} = cp((T, id, T), k) = (T, id, T)_k\,.
     \end{aligned}
\end{align}

For $k = 2$, this is $(T, id, (T, id, T))$.
Both the subject relation $\sim_s$ and the object relation $\sim_o$ are defined to be the tautology.
Hence, in the initialization step, every vertex is assigned the same value for $id_{\sim_s}$ and $id_{\sim_o}$: for all $i \in \{1,2,3,4\}$, 
\begin{equation*}
   v_i.id_{\sim_s}^0 = 
   v_i.id_{\sim_o}^0 = 0\,.
\end{equation*}
Here, $v_i.id_{\sim_s}^0$ denotes the value for vertex $v_i$ at iteration~$0$.
Next, in iteration~$1$, vertices $v_1$ and~$v_2$ receive the following messages:
\begin{align*}
    v_1 \, &: \, (\text{author}, v_2.id_{\sim_s}^0, v_2.id_{\sim_o}^0), \, (\text{title}, v_3.id_{\sim_s}^0, v_3.id_{\sim_o}^0) \\
    v_2 \, &: \, (\text{name}, v_4.id_{\sim_s}^0, v_4.id_{\sim_o}^0)\,. 
\end{align*}
Vertices $v_3$ and~$v_4$ do not receive any messages during the algorithm's execution, as they have no outgoing edges.
Accordingly, the vertices' $id$ values are updated:
\begin{align*}
    v_1.id_{\sim_s}^1 &= HASH\big([(\text{self}, v_1.id_{\sim_s}^0), \, (\text{author}, v_2.id_{\sim_s}^0), \, (\text{title}, v_3.id_{\sim_s}^0)] \big) \\
    v_1.id_{\sim_o}^1 &= HASH\big([(\text{author}, v_2.id_{\sim_o}^0), \, (\text{title}, v_3.id_{\sim_o}^0)] \big) \\
    v_2.id_{\sim_s}^1 &= HASH\big([(\text{self}, v_2.id_{\sim_s}^0), \,  (\text{name}, v_4.id_{\sim_s}^0)] \big) \\
    v_2.id_{\sim_o}^1 &= HASH\big([(\text{name}, v_4.id_{\sim_o}^0)] \big) \\
    v_3.id_{\sim_s}^1 &=  v_3.id_{\sim_s}^0 \\
    v_3.id_{\sim_o}^1 &=  v_3.id_{\sim_o}^0 \\
    v_4.id_{\sim_s}^1 &=  v_3.id_{\sim_s}^0 \\
    v_4.id_{\sim_o}^1 &=  v_4.id_{\sim_o}^0\,.
\end{align*}
To distinguish the previous value from the received messages, it is added as a tuple of the form $(\text{self}, \text{old-Id-Value})$, where \enquote{\emph{self}} denotes a unique id.
In the final iteration ($k=2$), $v_1$ and~$v_2$ receive the following messages:
\begin{align*}
    v_1 \, &: \, v_2.id_{\sim_o}^1, \, v_3.id_{\sim_o}^1 \\
    v_2 \, &: \, v_4.id_{\sim_o}^1\,.
\end{align*}
Finally, the $id_{\sim_s}$ values are updated to
\begin{align*}
    v_1.id_{\sim_s}^2 &= HASH\big([(\text{self}, v_1.id_{\sim_s}^1), \, (\text{final}, HASH([ v_2.id_{\sim_o}^1, v_3.id_{\sim_o}^1 ])] \big) \\
    v_2.id_{\sim_s}^2 &= HASH\big([(\text{self}, v_2.id_{\sim_s}^1), \, (\text{final}, HASH([ v_4.id_{\sim_0}^1 ])] \big)\,.
\end{align*}
Here, the final messages are first hashed and then put into a tuple $(\text{final}, \text{hashOfMessages})$, so it can be distinguished from the previous $id_{\sim_s}$ when computing the final value for $id_{\sim_s}$.
The resulting equivalence classes are $\{ v_1 \}$, $\{ v_2 \}$ and $\{ v_3, v_4 \}$.

\subsubsection{Complexity Analysis}

We now analyze the running time of our parallel hash-based algorithm for $k$-bisimulation (Algorithm~\ref{alg:brs-adjusted}).
We are given equivalence relations $\sim_s$, $\sim_o$ and~$\sim_p$ that have already been computed.

Initialization (lines~\ref{algo:initialization}--\ref{algo:intialization-end}) requires computing and hashing the schema of each vertex with respect to $\sim_s$ and~$\sim_o$. This takes time $O(|V|) = O(|E|)$.
This is followed by $k$~phases of computation.
For $k=1$, the single phase is the loop at lines \ref{algo:signal-k-is-1}--\ref{algo:update-id-only-t-o}.
For $k>1$, the phases are the loop at lines \ref{algo:initial-signal}--\ref{algo:update-id-o-intitial}, the loop at lines \ref{algo:k-2_signals}--\ref{algo:update-id-o} ($k-2$ times), and the loop at lines \ref{algo:final-signal}--\ref{algo:update-id-s-final}.
The phases differ slightly in the details but each one has three stages, which operate on every vertex in the graph. Each vertex:
\begin{enumerate}
\item signals its identifier w.r.t.\@ $\sim_s$ and~$\sim_o$ to each of its in-neighbors;
\item receives corresponding identifiers from each of its out-neighbors and merges the received information;
\item hashes the result.
\end{enumerate}
Each stage can be seen to run in time $O(|E|)$ on a graph with edge set~$E$.
In stage~(1), each vertex~$v$ sends one message to each of its in-neighbors.
The messages are of fixed size, independent of the graph, so the total time taken is $O(\sum_v d^-(v)) = O(|E|)$. 
For stage~(2), we consider $\sim_o$-identifiers. Some phases also use $\sim_s$-identifiers, for which the argument is identical.
Each vertex~$v$ receives a $\sim_o$-identifier from each of its out-neighbors.
It collates the received identifiers into an array, which takes time $O(d^+(v))$. It then removes duplicates from the array, which is done in time $O(d^+(v))$ using a hash table to detect duplicates. Thus, the total time for this stage is $O(\sum_v d^+(v)) = O(|E|).$
In stage~(3), each vertex~$v$ must compute hashes of one or two sets of size at most $d^+(v)$.
Computing the hash takes time linear in its size, so the total time taken for the stage is $O(\sum_v d^+(v)) = O(|E|)$.

Thus, each stage takes time $O(|E|)$ and, hence, each phase takes time $O(|E|)$.
There are $k$~phases, so the total running time is $O(k\,|E|)$.
For fixed~$k$, this is linear in~$|E|$.

We note that the same arguments about the number of messages passed applies to the base algorithm (Algorithm~\ref{algo:summarization}).
However, in that case, we do not obtain a $O(|E|)$ running time as the size of the messages is not constant.
In the case of Algorithm~\ref{algo:summarization}, the messages at the final iteration of the $k$-chaining computation contain the structure of all vertices within distance $k-1$ of the vertex that originated the message.
In a graph of maximum degree~$d$, each of these messages is of size $O(d^{k-1})$.
Thus, a vertex of degree~$d$ requires time $O(d^k)$ just to read its incoming messages.

\subsection{Summary}
Our parallel graph summarization algorithm can compute structural graph summaries for static graphs based graph summary models defined in a formal language. 
We extended the parallel base algorithm to an incremental algorithm for updating structural graph summaries when the data graphs evolves over time.
This is achieved by our proposed \vertexHashIndex{}. 
Incremental graph summaries can be updated in time $\mathcal{O}(\Delta \cdot d^k)$, where $\Delta$ is the number of additions, deletions, and modifications to the input graph, $d$ is its maximum degree, and $k$ is the maximum distance in the subgraphs considered. 
Finally, we have introduced a second extension of the base algorithm that can compute graph summaries based on a $k$-bisimulation for large values of $k$ such as $k=10$.
This extension is based on the idea of vertex identifiers from Schätzle \etal~\cite{DBLP:conf/sigmod/SchatzleNLP13}.

\section{Experimental Apparatus}
\label{sec:incremental-apparatus}
We empirically evaluate the time and space requirements of the graph summarization algorithms presented in \cref{sec:graph-summarization} for different graph summary models.
We run experiments on graph databases that evolve over time.
We run three sets of experiments.
First, we analyze the costs of computing a new graph summary from scratch (batch computation) compared to incrementally updating an existing graph summary.
Second, we evaluate the impact of the parallelization on the overall performance. 
Third, we evaluate our algorithm's memory consumption.
Below, we describe the datasets and summary models that we used, and our test system.
The measures used are described in each experiment.

\subsection{Datasets}
\label{sec:datasets}
We generated two synthetic benchmark datasets (LUBM100 and BSBM) and used two variants of the real-world DyLDO dataset weekly crawled from the web.
The two benchmark datasets are suitable for the cardinality computation and semantic entity retrieval tasks.
The web datasets are also used for the data search task.

\paragraph{LUBM100:} 
The Lehigh University Benchmark (LUBM) generates benchmark datasets containing people working at universities~\cite{DBLP:journals/ws/GuoPH05}.
We use the Data Generator v1.7 to generate $10$ versions of a graph containing $100$ universities~\cite{blume_till_2020_5714435_lubm}.
Thus, all versions are of similar size, but we emulate modifications by generating different vertex identifiers, \ie each version is considered as timestamped graph.
Each graph contains about $2.1\,$M vertices and $11\,$M edges.
Over all versions, the mean degree is $5.1$ ($\pm 0.1$).

\paragraph{BSBM:}
The Berlin SPARQL Benchmark (BSBM) is a suite of benchmarks built around an e-commerce use case~\cite{DBLP:journals/ijswis/BizerS09}.
We generate $21$ versions of the dataset with different scale factors~\cite{blume_till_2020_5714035_bsbm}.
The first dataset, with a scale factor of $100$, contains about $7,000$ vertices and $75,000$ edges.
We generate versions with scale factors between $2,000$ and $40,000$ in steps of $2,000$.
The largest dataset contains about $1.3\,$M vertices and $13\,$M edges.
For our experiments, we first use the different versions ordered from smallest to largest (version $0$ to~$20$) to simulate a growing graph database.
Subsequently, we reverse the order to emulate a shrinking graph database.
Over all versions, the mean degree is $10.9$ ($\pm 0.2$). 

\paragraph{DyLDO-core and DyLDO-ext:}
The Dynamic Linked Data Observatory (DyLDO) provides regular crawls of the Web of Data~\cite{DBLP:conf/esws/KaferAUOH13}.
The crawls start from about $95,000$ representative seed URIs (source labels of graphs).
There are two variants of this dataset:
The dataset containing only the graphs identified by the seed URIs is referred to as \textit{DyLDO-core}.
It contains the $95,000$ different graphs obtained from the seed URIs, stored in the multi-set $\mathcal{G}$ of the GDB.
The extended crawl (including the core) is referred to as \textit{DyLDO-ext}.
Starting from the seed URIs, a breadth-first search is conducted with a crawling depth of $2$, \ie graphs are added that are referenced from already crawled graphs.

For DyLDO-core, we use all 50 crawls conducted between January 20, 2019 and January 12, 2020.
DyLDO-core contains $2.1$--$3.5\,$M vertices and $7$--$13\,$M edges.
Over all weekly crawls, the mean degree is $4.8$ ($\pm 0.5$).
Note that week 21 (June 16, 2019) is an anomaly as DyLDO-core contains only eight edges due to a crawling failure: weeks 21 and~22 are excluded from the results.
For DyLDO-ext, we use the first five crawls, which contain $7$--$10\,$M vertices and $84$--$106\,$M edges.
The mean degree is $10.8$ ($\pm 0.7$).

\subsection{Selected Graph Summary Models}
There is a large variety of structural graph summary models~\cite{DBLP:journals/csur/LiuSDK18,DBLP:journals/corr/abs-2004-14794,DBLP:journals/pvldb/KhanBB17,DBLP:journals/vldb/CebiricGKKMTZ19,DBLP:journals/tcs/BlumeRS21}.
For our experiments 1 to 3, we chose three representative summary models based on the literature, namely Class Collection, Attribute Collection, and SchemEX, that use fundamental graph summary features of combining types and properties. 
Class Collection~\cite{DBLP:conf/dexaw/CampinasPCDT12} uses only the vertices' type sets to compute the schema, \ie it summarizes vertices that share the same label.
Class Collection is defined using the type cluster $\OCtype{}$.
Attribute Collection~\cite{DBLP:conf/dexaw/CampinasPCDT12} uses only property sets to compute the schema structure of vertices, \ie it summarizes vertices that have the same labels of outgoing edges.
It is defined as property cluster $\PCrel$.
SchemEX combines Class Collection and Attribute Collection: two vertices must have the same label, have edges with the same label that lead to neighbors with the same label.
SchemEX is defined as CSE $(\OCtype{},\idrel,\OCtype{})$ (see \cref{sec:data-structure}).
Finally, for our $k=1,\ldots,10$ bisimulation experiments, we employ the graph summary model by Schätzle \etal~\cite{DBLP:conf/sigmod/SchatzleNLP13}.
This is expressed by the $CSE_{\text{Schätzle}} = (T, id, T)_k$ (see \cref{sec:iterative-bisimulation}).

\subsection{Implementation Details}
\label{sec:implemantation}

We implement our algorithms in 
Scala 2.12.14 and Java 11.0.11 using Apache Spark GraphX 3.1.1~\cite{DBLP:conf/osdi/GonzalezXDCFS14,spark:graphx}. 
We use a single Spark context on a dedicated server\extended{~using up to 20~cores and 200~GB heap space}. 
The datasets are read in parallel as RDF graphs in gzipped $n$-triple files and mapped to GraphX's internal labeled property graph data structure (see \cref{sec:transformation}). 

Our test system for experiments 1 to 3 is a server with Ubuntu 20 OS equipped with $2 \times$ Intel Xeon CPU E$5$-$2690$ v$2$ at $3.00\,$GHz with $10$~cores and $20$~threads each.
Each CPU has access to $12 \times$ $16\,$GB DDR3 memory at $1600\,$MHz, resulting in a total of $384\,$GB main memory.
We run OrientDB in a Docker container with at most $20\,$GB heap space and $100\,$GB memory mapped files. 
All execution times are measured using the built-in logging feature of Apache Spark. 
For experiment 4, we use a machine with Ubuntu 20 OS with 2x AMD EPYC ROME 7F32 processors with each 8 cores and 16 threads and 2 TB RAM (32x 64GB DDR4 with 3200MHz).
For execution, all cores and 1.94 TB Heap Space were given to the Apache Spark Context.
Instrumentation was done using the Apache Spark Monitoring API.

\section{Experiment 1: Comparison of Batch-based vs. Incremental Graph Summarization}
\label{sec:experiment-1}
In this experiment, we directly compare the performance of our batch and incremental algorithms.
Both algorithms are provided $20$ cores. 

\subsection{Procedure}
\label{sec:metrics}
We compute the structural graph summary for each version of the GDB.
The batch algorithm recomputes the graph summary from scratch.
The incremental algorithm detects the changes in the GDB and only performs updates on the graph summary.

We employ different metrics to evaluate the size of the graph summaries, the update complexity for the graph summaries, and the runtime performance of the graph summarization algorithms.
Regarding size, we first count the number of vertices $|V|$ and edges $|E|$ in the GDB and the number of vertices $|\Vvs|$ and edges $|\Evs|$ in the graph summaries.
We write $\frac{|V|}{|V / \sim|}$ for the \emph{summarization ratio}, \ie the fraction of the number of all vertices $V$ in the GDB and the number of different equivalence classes under the equivalence relation $\sim$ in the graph summary $SG$.
The summarization ratio describes how many vertices $v$ are on average summarized by one vertex summary $vs$.
Higher summarization ratios indicate a low variety of schema structures.
A ratio of $1$ indicates that no two vertices share the same schema structure.

Regarding the update complexity, we count the number of vertices with a changed schema (\addSchema{}, \deleteSchema{}, and \modifySchema{}), \ie the relevant \emph{changes} in the GDB.
Since changed schemata require between $0$ and $\mathcal{O}(d^k)$ updates to the graph summary, we also count the number of updates on the graph summary.
Finally, the runtime performance of the algorithms is measured by the time needed to compute a summary from scratch or, in the incremental case, to update the summary $SG$ for each version of the GDB.

\begin{figure*}[t!]
{
\begin{subfigure}[t]{0.31\linewidth}
    \centering
    \includegraphics[trim={0cm 0cm 0cm 0cm},clip,width=\textwidth]{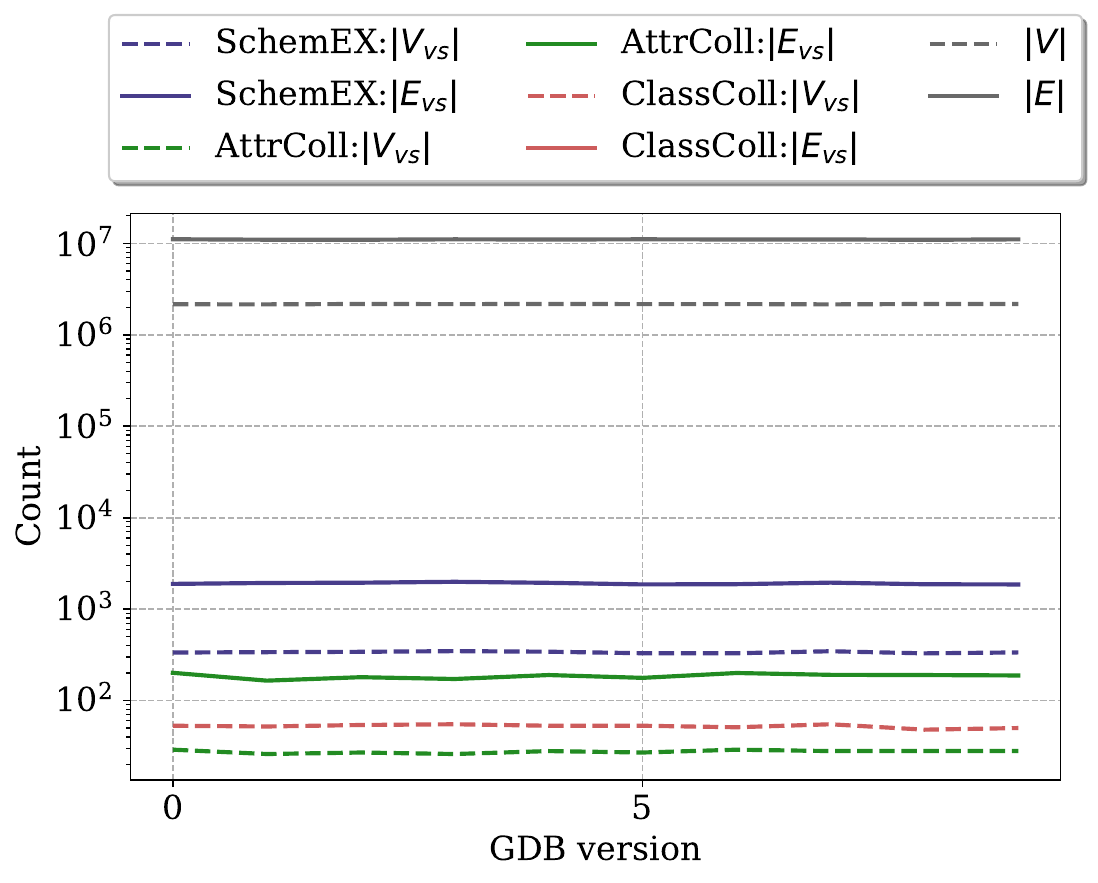}
    \caption{\label{fig:LUBM-size} LUBM100 dataset size.}
\end{subfigure}
\quad
\begin{subfigure}[t]{0.31\linewidth}
    \centering
       \includegraphics[trim={0cm 0cm 0cm 0cm},clip,width=\textwidth]{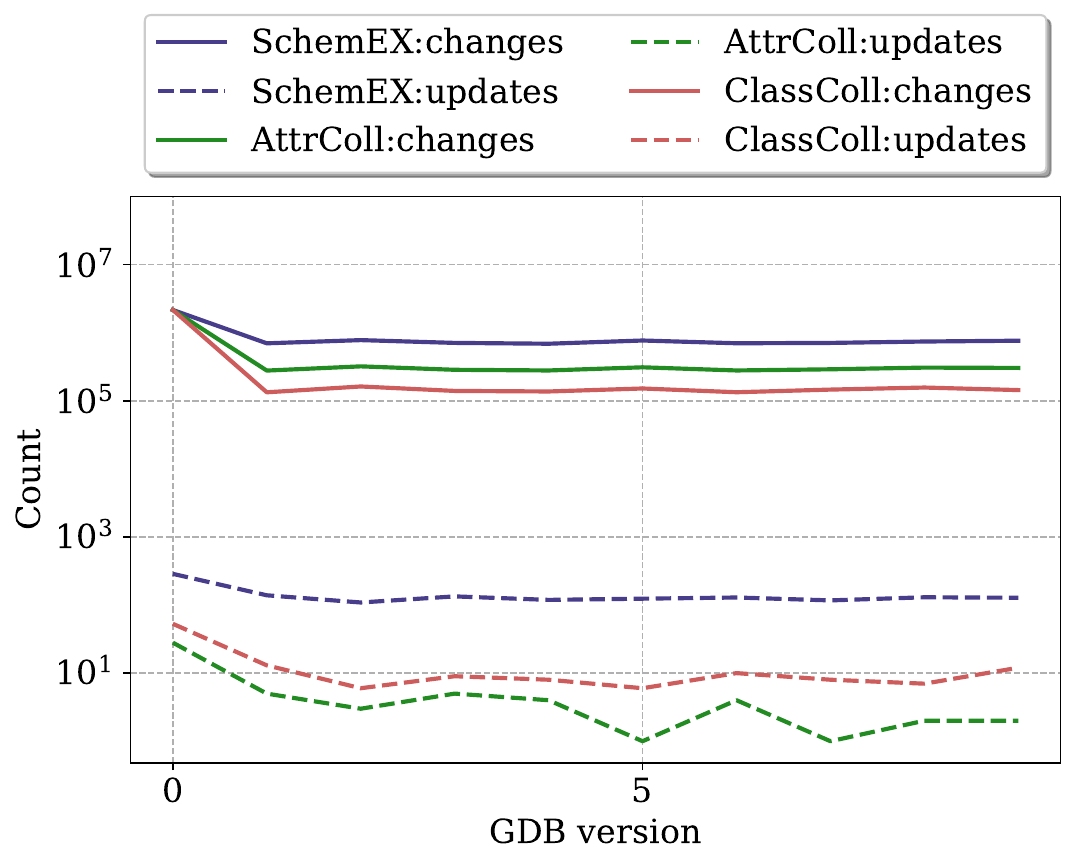}
	\caption{\label{fig:LUBM-changes} LUBM100 dataset update.}
\end{subfigure}
\quad
\begin{subfigure}[t]{0.31\linewidth}
    \centering
    \includegraphics[trim={0cm 0cm 0cm 0cm},clip,width=\textwidth]{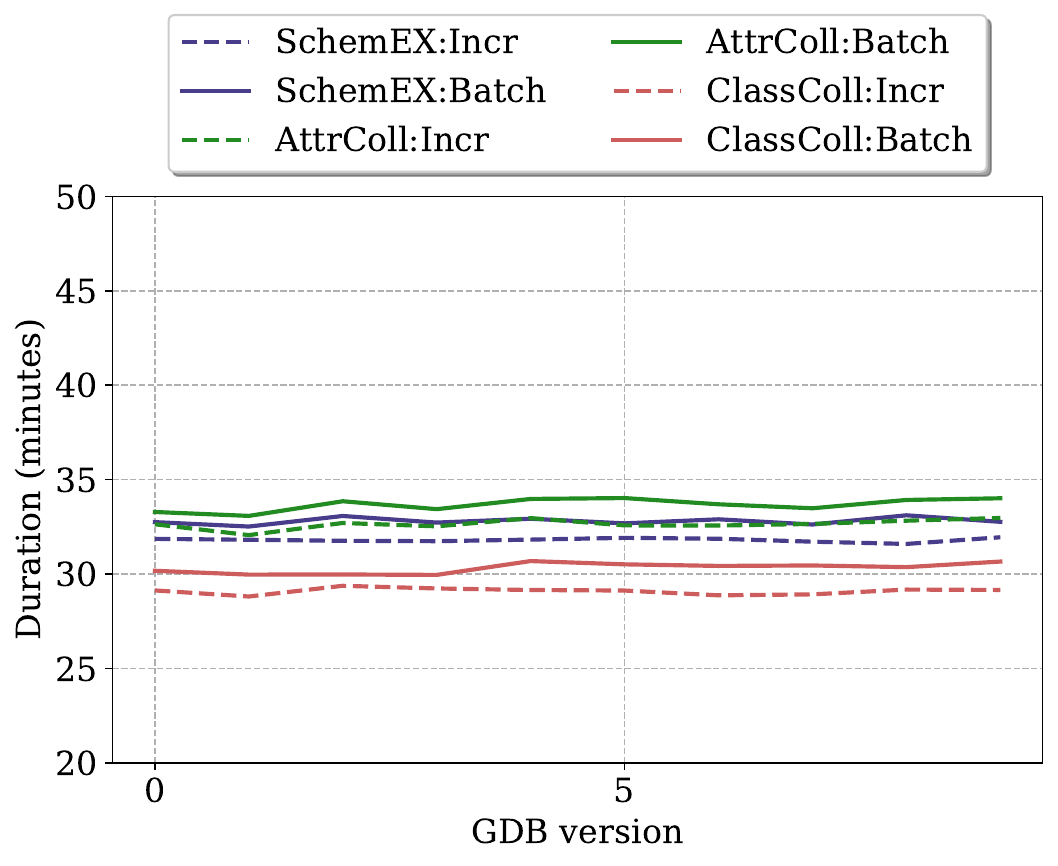}
    \caption{\label{fig:LUBM-performance} LUBM100 dataset performance.}
\end{subfigure}
\par\bigskip 
\begin{subfigure}[t]{0.31\linewidth}
    \centering
    \includegraphics[trim={0cm 0cm 0cm 0cm},clip,width=\textwidth]{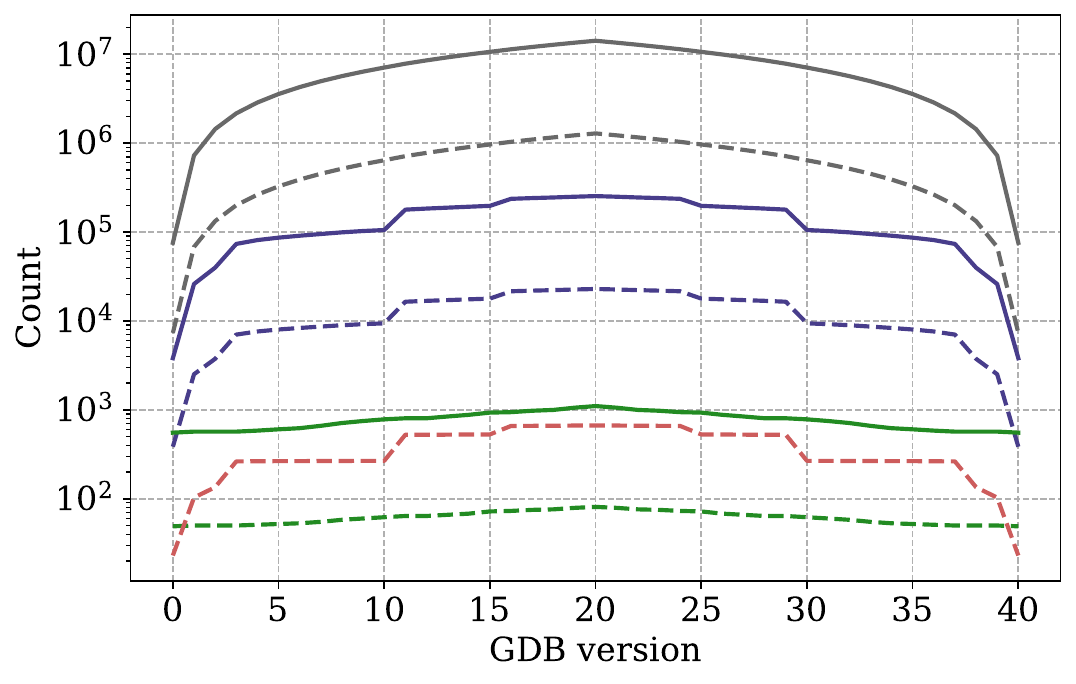}
    \caption{\label{fig:BSBM-size} BSBM dataset size.}
\end{subfigure}
\quad
\begin{subfigure}[t]{0.31\linewidth}
    \centering
    \includegraphics[trim={0cm 0cm 0cm 0cm},clip,width=\textwidth]{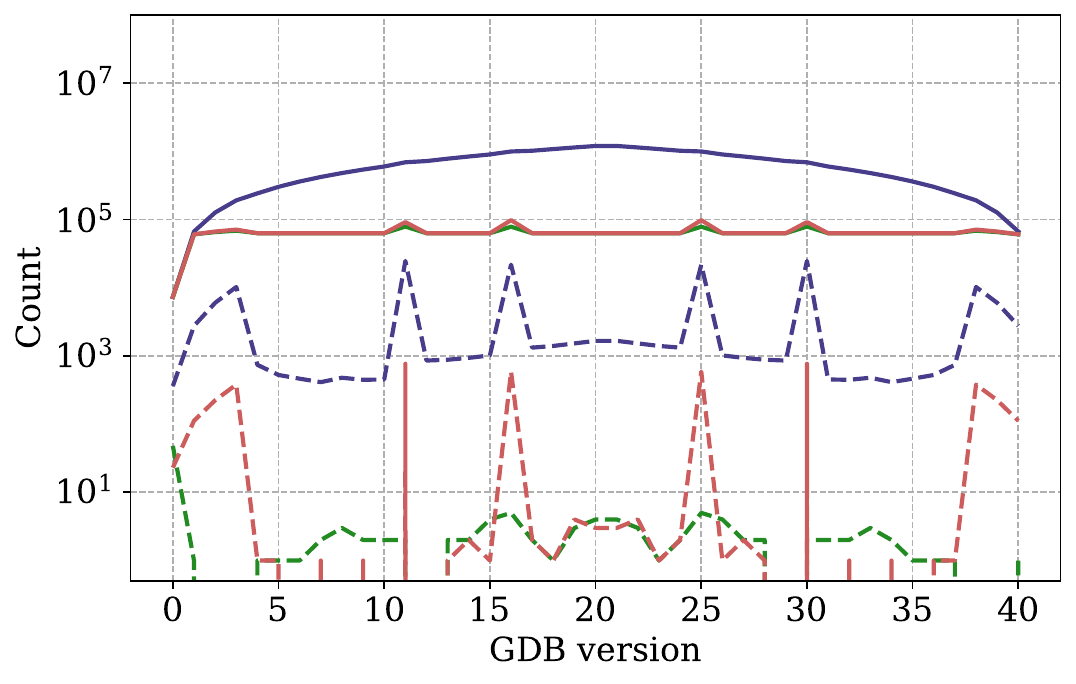}
	\caption{\label{fig:BSBM-changes} BSBM dataset update.}
\end{subfigure}
\quad
\begin{subfigure}[t]{0.31\linewidth}
    \centering
    \includegraphics[trim={0cm 0cm 0cm 0cm},clip,width=\textwidth]{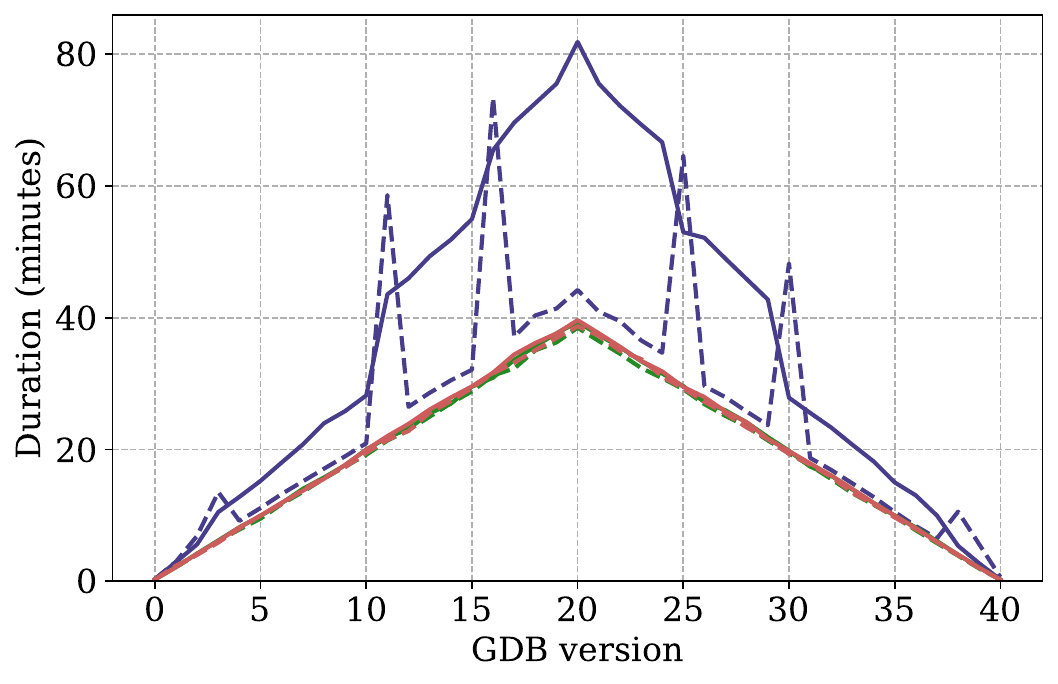}
    \caption{\label{fig:BSBM-performance} BSBM dataset performance.}
\end{subfigure}
\par\bigskip 
\begin{subfigure}[t]{0.31\linewidth}
    \centering
    \includegraphics[trim={0cm 0cm 0cm 0cm},clip,width=\textwidth]{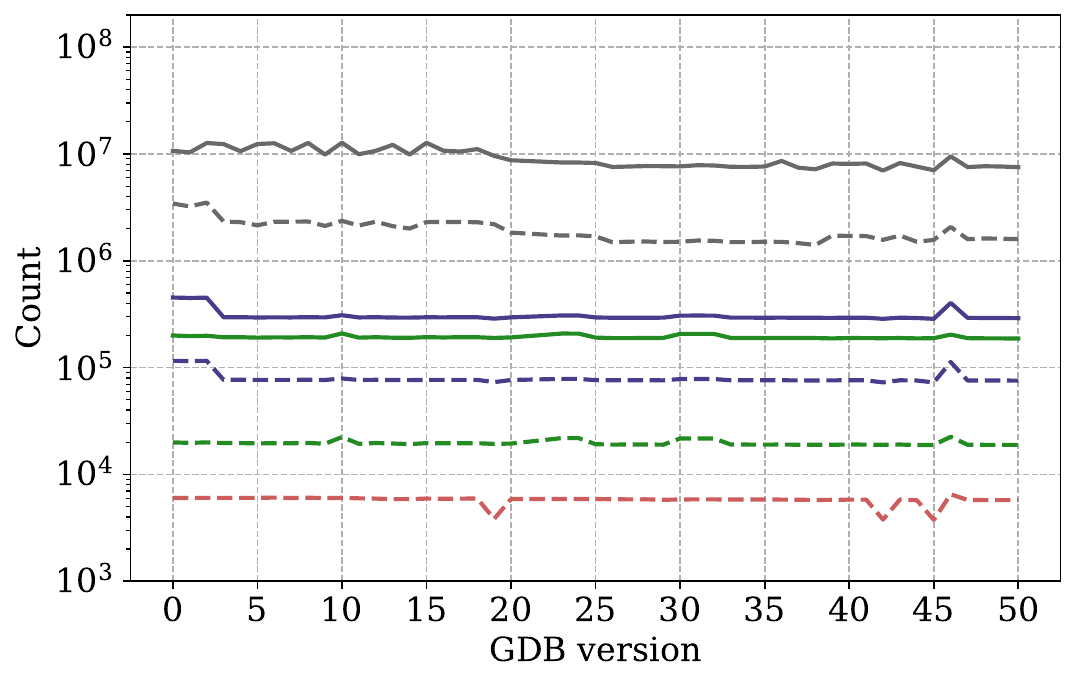}
    \caption{\label{fig:DyLDO-core-size} DyLDO-core size.}
\end{subfigure}
\quad
\begin{subfigure}[t]{0.31\linewidth}
    \centering
    \includegraphics[trim={0cm 0cm 0cm 0cm},clip,width=\textwidth]{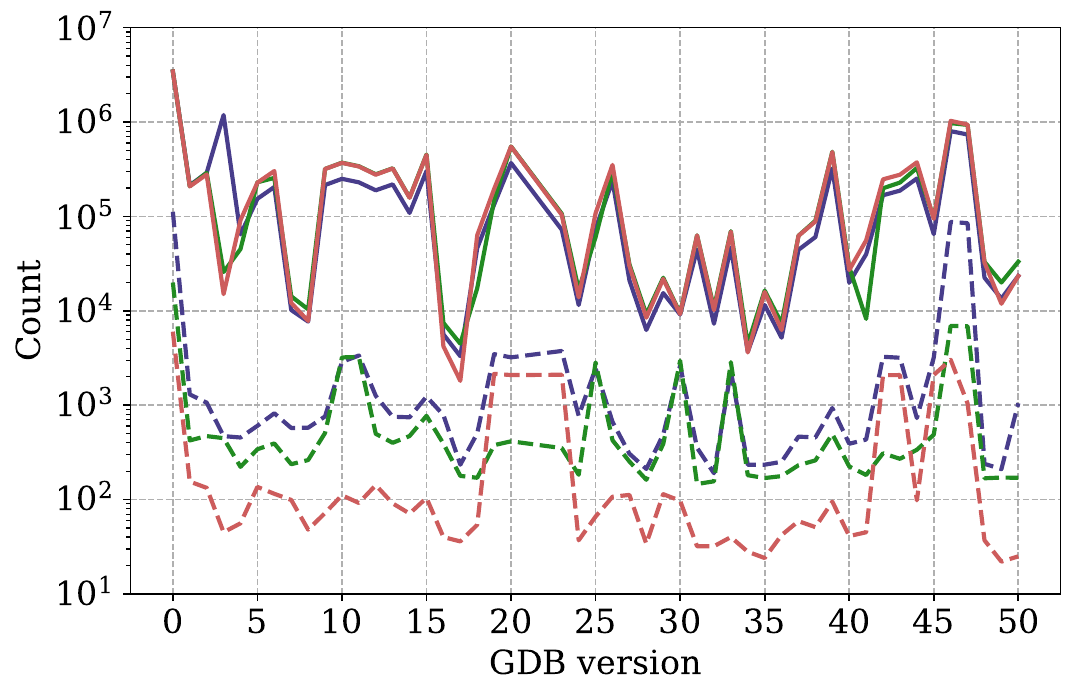}
	\caption{\label{fig:DyLDO-core-changes} DyLDO-core update.}
\end{subfigure}
\quad
\begin{subfigure}[t]{0.31\linewidth}
    \centering
    \includegraphics[trim={0cm 0cm 0cm 0cm},clip,width=\textwidth]{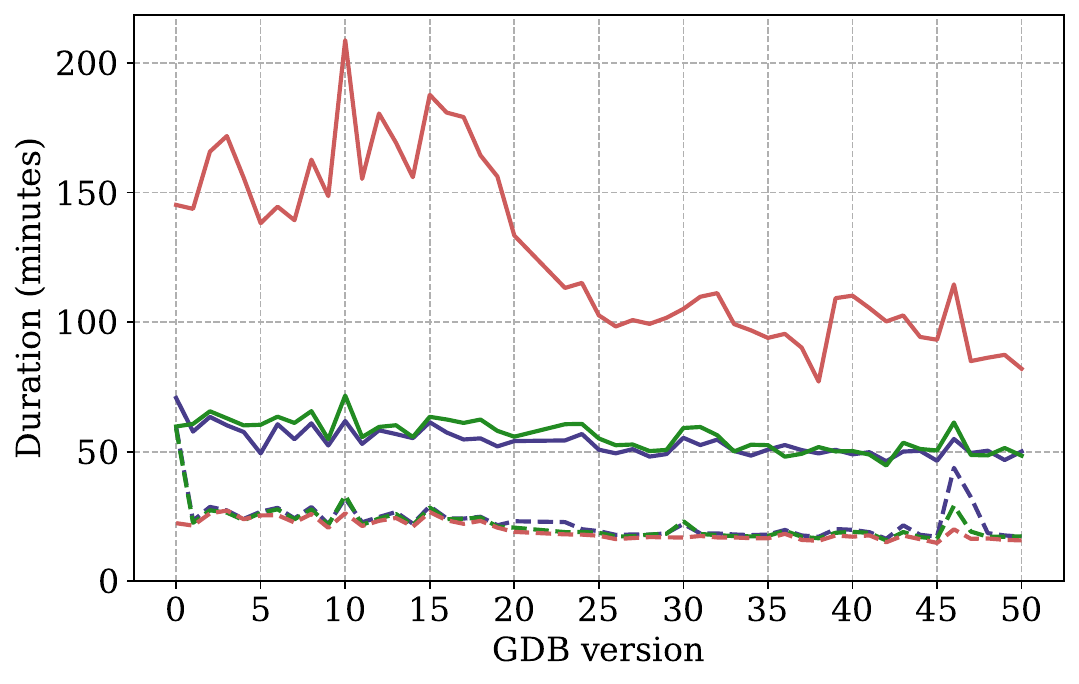}
     \caption{\label{fig:DyLDO-core-performance} DyLDO-core performance.}
\end{subfigure}
\par\bigskip 
\begin{subfigure}[t]{0.31\linewidth}
    \centering
    \includegraphics[trim={0cm 0cm 0cm 0cm},clip,width=\textwidth]{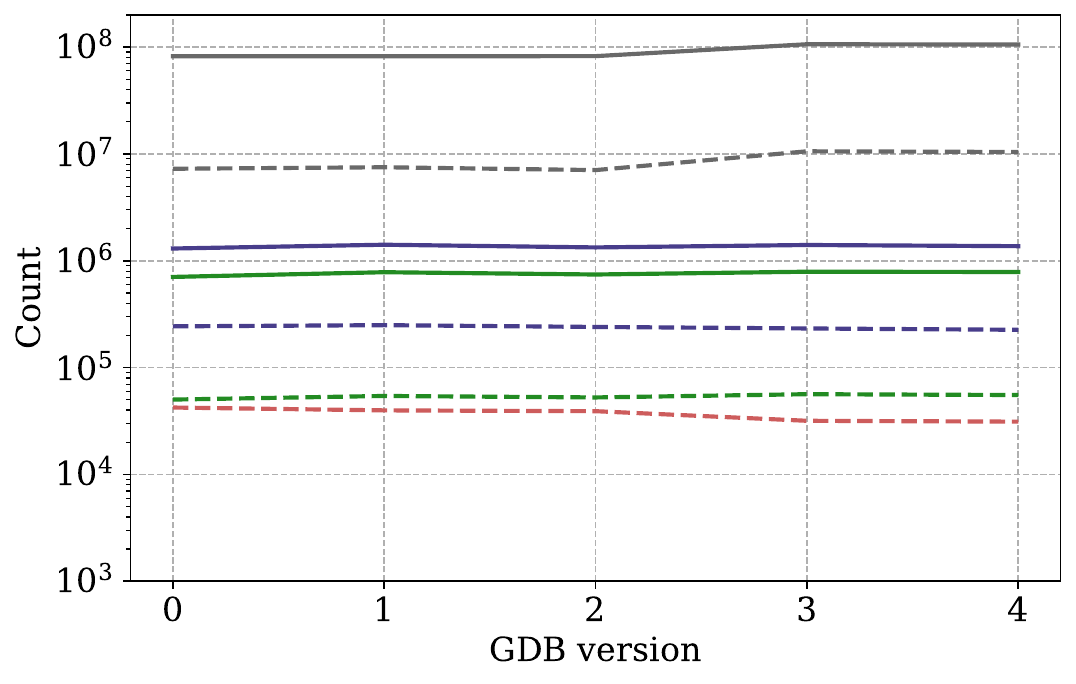}
    \caption{\label{fig:DyLDO-ext-size} DyLDO-ext dataset size.}
\end{subfigure}
\quad
\begin{subfigure}[t]{0.31\linewidth}
    \centering
    \includegraphics[trim={0cm 0cm 0cm 0cm},clip,width=\textwidth]{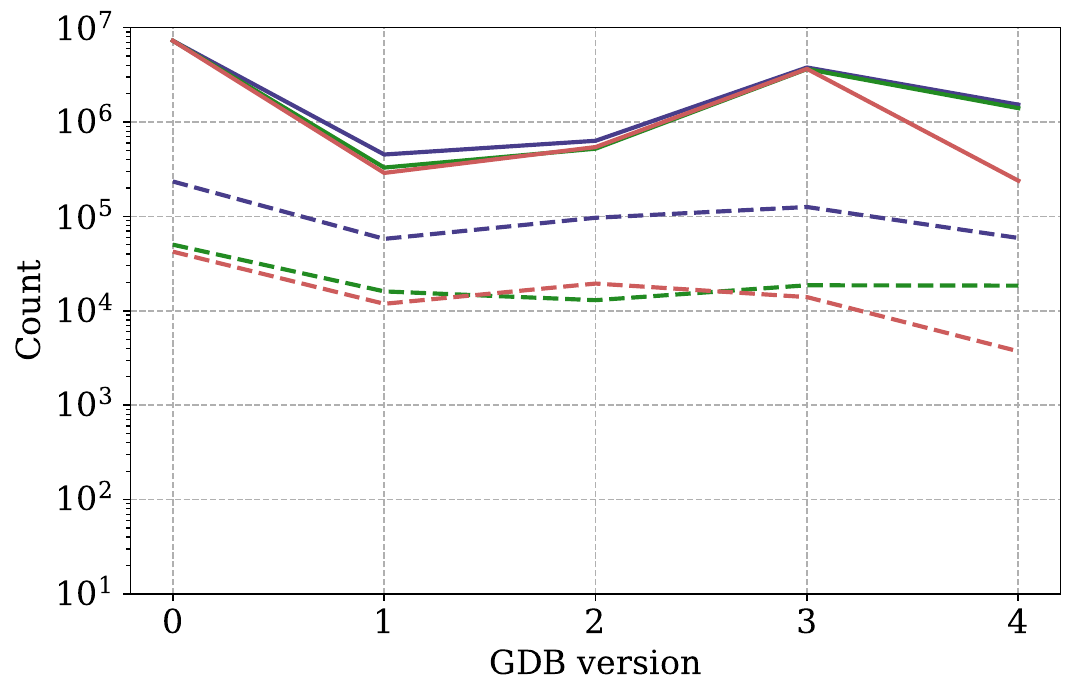}
	\caption{\label{fig:DyLDO-ext-changes} DyLDO-ext update.}
\end{subfigure}
\quad
\begin{subfigure}[t]{0.31\linewidth}
    \centering
    \includegraphics[trim={0cm 0cm 0cm 0cm},clip,width=\textwidth]{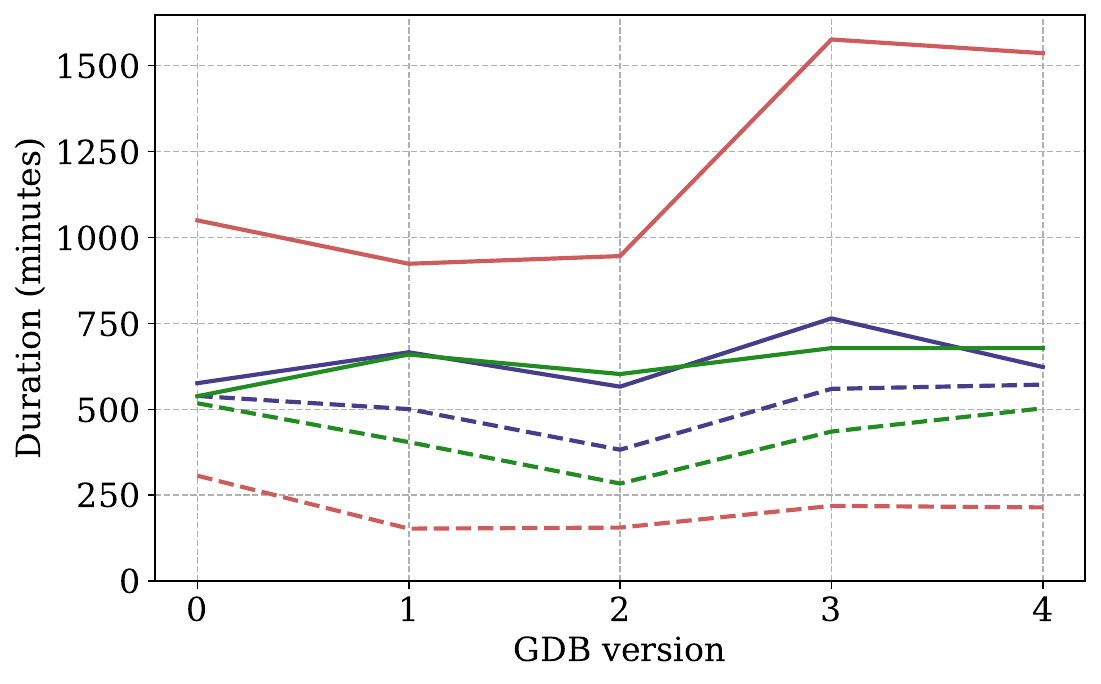}
    \caption{\label{fig:DyLDO-ext-performance} DyLDO-ext performance.}
\end{subfigure}

\caption{\label{fig:main-results-benchmark}Results of the experimental evaluation of three summary models: SchemEX (blue), Attribute Collection (green), and Class Collection (red). Each row is a different dataset and each column a different metric.}
}
\end{figure*}

\subsection{Results}
\label{sec:results}
The experimental results are visualized in \cref{fig:main-results-benchmark}. 
Each plot shows on the x-axis the database versions over time and on the y-axis the metrics.
From left to right, the columns of \cref{fig:main-results-benchmark} show the metrics size, update complexity, and performance.
From top to bottom, the rows of \cref{fig:main-results-benchmark} show the results for the four datasets, \ie LUBM100, BSBM, DyLDO-core, and DyLDO-ext. 
The subfigures in the left column of \cref{fig:main-results-benchmark} show the size of the GDB and the size of the graph summary $SG$ for each version of the database.
The summarization ratio is the number of vertices in the GDB divided by the number of vertices in the graph summary.
Intuitively, the distance between the colored dotted lines and the grey dotted line indicates the summarization ratio, and an intersection of both lines indicates a summarization ratio of~$1.0$.
Note that, for Attribute Collections (AttrColl) and Class Collections (ClassColl), the number of summary vertices $|\Vvs|$ is identical to the number of vertex summaries, as all summary vertices are primary vertices.
For SchemEX, secondary vertices are needed to model the schema structure.
Over all experiments and datasets, these secondary vertices make up only about $4\%$ of all vertices.
Thus, we do not display primary vertices separately in \cref{fig:main-results-benchmark}.
The graphs in the center column show the vertex change percentage and the graph summary update percentage.
The graphs in the right column show the average execution times of three consecutive repetitions of the graph summary computation.
Unless specified otherwise, we refer in the text to mean and standard deviation values computed over all versions of a GDB.

Comparing the size of the graph summaries ($|\Vvs|$, $|\Evs|$) and the graph database ($|V|$, $|E|$) (\cref{fig:main-results-benchmark} left column), the graph summaries are orders of magnitude smaller (consistently over all experiments).
Also, the Class Collections (ClassColl) and Attribute Collections (AttrColl) have higher summarization ratios than SchemEX.
This means, fewer vertex summaries are needed to partition the GDB based on vertex labels or edge labels compared to combining vertex and edge labels.
Over all datasets, Attribute Collections have the highest summarization ratios and have a factor of around $10^4$ fewer vertices, followed by Class Collections with a factor of around $10^3$ fewer vertices.
SchemEX has the lowest summarization ratio, \ie the highest number of vertex summaries are needed to partition the vertices.
Still, SchemEX has around a factor of $10^2$ fewer vertices than the GDB.
For Class Collections, no edges are needed to represent the vertex summaries.
Thus, for Class Collections, in each version the number of edges in the graph summary $|\Evs|$ is zero, which further reduces the overall size of the graph summary.

Comparing the \emph{changes} in GDB and the \emph{updates} on SG (\cref{fig:main-results-benchmark} center column), over all datasets, the number of updates on the graph summary is orders of magnitude smaller than the number of changes in the graph database.
On average, there are $12,967$ ($\pm 18,583$) more changes than updates for the Class Collection, $19,721$ ($\pm 30,157$) more for the Attribute Collection, and $901$ ($\pm 1,737$) more for SchemEX.
Furthermore, the incremental algorithm computes Class Collections the fastest in, on average, $27$~($\pm 33$) minutes and SchemEX the slowest in, on average, $42$ ($\pm 95$) minutes (\cref{fig:main-results-benchmark} right column).
Compared to the batch computation, this is a speed up of $1.8$ ($\pm 0.7$) for SchemEX, $1.8$ ($\pm 0.9$) for Attribute Collection, and $3.7$ ($\pm 2.7$) for Class Collection. 

\subsection{Discussion}
\label{sec:discussion}

The key insight from our experiments is that, over all summary models and datasets, the incremental algorithm is almost always faster than the batch counterpart.
In total, we have run $312$ experiments, \ie we have $n= 312$ measure points of our three graph summary models over all versions of our four datasets ($10 \times$ LUBM100, $40 \times $ BSBM, $49 \times$ DyLDO-core, $5 \times$ DyLDO-ext).
Detailed evaluation of the performance metrics shows a strong correlation\footnote{Since our experimental data does not follow a normal distribution but a skewed distribution (verified using D'Agostino's K-squared test), we calculate the Spearman rank-order correlation coefficient.} between the schema computation (Phase~1) and the number of edges in the GDB, $\rho(311)=0.867$, $p < .0001$,
substantiating our theoretical complexity analysis.
Phase~1 is identical for both algorithms (see \cref{sec:incremental-algorithm}).

Regarding the overall runtime (Phase~1 and~2), we observe that the incremental algorithm outperforms the batch variant almost always on the BSBM dataset, even though about $90\%$ of the GDB changes in each version.
Even for the DyLDO-core dataset, when about $46\%$ of the GDB changes from version 46 to 47, the incremental algorithm is still $1.5$ times faster.
This can be explained by the fact that changes in the GDB do not necessarily require an update to the graph summary.

If many updates are required, they seem to be more expensive on the BSBM dataset than on the DyLDO-core dataset since the batch algorithm outperforms the incremental algorithms during periods in which a large proportion of the data is updated only in the BSBM dataset.
The average vertex degree is more than twice as high in the BSBM dataset.
Since the update complexity depends on the degree of the vertex, this could be a reasonable explanation for the increased runtime.

From our results, we can also state that graph summarization on benchmark datasets is easier than on real-world datasets, \ie there are fewer vertex summaries needed to summarize a similar-sized graph database.
This observation highlights the importance of using real-world datasets when evaluating graph summarization algorithms, as the observed variety of schema structures in the DyLDO datasets is not covered by existing synthetic datasets.
All three graph summary models produce graph summaries of the DyLDO-core dataset that are $10$--$100$ times larger than on the similarly sized BSBM dataset at version $20$, \ie the $19$th temporal update. 
Experiments on the BSBM dataset suggest that cascading updates due to neighbor changes have a huge impact on the performance of the incremental graph summarization algorithm.
More than $99\%$ of all vertex modifications are due to a neighbor change.
Of all vertex changes, this makes up $83.88\%$ ($\pm 15.64\%$).

Batch computation of Class Collections on the two DyLDO datasets takes at least twice as long as the Attribute Collection and SchemEX.
This is surprising given that SchemEX is a much more complex graph summary model than the Class Collection model.
We noticed this in several runs and on a different test system.
Further investigation revealed that this is caused by a few \enquote{hot} vertex summaries that summarize most of the vertices when using the Class Collection.
In particular, the vertex summary that summarizes all vertices in the GDB that have \textit{no type information} (\ie that have an empty label set $\ell_V$) stands out.
During the find and merge phase, the payload information of different vertices that are summarized by the same vertex summary $vs$ is merged (\addInstance{}).
Although individual merges of payload information in our experiment are done in constant time (merging sets of source graph labels), merging almost all the summary's payload information into a few centrally stored payload elements cannot be done in parallel.
This means that payload merging for the Class Collection on the DyLDO-core dataset requires an exceptional number of synchronizations. 
For example, the vertex summary that summarizes vertices with no type information in the Class Collection needs on average $12,828$ ($\pm 4,148$) synchronizations for each version update.
In contrast, for Attribute Collection and SchemEX, on average $122$ ($\pm 81$) synchronizations are needed for all vertex summaries.

For the incremental algorithm, this has no effect since the payload information is updated in parallel in the \vertexHashIndex{}.
For the \vertexHashIndex{}, no additional synchronization for updating L3 is needed.
One might consider this as a design flaw in the batch algorithm, which determines whether the incremental algorithm outperforms the batch algorithm or vice versa.
However, first, there is no alternative to storing the payload information in the graph summary since the batch algorithm uses no additional data structure.
Second, the source graph payload is only used in real-world datasets.
For the benchmark datasets, no additional synchronization due to payload merges is done during the batch computation.
Still, the incremental algorithm, with very few exceptions, outperforms the batch counterpart.

To compare the time performance of incremental versus batch computation, we evaluated with over $100$ GDB versions of four datasets (two synthetic and two real-world).
Each dataset has different characteristics in terms of data change rates, types of changes, schema heterogeneity, size, and degree.
Furthermore, we computed three representative summary models on these datasets.
The selected summary models are widely used across different tasks and datasets~\cite{DBLP:journals/vldb/CebiricGKKMTZ19}.
Thus, the results capture the performance of our graph summarization algorithm for a wide, representative range of experimental settings.
Beyond the three representative graph summary models used in the evaluation, our incremental algorithm is suitable for other summary models reported in the literature, which is possible due to our \model{} language~\cite{DBLP:journals/tcs/BlumeRS21}.

\section{Experiment 2: Scalability of Incremental Summarization through Parallelization}
\label{sec:experiment-2}
In this experiment, we measure the scalability of our incremental algorithm.
Particularly, we are interested how much performance our incremental algorithms still gains through parallelization when using more CPU cores.

\begin{figure*}[t]
\begin{subfigure}[t]{0.23\linewidth}
    \centering
    \includegraphics[trim={0cm 0cm 0cm 1cm},clip,width=\textwidth]{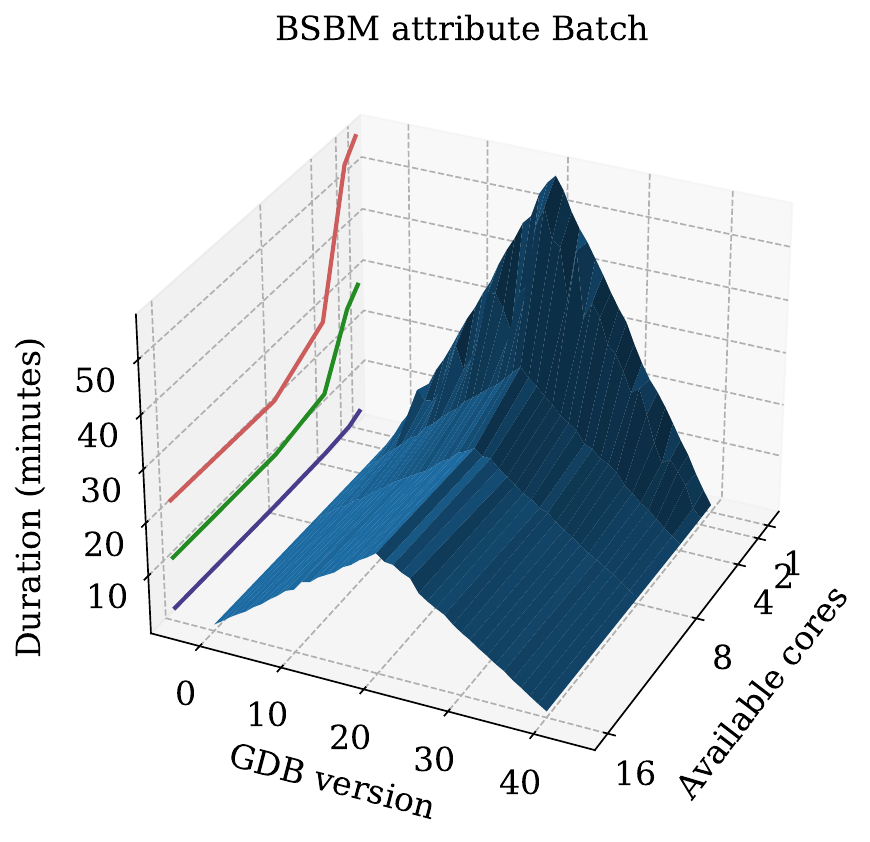}
     \caption{\label{fig:BSBM-attribute-batch}Batch computation of Attribute Collection on BSBM.}
\end{subfigure}
\quad
\begin{subfigure}[t]{0.23\linewidth}
    \centering
    \includegraphics[trim={0cm 0cm 0cm 1cm},clip,width=\textwidth]{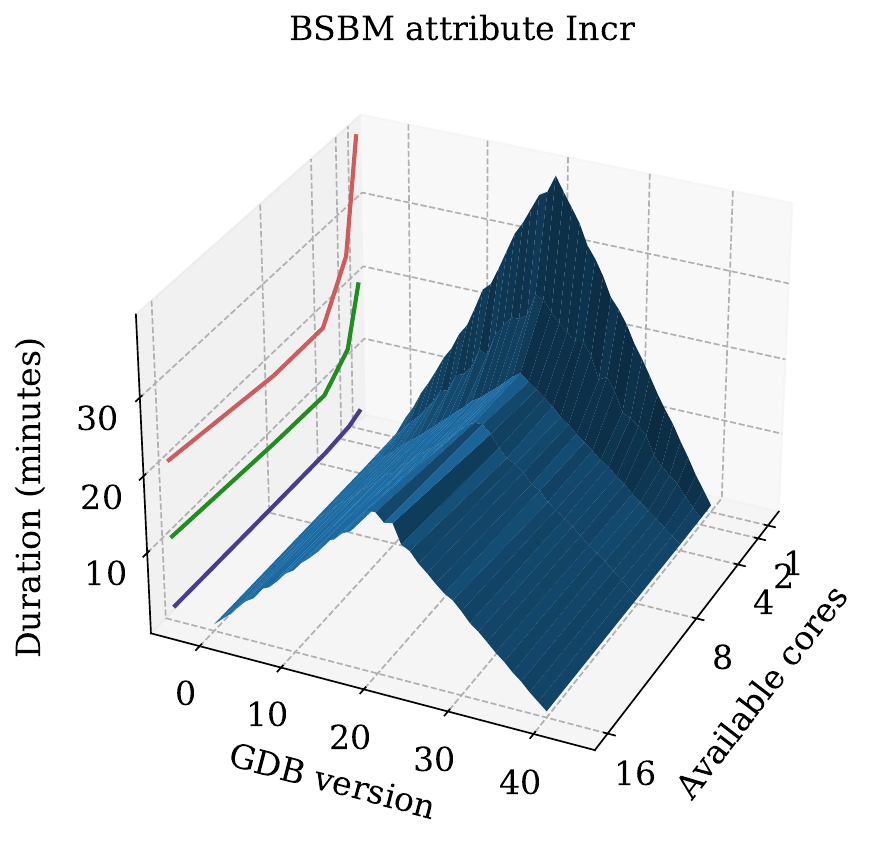}
	\caption{\label{fig:BSBM-attribute-incr}Incremental computation of Attribute Collection on BSBM.}
\end{subfigure}
\quad
\begin{subfigure}[t]{0.23\linewidth}
    \centering
    \includegraphics[trim={0cm 0cm 0cm 1cm},clip,width=\textwidth]{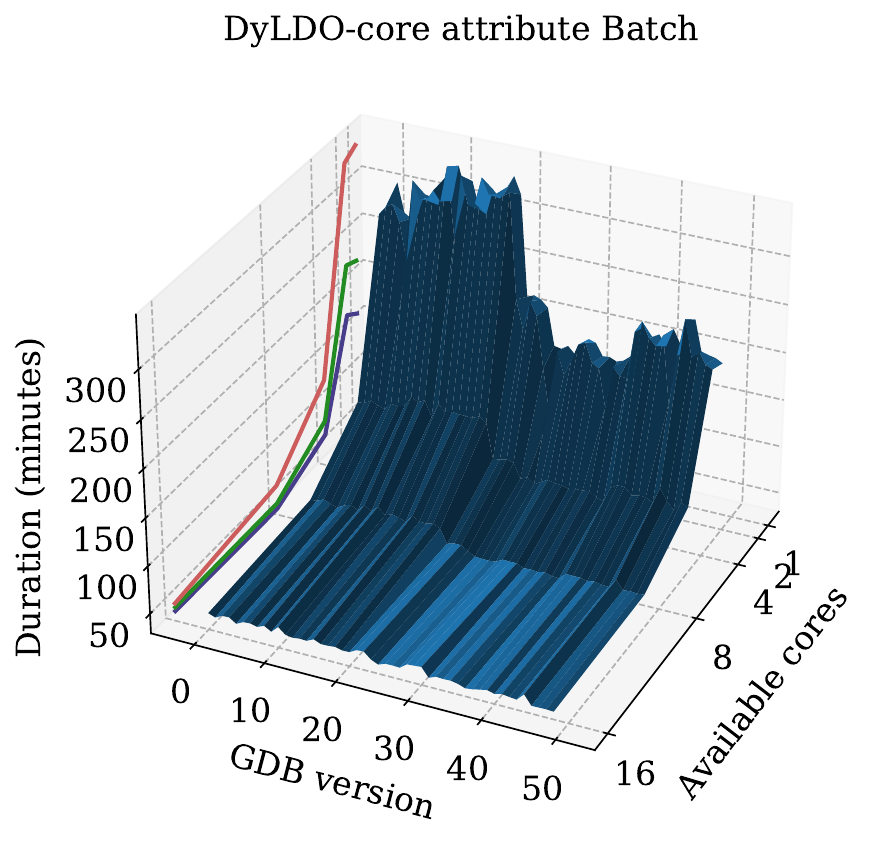}
     \caption{\label{fig:DyLDO-core-attribute-batch}Batch computation of Attribute Collection on DyLDO-core.}
\end{subfigure}
\quad
\begin{subfigure}[t]{0.23\linewidth}
    \centering
    \includegraphics[trim={0cm 0cm 0cm 1cm},clip,width=\textwidth]{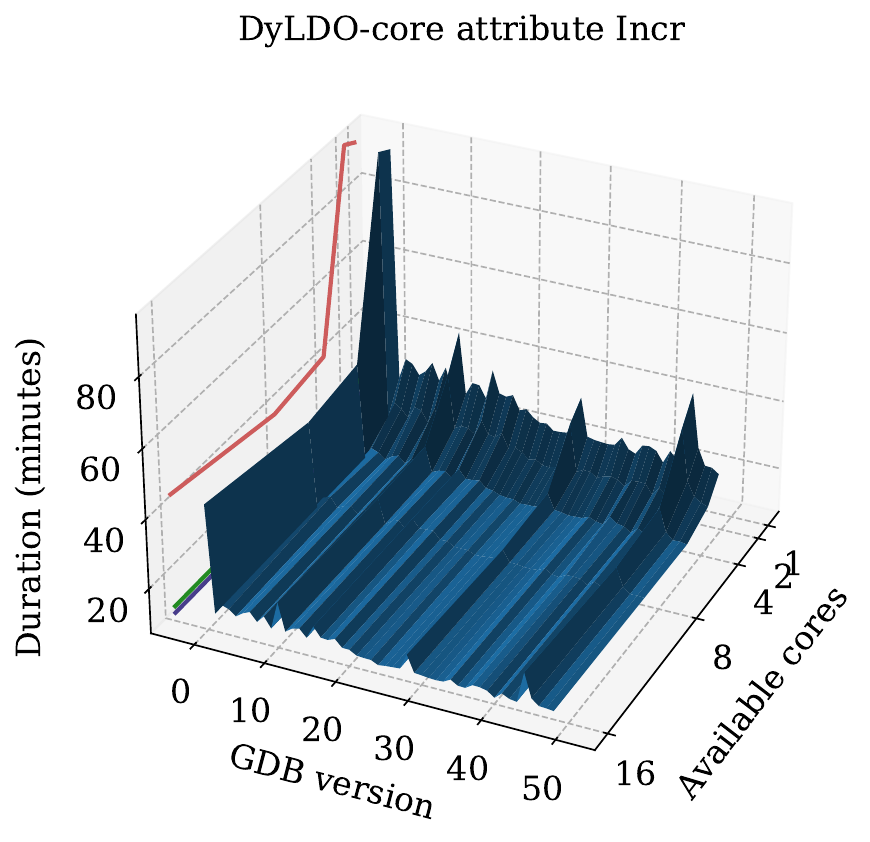}
	\caption{\label{fig:DyLDO-core-attribute-incr}Incremental computation of Attribute Collection on DyLDO-core.}
\end{subfigure}
\quad
\begin{subfigure}[t]{0.23\linewidth}
    \centering
    \includegraphics[trim={0cm 0cm 0cm 1cm},clip,width=\textwidth]{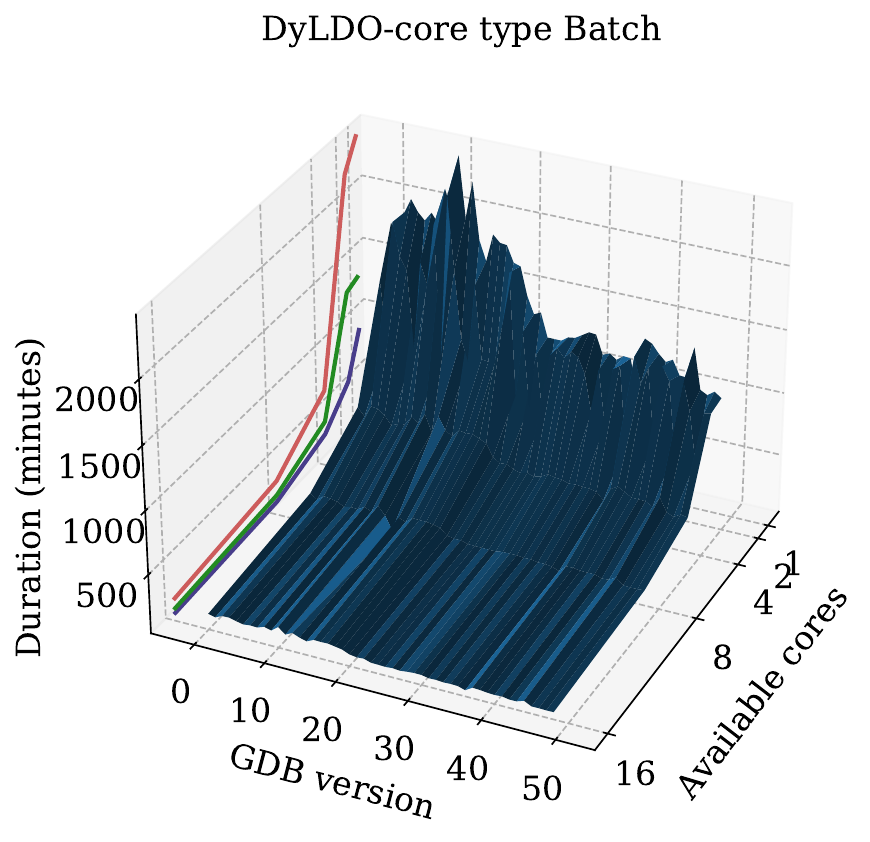}
     \caption{\label{fig:DyLDO-core-type-batch}Batch computation of Class Collection on DyLDO-core.}
\end{subfigure}
\quad
\begin{subfigure}[t]{0.23\linewidth}
    \centering
    \includegraphics[trim={0cm 0cm 0cm 1cm},clip,width=\textwidth]{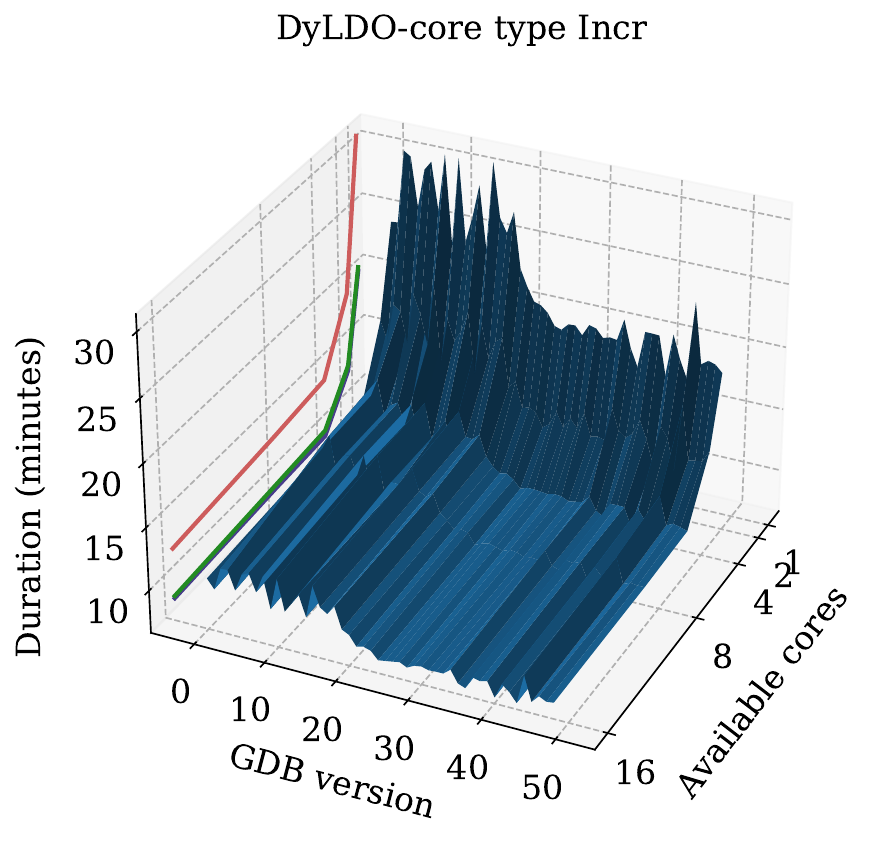}
	\caption{\label{fig:DyLDO-core-type-incr}Incremental computation of Class Collection on DyLDO-core.}
\end{subfigure}
\quad
\begin{subfigure}[t]{0.23\linewidth}
    \centering
    \includegraphics[trim={0cm 0cm 0cm 1cm},clip,width=\textwidth]{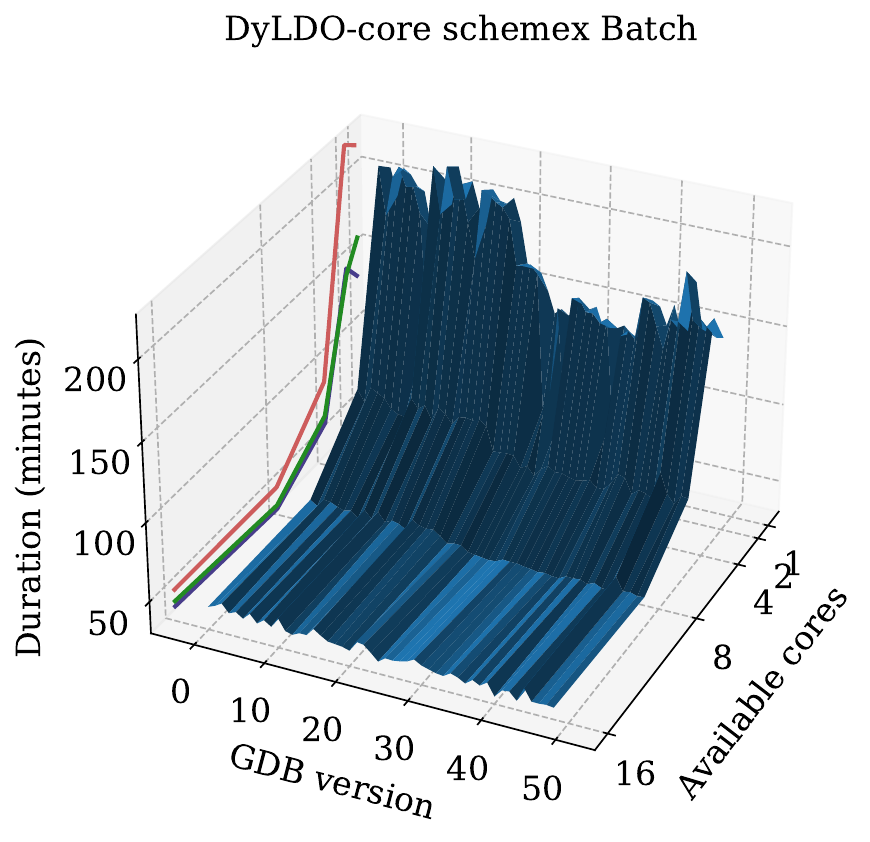}
     \caption{\label{fig:DyLDO-core-schemex-batch}Batch computation of SchemEX on DyLDO-core.}
\end{subfigure}
\quad
\begin{subfigure}[t]{0.23\linewidth}
    \centering
    \includegraphics[trim={0cm 0cm 0cm 1cm},clip,width=\textwidth]{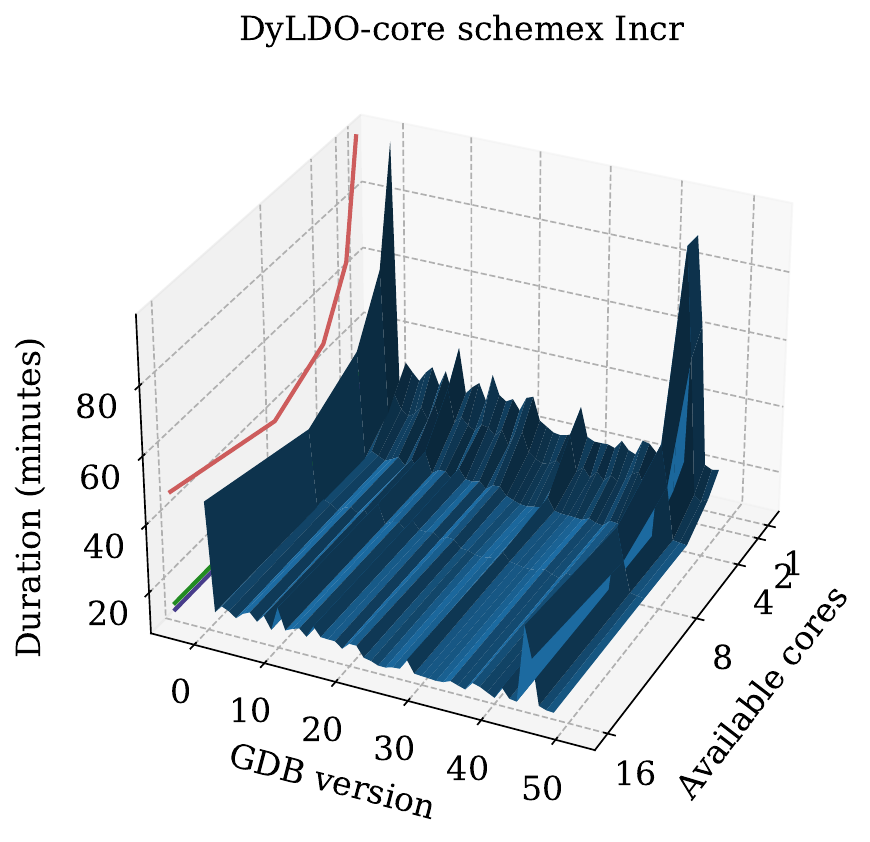}
	\caption{\label{fig:DyLDO-core-schemex-incr}Incremental computation of SchemEX on DyLDO-core.}
\end{subfigure}

\caption{\label{fig:results-experiment-2}Performance of batch and incremental computation of Attribute Collection on the BSBM dataset and Attribute Collection, Class Collection, and SchemEX on the DyLDO-core dataset. The $x$-axis shows the GDB version, the $y$-axis the duration of the summary computation, and the $z$-axis the number of available cores. The three contour lines show the runs with the shortest (blue), longest (red), and median duration (green).} 
\vspace{1cm}
\end{figure*}

\subsection{Procedure}
We systematically analyze the impact of the number of cores available for the graph summary computation.
To this end, we repeat Experiment~1, but we limit the number of available cores for the computation to $1$, $2$, $4$, $8$, and~$16$ cores.
We focus on two datasets: the synthetic BSBM dataset and the real-world DyLDO-core dataset.
We designed the evolution of the BSBM benchmark dataset so that the size grows linearly for the first $20$ versions and shrinks linearly for the remaining $20$ versions.  
Thus, we can easily see any non-linear performance changes.
As a representative real-world dataset, we use the DyLDO-core dataset.
As shown in Experiment~1, real-world datasets have different characteristics compared to the benchmark datasets.
DyLDO-core is considerably smaller than DyLDO-full, which allows all graph summaries to be computed, even with a small number of cores, in reasonable time. 
Analogously to Experiment~1, the runtime performance of the algorithm is measured by the time needed to compute a summary from scratch or, in the incremental case, to update the summary $SG$ for each version of the GDB.

We observed in Experiment~1 that graph summary computation on real-world datasets is considerably more complex than on synthetically generated datasets, which leads to significant differences between the evaluated graph summary models.
Furthermore, overall, the performance of Attribute Collection, Class Collection, and SchemEX differs only marginally on the synthetic BSBM dataset.
Thus, for this experiment we focus on the real-word DyLDO-core dataset.
Given these two observations from Experiment~1, we can reasonably assume that for our three graph summary models, the impact of using a different number of cores will also differ only marginally.  
Thus, for the BSBM dataset, we only report the evaluation results of Attribute Collection.

To compare the performance results over multiple versions of a dataset that evolves over time, we aggregate the speed-up for each version.
Inspired by macro and micro F1-scores, we report the macro and micro average speed-ups in \cref{tab:parallelization-speed-up-batch-incr}.
For the macro average speed-up, we compute the speed-up for each version and then calculate the mean and the standard deviation.
For the micro average speed-up, we sum-up the execution times for each version and then calculate the mean speed-up. 

\begin{table*}[h!]
{\renewcommand\arraystretch{1.25}
    \centering
    \footnotesize
    \caption{\label{tab:parallelization-speed-up-batch-incr}The macro and micro average speed-ups comparing the batch and incremental algorithms for different numbers of cores. The incremental algorithm outperforms the batch algorithm (speed-up value greater than $1$).}
    \begin{tabularx}{\textwidth}{p{0.2cm} p{1.8cm} p{1.9cm} *{5}{>{\centering\arraybackslash}X}}
        \toprule
        & & & \multicolumn{4}{c}{\textbf{Average Speed-up}}\\
        & \textbf{Model} & \textbf{Metric}
        & \textbf{$\bm{1}$ core}
        & \textbf{$\bm{2}$ cores}
        & \textbf{$\bm{4}$ cores}
        & \textbf{$\bm{8}$ cores}
        & \textbf{$\bm{16}$ cores}\\
        \midrule
        \arrayrulecolor{black!0}\midrule
        \parbox[t]{4mm}{\multirow{2}{*}{\rotatebox[origin=c]{90}{\textbf{\shortstack[c]{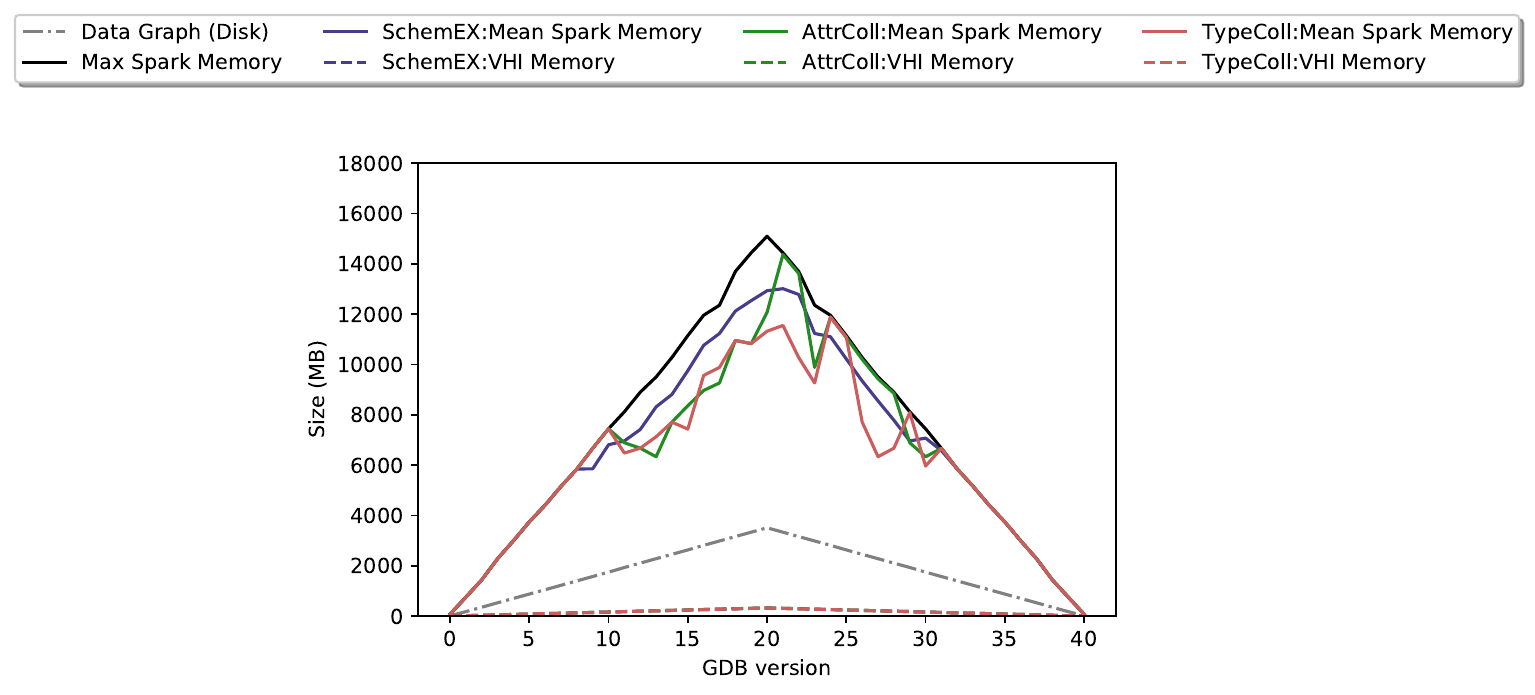}}}}}
        & \parbox[t]{4mm}{\multirow{2}{*}{\textbf{\shortstack[c]{AttrColl}}}}
        & \textbf{Macro Avg.}
        & $1.44 \pm 0.05$ 
        & $1.83 \pm 0.31$ 
        & $1.46 \pm 0.08$ 
        & $1.16 \pm 0.05$
        & $1.07 \pm 0.04$ \\
        & & \textbf{Micro Avg.}
        & $1.44$ 
        & $1.86$ 
        & $1.48$ 
        & $1.17$
        & $1.07$ \\
        \arrayrulecolor{black!0}\midrule
        \arrayrulecolor{black!50}\midrule
        \parbox[t]{4mm}{\multirow{6}{*}{\rotatebox[origin=c]{90}{\textbf{\shortstack[c]{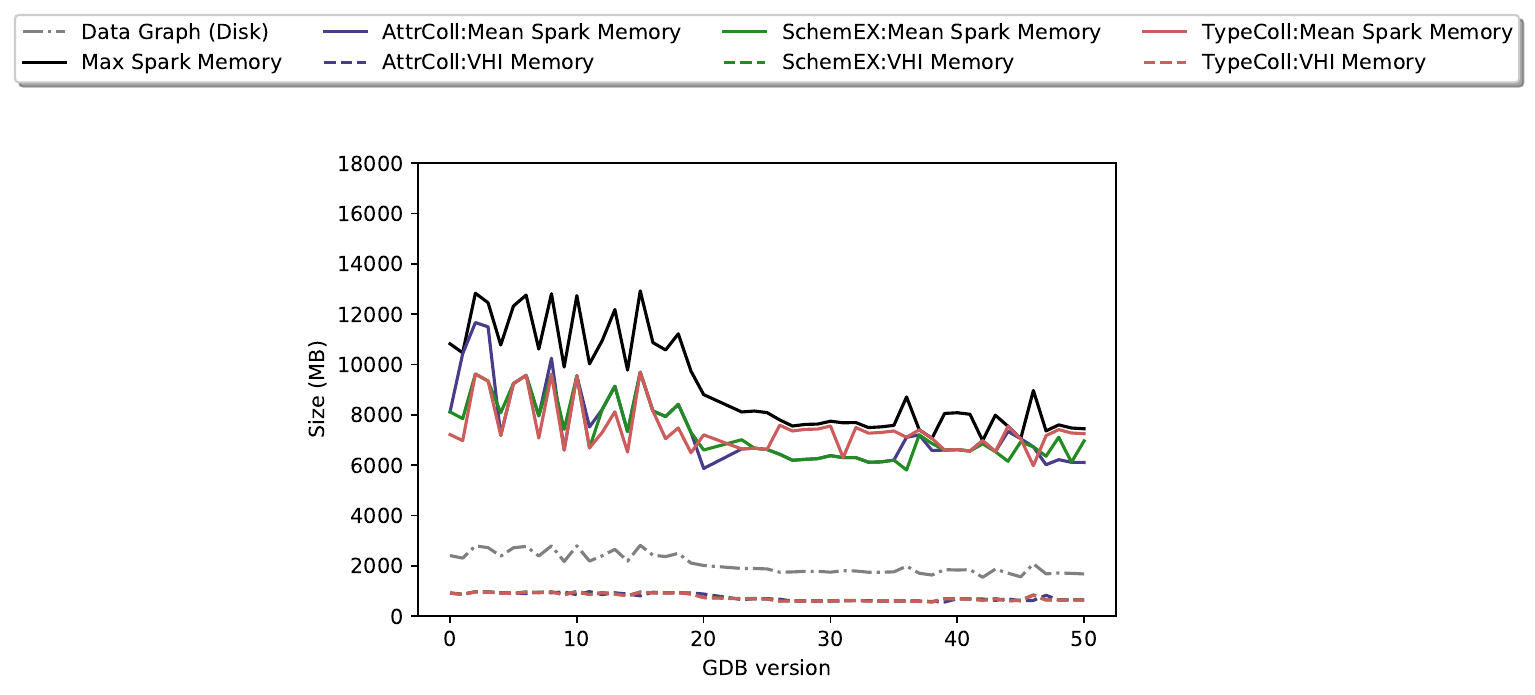}}}}}
        & \parbox[t]{4mm}{\multirow{2}{*}{\textbf{\shortstack[c]{AttrColl}}}}
        & \textbf{Macro Avg.}
        & $9.35 \pm 1.81$ 
        & $13.48 \pm 3.21$ 
        & $7.62 \pm 1.33$ 
        & $4.02 \pm 0.56$
        & $3.71 \pm 0.66$\\
        & & \textbf{Micro Avg.}
        & $8.98$ 
        & $12.44$ 
        & $7.38$ 
        & $3.89$
        & $3.48$\\
        \arrayrulecolor{black!25}\cmidrule(r{1pt}){2-8}
        & \parbox[t]{4mm}{\multirow{2}{*}{\textbf{\shortstack[c]{ClassColl}}}}
        & \textbf{Macro Avg.}
        & $64.52 \pm 7.7$ 
        & $77.82 \pm 20.29$ 
        & $44.08 \pm 4.45$ 
        & $21.48 \pm 2.23$
        & $14.06 \pm 1.4$\\
        & & \textbf{Micro Avg.}
        & $65.23$ 
        & $77.87$ 
        & $44.32$ 
        & $21.48$
        & $14.16$\\
        \arrayrulecolor{black!25}\cmidrule(r{1pt}){2-8}
        & \parbox[t]{4mm}{\multirow{2}{*}{\textbf{\shortstack[c]{SchemEX}}}}
        & \textbf{Macro Avg.}
        & $6.26 \pm 1.23$ 
        & $8.43 \pm 2.14$ 
        & $5.49 \pm 1.01$ 
        & $3.07 \pm 0.47$
        & $3.26 \pm 0.68$\\
        & & \textbf{Micro Avg.}
        & $5.71$ 
        & $7.55$ 
        & $5.07$ 
        & $2.92$
        & $3.0$\\
        \arrayrulecolor{black!100}\bottomrule
    \end{tabularx}
}
\end{table*}

\subsection{Results}

The main results of Experiment~2 are shown in \cref{fig:results-experiment-2}.
For both datasets, our main observation is that increasing the number of available cores improves the performance of batch and incremental computation. 
Remarkable, however, is that the single-core incremental runs on the real-world DyLDO-core dataset are competitive with the $16$-core runs of the batch computation, often even outperforming the batch computation.
Attribute Collections are computed incrementally between $1.2$ and $2.4$ (mean $1.8 \pm 0.3$) times faster using only one core than the batch computation with $16$ cores on the DyLDO-core dataset.
For SchemEX, the incremental computation is between $0.5$ and $2.0$ (mean $1.6 \pm 0.3$) times faster.
Notably for SchemEX, only two single-core incremental computations were slower than the $16$-core batch computations, namely versions $46$ (speed-up of $0.5$) and $47$ (speed-up of $0.6$).
As discussed in Experiment~1, here about $50\%$ of the data graph changed. 
For Class Collections, the speed-up is even higher than for Attribute Collections.
Class Collections are computed incrementally between $4.9$ and $6.9$ (mean $5.7 \pm 0.5$) times faster using only one core than the batch computation with $16$ cores on the DyLDO-core dataset.

In general, adding more cores benefits the batch computation more than the incremental computation. 
For the BSBM benchmark dataset, the four-core batch runs outperform the single-core incremental runs.
Beyond four cores, incremental computation is as fast as batch computation, within a margin of error of $30$ seconds.  
On the BSBM dataset version~$20$ (the largest version), the single-core batch computation of Attribute Collection takes $57$~minutes and the $16$-core batch computation takes $22$~minutes.
This is a speed-up of about~$2.6$.
For the same BSBM dataset version, the incremental computation time for Attribute Collection drops from $40$~minutes (single core) to $20$~minutes ($16$ cores), a speed-up of~$2.0$.

The difference between batch and incremental computation is more apparent in the DyLDO-core dataset.
Here, single-core batch computations of Attribute Collection take $157$ to $338$~minutes (mean $237 \pm 60$) and $16$-core batch computations only take $37$ to $51$~minutes (mean $44 \pm 4$).
Thus, the speed-up for the batch computation of Attribute Collection when going from $1$ to~$16$ cores is between $3.8$ and~$7.4$ (mean $5.4 \pm 1.2$).
For SchemEX, we observe a similar speed-up behavior. 
Single-core batch computation of SchemEX takes $131$ to $215$~minutes (mean $167 \pm 27$) and $16$-core batch computations only take $35$ to $49$~minutes (mean $42 \pm 3$).
For the extreme case of the Class Collection, single-core batch computation takes $879$ to $2430$~minutes (mean $1463 \pm 398$) and $16$-core takes $84$ to $199$~minutes (mean $129 \pm 30$).
This is a speed-up of between $10.4$ and $12.7$ (mean $11.2 \pm 0.6$).

In contrast, excluding the first initial summary computation, the single-core incremental computations of all three summary models take $17$ to $90$~minutes (mean $25 \pm 7$) and the $16$-core incremental computations take $7$ to $33$~minutes (mean $11 \pm 3$).
The computation of the first graph summary is a batch computation, as there is no prior graph summary to update. 
For incremental computation, the speed-up from one core to $16$~cores is between $1.9$ and~$2.7$ (mean $2.2 \pm 0.1$).
In general for DyLDO-core, incrementally, SchemEX is computed fastest and Attribute Collections slowest.

\begin{table*}[t!]
{\renewcommand\arraystretch{1.25}
    \centering
    \footnotesize
    \caption{\label{tab:parallelization-speed-up}The mean of the relative speed-up when doubling the number of cores used for the batch and incremental algorithms. A speed-up of $2.0$ means that doubling the number of cores halves the computation time.}
    \begin{tabularx}{\textwidth}{p{0.2cm} p{1.8cm} p{1.8cm} *{4}{>{\centering\arraybackslash}X}}
        \toprule
        & & & \multicolumn{4}{c}{\textbf{Average Speed-up}}\\
        & \textbf{Model} & \textbf{Algorithm}
        & \textbf{$\bm{1} \rightarrow \bm{2}$ cores}
        & \textbf{$\bm{2} \rightarrow \bm{4}$ cores}
        & \textbf{$\bm{4} \rightarrow \bm{8}$ cores}
        & \textbf{$\bm{8} \rightarrow \bm{16}$ cores}\\
        \midrule
        \arrayrulecolor{black!00}\midrule
        \parbox[t]{4mm}{\multirow{2}{*}{\rotatebox[origin=c]{90}{\textbf{\shortstack[c]{BSBM}}}}}
        & \parbox[t]{4mm}{\multirow{2}{*}{\textbf{\shortstack[c]{AttrColl}}}}
        & \textbf{batch}
        & $1.25 \pm 0.20$ 
        & $1.75 \pm 0.27$ 
        & $1.24 \pm 0.05$ 
        & $0.98 \pm 0.03$\\
        & & \textbf{incremental} 
        & $1.55 \pm 0.09$
        & $1.40 \pm 0.06$
        & $0.99 \pm 0.02$
        & $0.91 \pm 0.03$\\
        \arrayrulecolor{black!0}\midrule
        \arrayrulecolor{black!50}\midrule
        \parbox[t]{4mm}{\multirow{6}{*}{\rotatebox[origin=c]{90}{\textbf{\shortstack[c]{DyLDO-core}}}}}
        & \parbox[t]{4mm}{\multirow{2}{*}{\textbf{\shortstack[c]{AttrColl}}}}
        & \textbf{batch}
        & $1.00 \pm 0.12$ 
        & $2.60 \pm 0.26$ 
        & $1.90 \pm 0.11$ 
        & $1.09 \pm 0.17$ \\
        & & \textbf{incremental}
        & $1.42 \pm 0.10$ 
        & $1.50 \pm 0.17$
        & $1.01 \pm 0.04$
        & $0.99 \pm 0.04$\\
        \arrayrulecolor{black!25}\cmidrule(r{1pt}){2-7}
        & \parbox[t]{4mm}{\multirow{2}{*}{\textbf{\shortstack[c]{ClassColl}}}}
        & \textbf{batch}
        & $1.36 \pm 0.40$ 
        & $2.55 \pm 0.64$ 
        & $2.15 \pm 0.37$ 
        & $1.64 \pm 0.13$ \\
        & & \textbf{incremental}
        & $1.54 \pm 0.08$ 
        & $1.44 \pm 0.09$
        & $1.03 \pm 0.02$
        & $1.07 \pm 0.03$\\
        \arrayrulecolor{black!25}\cmidrule(r{1pt}){2-7}
        & \parbox[t]{4mm}{\multirow{2}{*}{\textbf{\shortstack[c]{SchemEX}}}}
        & \textbf{batch}
        & $1.02 \pm 0.20$ 
        & $2.40 \pm 0.34$ 
        & $1.78 \pm 0.09$ 
        & $0.94 \pm 0.08$ \\
        & & \textbf{incremental}
        & $1.33 \pm 0.10$ 
        & $1.57 \pm 0.17$
        & $1.01 \pm 0.08$
        & $0.99 \pm 0.05$\\
        \arrayrulecolor{black!100}\bottomrule
    \end{tabularx}
}
\end{table*}

The mean speed-up for each increase in number of cores is presented in \cref{tab:parallelization-speed-up}.
On both datasets, we observe that the performance of incremental computation greatly improves when using up to $4$~cores, but increases only marginally or even decreases when adding more than $4$ cores.
For batch computation, we still see large performance improvements up to $8$~cores, in particular on the real-world DyLDO-core dataset. 
Class Collection is the only summary model we evaluated where performance continues to increase beyond $8$~cores.
On the DyLDO-core dataset, over all analyzed numbers of cores, we have an average speed-up of the incremental algorithm over the batch algorithm for Attribute Collection between $5.5$ and $13.3$ (mean $9.5 \pm 1.6$), for Class Collection between $51.5$ and $85.9$ (mean $64.3 \pm 7.7$), and for SchemEX between $2.0$ and $7.4$ (mean $6.3 \pm 1.1$).
Overall, on the DyLDO-core dataset, over all analyzed numbers of cores, we have a macro average speed-up between $3.0$ and $77.8$ (mean $19 \pm 3.3$) for the incremental algorithm.

\subsection{Discussion}
\label{sec:incremental-discussion-2}
The key insight from the experiment on the parallelization of our algorithm is that incremental computation outperforms batch computation for $1$, $2$, $4$, $8$, and $16$~cores. 
In addition, on the DyLDO-core dataset, the $16$-core batch computation is slower than the single-core incremental computation, except for SchemEX on DyLDO-core versions $46$ and~$47$.
This is remarkable, as we also observe that adding cores benefits batch computation more than incremental computation. 
To explain this, we evaluate the runtime of all phases during the computation.
We find that the find and merge phase takes a significant portion of runtime (see \cref{tab:find-and-merge-phase-multicore}).

Incremental computation is faster because it reduces the number of find and merge operations.
As described in \cref{sec:incremental-apparatus}, the find and merge phase is done in the graph database containing the graph summary. 
To write and update the graph summary, we use non-transactional, optimistic operations.
This reduces synchronization overhead and allows fast parallelization. 
Batch computation always writes the entire graph summary for each version of the input graph database.
This explains two observations:
First, using multiple cores allows parallel access to the graph summary, reducing the time required to find and merge vertex summaries. 
Second, since incremental computation avoids unnecessary find and merge operations, speeding up the find and merge phase has less impact on its overall runtime. 
It remains to be evaluated if different database implementations can significantly improve the overall runtime.
Both the batch and incremental algorithms will run faster with a faster database implementation, though the incremental algorithm makes fewer database accesses.

The make-set phase is the same for incremental and batch computation.
The parallelization of this step is done by Apache Spark, which also de-serializes the gzipped input files. 
Based on existing studies of graph processing frameworks~\cite{DBLP:journals/pvldb/AmmarO18}, we assume that for larger datasets and more complex graph summaries, \eg using the $k$-chaining parameterization, multi-core performance will scale beyond $4$ cores. 
Apache Spark is a state-of-the-art processing framework~\cite{DBLP:journals/pvldb/AmmarO18};
optimizing it is beyond the scope of this article. 

\begin{table*}[h!]
{\renewcommand\arraystretch{1.25}
    \centering
    \footnotesize
    \caption{\label{tab:find-and-merge-phase-multicore}The percentage of runtime due to the find and merge phase, in the batch and incremental algorithms.}
    \begin{tabularx}{\textwidth}{p{0.2cm} p{1.6cm} p{1.8cm} *{5}{>{\centering\arraybackslash}X}}
        \toprule
        & \textbf{Model} & \textbf{Algorithm} & \multicolumn{5}{c}{\textbf{Find and Merge Phase for Number of Cores}}\\
        & & 
        & \textbf{$\bm{1}$}
        & \textbf{$\bm{2}$}
        & \textbf{$\bm{4}$}
        & \textbf{$\bm{8}$}
        & \textbf{$\bm{16}$}\\
        \midrule
        \arrayrulecolor{black!0}\midrule
        \parbox[t]{4mm}{\multirow{2}{*}{\rotatebox[origin=c]{90}{\textbf{\shortstack[c]{BSBM}}}}}
        & \parbox[t]{4mm}{\multirow{2}{*}{\textbf{\shortstack[c]{AttrColl}}}}
        & \textbf{batch}
        & $36\% \pm \phantom{0}1\%$ 
        & $52\% \pm \phantom{0}8\%$ 
        & $50\% \pm \phantom{0}2\%$ 
        & $40\% \pm \phantom{0}2\%$
        & $42\% \pm \phantom{0}1\%$\\
        & & \textbf{incremental} 
        & $\phantom{0}7\% \pm \phantom{0}1\%$
        & $13\% \pm \phantom{0}1\%$
        & $28\% \pm \phantom{0}1\%$
        & $30\% \pm \phantom{0}1\%$
        & $38\% \pm \phantom{0}2\%$\\
        \arrayrulecolor{black!0}\midrule
        \arrayrulecolor{black!50}\midrule
        \parbox[t]{4mm}{\multirow{6}{*}{\rotatebox[origin=c]{90}{\textbf{\shortstack[c]{DyLDO-core}}}}}
        & \parbox[t]{4mm}{\multirow{2}{*}{\textbf{\shortstack[c]{AttrColl}}}}
        & \textbf{batch}
        & $91\% \pm \phantom{0}1\%$ 
        & $94\% \pm \phantom{0}1\%$ 
        & $91\% \pm \phantom{0}1\%$ 
        & $83\% \pm \phantom{0}1\%$
        & $83\% \pm \phantom{0}2\%$ \\
        & & \textbf{incremental}
        & $17\% \pm 12\%$ 
        & $26\% \pm 13\%$
        & $33\% \pm \phantom{0}9\%$
        & $31\% \pm \phantom{0}8\%$
        & $38\% \pm \phantom{0}8\%$\\
        \arrayrulecolor{black!25}\cmidrule(r{1pt}){2-8}
        & \parbox[t]{4mm}{\multirow{2}{*}{\textbf{\shortstack[c]{ClassColl}}}}
        & \textbf{batch}
        & $99\% \pm \phantom{0}1\%$ 
        & $99\% \pm \phantom{0}1\%$ 
        & $98\% \pm \phantom{0}1\%$ 
        & $96\% \pm \phantom{0}1\%$
        & $95\% \pm \phantom{0}1\%$ \\
        & & \textbf{incremental}
        & $\phantom{0}6\% \pm \phantom{0}1\%$ 
        & $12\% \pm \phantom{0}1\%$
        & $21\% \pm \phantom{0}1\%$
        & $23\% \pm \phantom{0}1\%$
        & $28\% \pm \phantom{0}1\%$\\
        \arrayrulecolor{black!25}\cmidrule(r{1pt}){2-8}
        & \parbox[t]{4mm}{\multirow{2}{*}{\textbf{\shortstack[c]{SchemEX}}}}
        & \textbf{batch}
        & $88\% \pm \phantom{0}1\%$ 
        & $92\% \pm \phantom{0}2\%$ 
        & $89\% \pm \phantom{0}1\%$ 
        & $79\% \pm \phantom{0}2\%$
        & $82\% \pm \phantom{0}2\%$ \\
        & & \textbf{incremental}
        & $25\% \pm 15\%$ 
        & $34\% \pm 14\%$
        & $38\% \pm 11\%$
        & $35\% \pm 10\%$
        & $43\% \pm 10\%$\\
        \arrayrulecolor{black!100}\bottomrule
    \end{tabularx}
}
\end{table*}

\begin{figure*}[t!]
\centering
\begin{subfigure}[t]{1\linewidth}
    \centering
    \includegraphics[trim={0cm 0cm 0cm 0cm},clip,width=\textwidth]{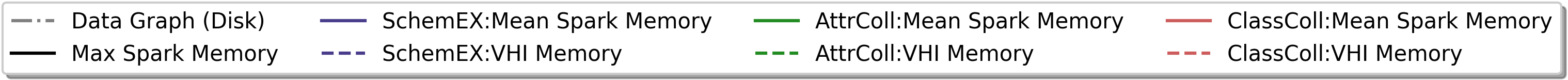}
\end{subfigure}

\begin{subfigure}[t]{0.48\linewidth}
    \centering
    \includegraphics[trim={0cm 0cm 0cm 0cm},clip,width=.95\textwidth]{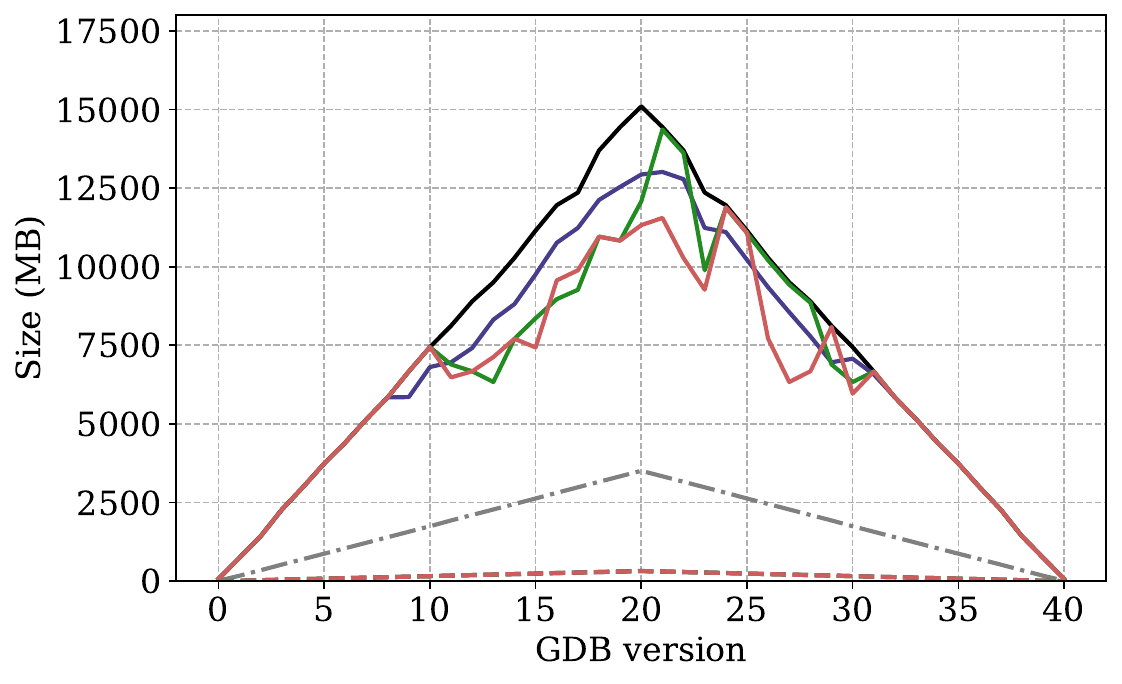}
     \caption{\label{fig:BSBM-attribute-memory} Memory consumption on the BSBM dataset.}
\end{subfigure}
\quad
\begin{subfigure}[t]{0.48\linewidth}
    \centering
    \includegraphics[trim={0cm 0cm 0cm 0cm},clip,width=.95\textwidth]{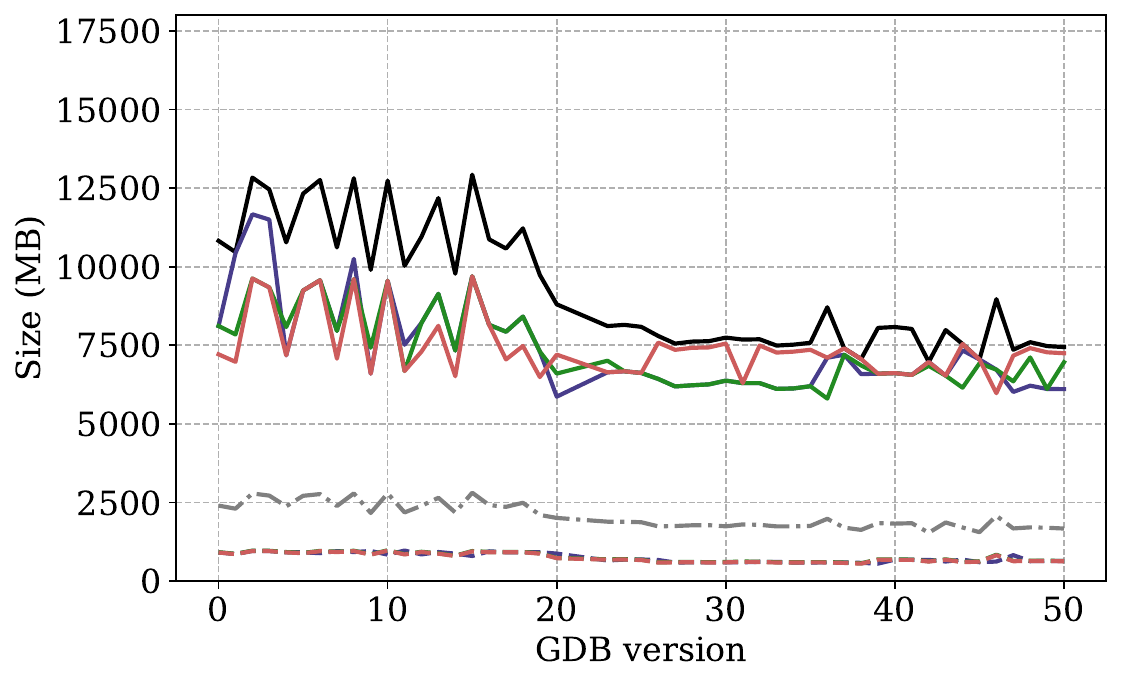}
     \caption{\label{fig:DyLDO-core-attribute-memory} Memory consumption on the DyLDO-core dataset.}
\end{subfigure}

\caption{\label{fig:memory-results}Measurements of the memory consumption on the BSBM (left) and DyLDO-core (right) datasets.
The dash-dotted grey lines depict the raw dataset size as RDF $n$-triple text file. 
The black lines depict the maximum allocated memory of Spark's GraphX in-memory representation
engine for the input graph and all vertex summaries of the graph summary before the find and merge phase (equal for all three summary models on each dataset). 
The colored solid lines depict the mean allocated memory of Spark for the respective summary models and the colored dotted lines the memory allocated by the \vertexHashIndex{} (VHI) after the graph summary is computed.}
\end{figure*}

\section{Experiment 3: Memory Overhead Induced by Incremental Summarization}
\label{sec:experiment-3}
We measure the memory consumption of our algorithms and the memory overhead induced by the \vertexHashIndex{} used by the incremental algorithm.

\subsection{Procedure}
To analyze the memory consumption of our algorithm, we log the memory consumption of Apache Spark.
More precisely, we enable Spark's metrics log to estimate the size of GraphX's in-memory representation of the input graph and the graph summary.
Furthermore, we use a Java Instrumentation Agent to log the actual in-memory size of the \vertexHashIndex{}.
To this end, we log the sum of all shallow sizes of each referenced object in the \vertexHashIndex{}.
To exclude the possibility of negative impacts on the runtime, we logged the memory consumption in a separate set of experiments from Experiment~1 and~2. 
Analogously to Experiment~2, we use the BSBM and DyLDO-core datasets.
Based on the discussion in \cref{sec:incremental-discussion-2}, we use $4$~cores for all runs. 
Partitioning and distributing graph data during processing potentially creates data overheads due to data redundancy. 
Since Spark does not use more than $4$ partitions for our datasets, adding more cores does not benefit the performance of the make-set phase (Phase~1).
Reducing the number of cores means unnecessarily long runtimes.
Thus, $4$ cores is a reasonable choice for this experiment. 

\subsection{Results}
The results of this experiment are shown in \cref{fig:memory-results}.
Notice that while the memory consumption of the summary computations is the memory usage of the complete run logged by Spark (labeled Spark Memory), the size of the final \vertexHashIndex{} (labeled VHI) is calculated via a Java Instrumentation Agent.

Our main observation in this experiment is that the memory consumption scales linearly with the input size for all three summary models. 
Furthermore, the memory overhead due to the \vertexHashIndex{} is only a small fraction of the overall memory consumption.
For the BSBM dataset, we find a memory overhead of about $2\%$ ($\pm 0.06$).
For the DyLDO-core dataset, it is between $6\%$ and $11\%$ (mean $8\% \pm 0.68$).

Another interesting observation in this experiment is that, for each dataset, all three summary models have exactly the same maximum memory allocation for the summary computation (\cref{fig:memory-results}, black lines).  
Only the mean memory allocated by Spark differs between the summary models (\cref{fig:memory-results}, colored solid lines). 
We also see that computing graph summaries on the synthetic BSBM benchmark dataset and the real-world DyLDO-core dataset requires a comparable amount of memory. 
Comparing the maximum allocated memory, for BSBM we need to allocate on average $4.23$ ($\pm 0.06$) times more memory than to store the raw dataset on disk. 
For DyLDO-core, this factor averages $4.43$ ($\pm 0.11$).

Second, for each dataset, the size of the \vertexHashIndex{} of each graph summary model differs only marginally.
On the BSBM dataset, the \vertexHashIndex{} of SchemEX is between $2$ and $326\,$MB (mean $160 \pm 97\,$MB), the one of Attribute Collection is between $2$ and $324\,$MB (mean $159 \pm 97\,$MB), and the one of Class Collection is between $2$ and $324\,$MB (mean $159 \pm 97\,$MB).
On the DyLDO-core dataset, the \vertexHashIndex{} of SchemEX is between $569$ and $984\,$MB (mean $759 \pm 145\,$MB), the one of Attribute Collection is between $561$ and $976\,$MB (mean $755 \pm 146\,$MB), and the one of Class Collection is between $559$ and $974\,$MB (mean $750 \pm 145\,$MB).

The results of Experiment~1 show a strong correlation between the \vertexHashIndex{} size and the number of vertices in the GDB ($\rho(311)=0.908$, $p < .0001$).
L2 of the \vertexHashIndex{} dominates the overall size.
The number of edges is not as strongly correlated with the size of the \vertexHashIndex{} ($\rho(311)=0.661$, $p < .0001$), since edges are reflected by the vertex summaries.
 
The difference in size of the \vertexHashIndex{} between the summary models on a fixed dataset is only a few megabytes, but the difference between the datasets is much larger.
The largest version of the BSBM dataset is the $20$th, with $1,286,066$ vertices. This is comparable to the smallest version of the DyLDO-core dataset (version $36$), with $1,503,466$ vertices.
For all three summary models, the \vertexHashIndex{} in the $36$th version of the DyLDO-core dataset is about twice as large as the \vertexHashIndex{} in the $20$th version of the BSBM dataset. 
On average, the size of the \vertexHashIndex{} in the $36$th version of the DyLDO-core dataset is $602$ ($\pm 5$)\,MB.
The \vertexHashIndex{} in the $20$th version of the BSBM dataset requires $324$ ($\pm 1$)\,MB of memory.

\subsection{Discussion}

From the experiment on the memory consumption of the \vertexHashIndex{}, the memory overhead required for maintaining this additional data structure is only a small fraction of the memory consumption of the graph summary computation.
In particular, considering the significant performance benefits found in Experiments~1 and~2 for the incremental algorithm, the memory overhead induced by the \vertexHashIndex{} of less than $10\%$ on the real-world dataset is reasonable. 
On the synthetic benchmark dataset, this overhead is only $2\%$.
Since the maximum Spark memory consumption for the graph summary computation is roughly the same, the main varying factor is the absolute size of the \vertexHashIndex{}.

The size of the \vertexHashIndex{}, \ie the memory overhead for the incremental graph summarization, is mostly influenced by the number of vertices in the input graph.
A secondary factor is the number of primary vertices in the computed graph summary, \ie how many different schema structures appear in the input dataset.
We noticed a nearly two-fold increase of  \vertexHashIndex{} size for the real-world DyLDO-core dataset over the synthetic BSBM dataset in two versions of each dataset, where the difference in number of vertices is less than $15\%$.
References to each vertex are stored in L2.
The payload (L3) and/or the references to vertex summaries (L1) make up the remainder.
In the final graph summary, L1~is orders of magnitude larger on the DyLDO-core dataset than on the BSBM dataset.
For example, the Attribute Collection in the $20$th version of the BSBM dataset contains about $10^2$ vertices, while it contains about $2 \times 10^4$ vertices in the $36$th version of the DyLDO-core dataset (see Experiment~1, \cref{fig:main-results-benchmark}).
The \vertexHashIndex{} stores a reference to each primary vertex in the graph summary so its grows with the number of primary vertices in the computed graph summary.
Regarding L3, note that we store the data source payload for the DyLDO-core dataset and the vertex count payload for the BSBM dataset, which also contributes to an overall larger \vertexHashIndex{} on the DyLDO-core dataset.
However, we evaluated the impact of the data source payload by measuring the size of the \vertexHashIndex{} excluding the payload information.
On average, the sizes of the \vertexHashIndex{} decreased on both datasets by about $25\%$ to $30\%$, which still results in the \vertexHashIndex{} being about twice as big for the DyLDO-core dataset.

Another observation is that, in our experiments, we needed about $4$ to $5$ times more memory than the raw dataset size.
During the graph summary computation, Spark stores the original graph and the corresponding vertex summary for each vertex in memory.
It remains to be evaluated if a different graph processing framework can further compress this information to reduce the overall memory consumption, \eg Bogel~\cite{DBLP:journals/pvldb/YanCLN14,DBLP:journals/pvldb/AmmarO18}.

\section{Experiment 4: Efficiency of Computing Long $k$-chains using a Hash-Messages}
\label{sec:experiment-4}
We evaluate the running time and memory consumption of the hash-messaging algorithm (\cref{alg:brs-adjusted}) on datasets with between $10\,$M to $1\,$B triples.
We use the largest synthetic and real-world datasets from the previous datasets and additional datasets that are two orders of magnitude larger.

\subsection{Procedure}

We experiment in total with five datasets, two smaller graphs and three larger graphs.
For the smaller graphs, we use the largest versions of the graph databases from the previous experiments as described in Section~\ref{sec:datasets}.
The largest version of LUBM100 is the first version, which has about $3.3\,$M vertices and $13\,$M edges.\footnote{The name LUBM100 stems from the choice to generate a graph for $100$ universities, rather than indicating the size of the dataset. 
The latter is done by the appendices M for millions and B for billions.}
The largest version of the BSBM dataset is version~$20$, which has about $3\,$M vertices and $14\,$M edges.

Our first larger dataset is the first version of the DyLDO-ext dataset, which is the largest of the weekly crawls with about $13\,$M vertices and $115\,$M edges.
We also generate a larger version of the BSBM dataset.
This is referred to as the \textit{$BSBM100M$ dataset} in this experiment and comprises about $17\,$M vertices and $90\,$M edges.
Furthermore, to demonstrate scalability of our algorithm, we go another magnitude larger and generate a variant of the BSBM dataset with $1\,$B triples, which we refer to as \textit{$BSBM1B$}.
This largest dataset contains about $172\,$M vertices and $941\,$M edges.

For our experiment, we compute the $k=1, \ldots, 10$ forward bisimulation defined by the the graph summary model from Schätzle \etal~\cite{DBLP:conf/sigmod/SchatzleNLP13}, as defined in Section~\ref{sec:iterative-bisimulation} \cref{gsm:schaetzle}.
We run the algorithm six times consecutively while maintaining exclusive execution rights on the server machine.
The first run is used as a warm up and we exclude it from the measurements.
This is to reduce potential bias such as filling up caches and eliminating potential side effects from other processes.
We measure the results from the last five runs and report both the numbers per run and per $k$th-iteration, as well as providing summary statistics.
We measure running times for the initialization (see lines \ref{algo:initialization}--\ref{algo:initialize-id-object} of Algorithm~\ref{alg:brs-adjusted}) and for each individual iteration.
We also report the total running time and the maximum memory consumption.
The latter is reported by the maximum memory usage (in GB) of the JVM Memory OnHeap observed during the algorithm execution.

\subsection{Results}

First, we report the results for the smaller datasets, LUBM100 and BSBM.
These graphs are of similar size, and the execution of the bisimulation is similar.
The running time for initialization is $6$~seconds and the execution times for each of the $k$~iterations is $6$~seconds.
Thus, the total execution time is $66$ seconds for both the LUBM100 and the BSBM datasets.
We omit detailed tables of the running times of these smaller datasets.

\begin{comment}
    \subfloat[Running times in minutes on the LUBM100 dataset for $k=1,\ldots,10$ bisimulation with $CSE_{\text{Schätzle}}$.]{
    \label{table-exp-lubm100}
        \begin{tabular}{|*{15}{c|}}
        \hline
            \textbf{Run} & \multicolumn{11}{c|}{\textbf{Iteration}} & \multicolumn{3}{c|}{\textbf{Aggregates}} \\
            \hline
            \multicolumn{1}{|c|}{} & 0 & 1 & 2 & 3 & 4 & 5 & 6 & 7 & 8 & 9 & 10 & \textbf{$\Sigma$} & \textbf{$\mu$} & \textbf{$\sigma$} \\
            \hline
            1 & 0.1 & 0.1 & 0.1 & 0.1 & 0.1 & 0.1 & 0.1 & 0.1 & 0.1 & 0.1 & 0.1 & 1.1 & 0.1 & 0.0 \\ 
2 & 0.1 & 0.1 & 0.1 & 0.1 & 0.1 & 0.1 & 0.1 & 0.1 & 0.1 & 0.1 & 0.1 & 1.1 & 0.1 & 0.0 \\ 
3 & 0.1 & 0.1 & 0.1 & 0.1 & 0.1 & 0.1 & 0.1 & 0.1 & 0.1 & 0.1 & 0.1 & 1.1 & 0.1 & 0.0 \\ 
4 & 0.1 & 0.1 & 0.1 & 0.1 & 0.1 & 0.1 & 0.1 & 0.1 & 0.1 & 0.1 & 0.1 & 1.1 & 0.1 & 0.0 \\ 
5 & 0.1 & 0.1 & 0.1 & 0.1 & 0.1 & 0.1 & 0.1 & 0.1 & 0.1 & 0.1 & 0.1 & 1.1 & 0.1 & 0.0 \\ 
\hline
\textbf{$\mu$} & 0.1 & 0.1 & 0.1 & 0.1 & 0.1 & 0.1 & 0.1 & 0.1 & 0.1 & 0.1 & 0.1 & \textbf{1.1}
& \multicolumn{2}{c|}{} \\
\hline
\textbf{$\sigma$} & 0.0 & 0.0 & 0.0 & 0.0 & 0.0 & 0.0 & 0.0 & 0.0 & 0.0 & 0.0 & 0.0 & \textbf{0.0} 
& \multicolumn{2}{c|}{} \\
\hline
        \end{tabular}
        
  \vspace{3mm}
    }
\end{comment}

Since the smaller datasets give only limited insights into the properties of our algorithm, we apply it to three larger datasets.
The first version of the the web-crawled DyLDO-ext dataset is an order of magnitude larger than LUBM100 and BSBM.
But it also differs in terms of being a real-world dataset rather then a synthetically generated dataset.
The running times of our hash-message $10$-bisimulation on the largest DyLDO-ext version are reported in Table~\ref{table-exp-dyldo}, where the initialization of the algorithm takes $0.3$~minutes and executing the bisimulation iterations needs between $0.16$ and $0.46$~minutes each, on average.
Table~\ref{table-exp-bsbm100m} shows the results for the synthetic BSBM100M dataset, where initialization takes $1$~minute and each iteration requires $0.5$ to $0.8$~minutes, on average.
The total processing time to reach a $10$-bisimulation is under $7$~minutes, including initialization.
Furthermore, it can be observed that the execution times per run (rows) as well as iterations (columns) are quite stable.
Finally, for the largest dataset, BSBM1B, the initialization step takes $10.2$~minutes and each iteration requires between $4.14$ to $5.04$~minutes, on average, as reported in Table~\ref{table-exp-bsbm1b}.
The computation of a $10$-bisimulation on the $1$~billion edges dataset requires less then $55$~minutes, including initialization.
The standard deviation between the runs (rows) is considerably larger for BSBM1B than for BSBM100M, which is a order of magnitude smaller.
Within the runs, the standard deviation is between $0.28$ and $1.02$~minutes.
Notably, the standard deviation within a run, \ie across the ten iterations is still low with a maximum of $0.51$~minutes for the fourth run.

\begin{table*}[t!]
    \centering
    \subfloat[
    Running times in minutes on the DyLDO-ext dataset for $k=1,\ldots,10$ bisimulation with $CSE_{\text{Schätzle}}$. ]{
    \label{table-exp-dyldo}
    
            \begin{tabular}{|*{15}{c|}}
        \hline
            \textbf{Run} & \multicolumn{11}{c|}{\textbf{Iteration}} & \multicolumn{3}{c|}{\textbf{Aggregates}} \\
            \hline
            \multicolumn{1}{|c|}{} & 0 & 1 & 2 & 3 & 4 & 5 & 6 & 7 & 8 & 9 & 10 & \textbf{$\Sigma$} & \textbf{$\mu$} & \textbf{$\sigma$} \\
            \hline 
1 & 0.3 & 0.3 & 0.4 & 0.3 & 0.4 & 0.5 & 0.5 & 0.4 & 0.4 & 0.4 & 0.4 & 4.3 & 0.4 & 0.06 \\ 
2 & 0.3 & 0.1 & 0.2 & 0.2 & 0.2 & 0.2 & 0.2 & 0.2 & 0.2 & 0.2 & 0.3 & 2.3 & 0.2 & 0.04 \\ 
3 & 0.3 & 0.1 & 0.2 & 0.2 & 0.2 & 0.2 & 0.2 & 0.2 & 0.2 & 0.2 & 0.5 & 2.5 & 0.22 & 0.10 \\ 
4 & 0.3 & 0.1 & 0.2 & 0.2 & 0.2 & 0.2 & 0.2 & 0.2 & 0.2 & 0.2 & 0.4 & 2.4 & 0.21 & 0.07 \\ 
5 & 0.3 & 0.2 & 0.2 & 0.2 & 0.2 & 0.2 & 0.2 & 0.2 & 0.2 & 0.2 & 0.7 & 2.8 & 0.25 & 0.15 \\ 
\hline
\textbf{$\mu$} & 0.3 & 0.16 & 0.24 & 0.22 & 0.24 & 0.26 & 0.26 & 0.24 & 0.24 & 0.24 & 0.46 & \textbf{2.86} & \multicolumn{2}{c|}{} \\ 
\hline
\textbf{$\sigma$} & 0.0 & 0.08 & 0.08 & 0.04 & 0.08 & 0.12 & 0.12 & 0.08 & 0.08 & 0.08 & 0.14 & \textbf{0.74} & \multicolumn{2}{c|}{} \\ 
        \hline
        \end{tabular}
        
   \vspace{3mm}
    }

\subfloat[
    Running times in minutes on the BSBM100M dataset for $k=1,\ldots,10$ bisimulation with $CSE_{\text{Schätzle}}$. 
    ]{
    \label{table-exp-bsbm100m}
        \begin{tabular}{|*{15}{c|}}
        \hline
            \textbf{Run} & \multicolumn{11}{c|}{\textbf{Iteration}} & \multicolumn{3}{c|}{\textbf{Aggregates}} \\
            \hline
            \multicolumn{1}{|c|}{} & 0 & 1 & 2 & 3 & 4 & 5 & 6 & 7 & 8 & 9 & 10 & \textbf{$\Sigma$} & \textbf{$\mu$} & \textbf{$\sigma$} \\
            \hline 
            1 & 1.0 & 0.5 & 0.5 & 0.5 & 0.6 & 0.8 & 0.5 & 0.5 & 0.5 & 0.5 & 0.4 & 6.3 & 0.53 & 0.10 \\ 
            2 & 1.0 & 0.5 & 0.6 & 0.5 & 0.5 & 0.9 & 0.5 & 0.5 & 0.5 & 0.5 & 0.5 & 6.5 & 0.55 & 0.12 \\ 
            3 & 1.0 & 0.5 & 0.6 & 0.5 & 0.5 & 0.9 & 0.5 & 0.5 & 0.5 & 0.5 & 0.5 & 6.5 & 0.55 & 0.12 \\ 
            4 & 1.0 & 0.5 & 0.5 & 0.5 & 0.6 & 0.9 & 0.5 & 0.5 & 0.5 & 0.5 & 0.5 & 6.5 & 0.55 & 0.12 \\ 
            5 & 1.0 & 0.5 & 0.5 & 0.5 & 0.9 & 0.5 & 0.5 & 0.5 & 0.5 & 0.6 & 0.5 & 6.5 & 0.55 & 0.12 \\ 
            \hline
            \textbf{$\mu$} & 1.0 & 0.5 & 0.54 & 0.5 & 0.62 & 0.8 & 0.5 & 0.5 & 0.5 & 0.52 & 0.48 & \textbf{6.46} & \multicolumn{2}{c|}{} \\
            \hline
            \textbf{$\sigma$} & 0.0 & 0.0 & 0.05 & 0.0 & 0.15 & 0.15 & 0.0 & 0.0 & 0.0 & 0.04 & 0.04 & \textbf{0.08} & \multicolumn{2}{c|}{} \\
        \hline
        \end{tabular}
        \vspace{3mm}
    }

    \subfloat[
    Running times in minutes on the BSBM1B dataset for $k=1,\ldots,10$ bisimulation with $CSE_{\text{Schätzle}}$. ]{
    \label{table-exp-bsbm1b}
        \begin{tabular}{|*{15}{c|}}
        \hline
            \textbf{Run} & \multicolumn{11}{c|}{\textbf{Iteration}} & \multicolumn{3}{c|}{\textbf{Aggregates}} \\
            \hline
            \multicolumn{1}{|c|}{} & 0 & 1 & 2 & 3 & 4 & 5 & 6 & 7 & 8 & 9 & 10 & \textbf{$\Sigma$} & \textbf{$\mu$} & \textbf{$\sigma$} \\
            \hline 
            1 & 10.1 & 3.9 & 4.7 & 4.3 & 4.3 & 4.5 & 5.3 & 4.7 & 5.1 & 4.7 & 5.7 & 57.3 & 4.72 & 0.50 \\ 
            2 & 10.1 & 2.9 & 3.7 & 3.3 & 3.4 & 3.3 & 4.1 & 4.1 & 4.1 & 3.8 & 3.1 & 45.9 & 3.58 & 0.42 \\ 
            3 & 10.0 & 3.9 & 4.8 & 4.3 & 4.4 & 4.5 & 5.2 & 4.7 & 4.9 & 4.5 & 4.2 & 55.4 & 4.54 & 0.36 \\ 
            4 & 10.1 & 3.9 & 4.8 & 4.4 & 4.4 & 4.5 & 5.2 & 4.9 & 5.2 & 4.5 & 5.8 & 57.7 & 4.76 & 0.51 \\ 
            5 & 10.0 & 3.9 & 4.8 & 4.4 & 4.4 & 4.6 & 5.4 & 4.8 & 5.0 & 4.4 & 4.2 & 55.9 & 4.59 & 0.41 \\ 
            \hline
            \textbf{$\mu$} & 10.2 & 3.7 & 4.56 & 4.14 & 4.18 & 4.28 & 5.04 & 4.64 & 4.86 & 4.38 & 4.6 & \textbf{54.44} & \multicolumn{2}{c|}{}\\
            \hline
            \textbf{$\sigma$} & 0.62 & 0.4 & 0.43 & 0.42 & 0.39 & 0.49 & 0.48 & 0.28 & 0.39 & 0.31 & 1.02 & \textbf{4.35} & \multicolumn{2}{c|}{}\\
        \hline
        \end{tabular}
    }

    \caption{Running times of $k=1,\ldots,10$ iterative bisimulation on the three large datasets DyLDO-ext, BSBM100, and BSBM1B 
    using the graph summary model from Schätzle \etal~\cite{DBLP:conf/sigmod/SchatzleNLP13}.
    Time is measured in minutes.
    Iteration 0 is the initialization phase (Lines \ref{algo:initialization}--\ref{algo:initialize-id-object}) of the iterative bisimulation algorithm, Algorithm~\ref{alg:brs-adjusted}.
    }
    \label{tab:detailed-results-bsbm1b}
\end{table*}

The maximum JVM OnHeap Memory for  $10$-bisimulation for the smaller LUBM100 and BSBM datasets is $104.1\,$GB and $103.8\,$GB, respectively.
The maximum memory usage for the largest version of DyLDO-ext is $126.3\,$GB.
For the BSBM100M dataset, the maximum memory consumption is $248.6\,$GB, while for the largest dataset BSBM1B the average maximum memory usage is $1248.2\,$GB.

\subsection{Discussion}

Our parallel algorithm for iterative $k$-bisimulation efficiently processes large graphs, both synthetically generated and real-world graphs crawled from the web.
Remarkably, the execution times on the smaller datasets LUBM100 and BSBM with about $14\,$M edges is within seconds per bisimulation iteration.
To demonstrate the scalability of the algorithm, we applied it to larger datasets.
The $1$~billion edges of the BSBM1B dataset are processed within $54.44$~minutes.
This is a factor of $8.42$ longer than the BSBM100M dataset, which is one order of magnitude smaller with $100$~million edges.
Similarly, the maximum memory consumption for the $1$~billion-triple dataset BSBM1B is $1248.23\,$GB only about five times higher than the maximum use of $248.6\,$GB for the smaller BSBM100M dataset.
This indicates the scalability of our algorithm, both in running time as well as memory consumption.

Notably, processing the DyLDO-ext dataset is faster and requires less memory than BSBM100M dataset.
The reason for this is twofold: 
First, although the DyLDO-ext dataset has with $115\,$M edges (about $25\,$M edges more than the BSBM100M dataset), there are fewer vertices to be summarized.
The DyLDO-ext dataset has $13\,$M vertices and thus about $4\,$M less than the BSBM100 dataset.
Second, the real-world DyLDO-ext dataset is much more heterogeneous than the BSBM100M dataset.
Thus, one finds on average smaller bisimilar graphs in the DyLDO-ext dataset compared to the cleaner BSBM100M dataset, which reduces the number of comparisons needed. 
To support this statement, we further analyzed the datasets and found that the DyLDO-ext dataset has in total $30,779$ vertex labels, which occur in $48,918$ unique label sets, \ie different combinations of vertex labels.
This variety in the data model comes from the nature of the web-sourced dataset.
In contrast, the BSBM100M dataset is synthetically generated and is comprised of $1,289$ vertex labels, which are combined in only $2,274$ label sets.

Another observation is the higher standard deviation of the running times between the five runs of the largest BSBM1B dataset.
This may be a result from different caching strategies and potentially accessing secondary storage.
However, our analyses were all executed on the same system, which we had exclusive use of during the experiments.
Furthermore, the maximum memory consumption for computing the bisimulation on the BSBM1B dataset shows that at no point more than $1.3$ TB of the available $2$ TB RAM of the system were used.
Thus, we can rule out this threat to validity.

\section{Related Work}
\label{sec:related-work}
\label{sec:related-work-incremental}

We discuss the related work in the area of incremental structural graph summarization.
Related to this are algorithms for incremental subgraph indices and incremental schema discovery.
There are also extensive surveys in the field of graph summarization, which group and classify the different summarization approaches~\cite{DBLP:journals/csur/LiuSDK18,DBLP:journals/corr/abs-2004-14794,DBLP:journals/pvldb/KhanBB17,DBLP:journals/vldb/CebiricGKKMTZ19}.
Finally, the recent survey of Kellou-Menouer \etal on semantic schema discovery~\cite{Kellou-MenouerEtAl-VLDB-2021} discusses several schema discovery approaches: those that make no explicit statements about the schema in the dataset, those that use explicit statements about the data schema, and those that discover structural patterns from the data. 
The last is related to graph summarization. 
Finally, the use of bisimulation in graph summarization is discussed by  \v{C}ebiri\'c \etal{}~\cite[Section~5.1.1]{DBLP:journals/vldb/CebiricGKKMTZ19}.

Below, we discuss the key papers in the field and relate them to our work.
We focus on works on incremental summarization, subgraph indices, incremental schema discovery and bisimulation.

\paragraph{Incremental Structural Graph Summarization.}

There are a large variety of structural graph summary models~\cite{DBLP:journals/csur/LiuSDK18,DBLP:journals/corr/abs-2004-14794,DBLP:journals/pvldb/KhanBB17,DBLP:journals/vldb/CebiricGKKMTZ19,DBLP:journals/tcs/BlumeRS21}.
The problems with this plethora of existing summary models are that each model defines its own data structure that is designed for solving only a specific task and the models are evaluated using different metrics~\cite{DBLP:conf/www/StefanoniMK18,DBLP:conf/icde/NeumannM11,lodex2015,loupe2015,DBLP:conf/semWeb/PietrigaGADCGM18,DBLP:conf/esws/SpahiuPPRM16a,DBLP:conf/esws/SchaibleGS16,DBLP:conf/www/CiglanNH12,DBLP:conf/kcap/GottronSKP13}.
Furthermore, existing structural graph summarization algorithms are often designed and/or evaluated using static graphs only~\cite{DBLP:journals/csur/LiuSDK18,DBLP:conf/icde/NeumannM11,lodex2015,DBLP:conf/esws/SpahiuPPRM16a,DBLP:conf/esws/SchaibleGS16,DBLP:conf/www/CiglanNH12}.

Few structural graph summaries are designed for evolving graphs~\cite{DBLP:journals/ws/KonrathGSS12,DBLP:conf/edbt/GoasdoueGM19}.
To the best of our knowledge, there is no solution that updates only those parts of the graph summary where there have been updates in the graph, without requiring a change log and without having to store the data graph locally, which loses the benefit of the graph summary being typically orders of magnitude smaller than the original graph while capturing a set of predefined features.
Konrath \etal~\cite{DBLP:journals/ws/KonrathGSS12} compute their graph summary over a stream of vertex-edge-vertex triples. They can deal with the addition of new vertices and edges to the graph but not deletion of vertices or edges, or modification of their labels.
Similarly, Goasdou{\'{e}} \etal~\cite{DBLP:conf/edbt/GoasdoueGM19} only support iterative additions of vertices and edges to their structural graph summaries, and do not handle deletions.
Thus, these approaches are not suited to updating structural summaries of evolving graphs.
Goasdou{\'{e}} \etal also do not support payload information, which is required for tasks such as cardinality estimations and data search.
The purpose of their summaries is to visualize them to a human viewer.

Finally, there are also works on graph summarization, which use a technique called ``corrections''~\cite{DBLP:conf/www/ShinG0R19,DBLP:conf/kdd/KoKS20} to provide a concise representation of the full graph.
Thus, these works are actually more related to lossless graph compression.
A correction is a set of edges that are to be added or removed from the compressed graph to reconstruct the original graph~\cite{DBLP:conf/www/ShinG0R19}.
Incremental variants of corrections for graph compression also exist, where the changes are stored as sets of corrections~\cite{DBLP:conf/kdd/KoKS20}.

\paragraph{Incremental Subgraph Indices.}
Besides structural graph summaries that abstract from a graph based on common vertex or edge features, there are also general purpose graph database indices.
Commonly, graph databases use path indices, tree indices, and subgraph indices~\cite{DBLP:journals/pvldb/HanLPY10}.
A seminal approach on computing subgraph indices is DataGuide~\cite{DBLP:conf/vldb/GoldmanW97}, implemented in the Lore database management system (DBMS)~\cite{DBLP:journals/debu/Widom99}. 
A DataGuide is a graph index build on-the-fly while executing queries on an XML database.
It indicates to the query engine if and how a specific path defined in the query can be reached or not.
To this end, a DataGuide represents all possible paths between two vertices in an XML file.
While DataGuides operate on semi-structured data in terms of XML trees, Tran et al.~\cite{DBLP:journals/tkde/TranLR13, DBLP:conf/icde/TranWRC09} took up the idea and applied it to RDF data.
Representative Objects (RO) by Nestorov et al.~\cite{DBLP:conf/icde/NestorovUWC97} takes up the ideas of DataGuides and are also implemented in the Lore DBMS with focus on path queries, query optimization, and schema discovery.
While the Full RO capture a description of the global structure of the graph, the authors also introduce a notion of $k$-RO.
The latter limits the considered length of paths to a maximum of $k$.

Below, we focus on incremental subgraph indices as they most closely relate to structural graph summaries.
Yuan \etal~\cite{DBLP:conf/icde/YuanMYG12} propose an index based on mining frequent and discriminative features in subgraphs.
Their algorithm minimizes the number of index lookups for a given query, and regroups subgraphs based on newly added features.
The runtime performance was improved by the same authors in 2015~\cite{DBLP:conf/sigmod/YuanMYG15} and by Kansal and Spezzano in 2017~\cite{DBLP:conf/cikm/KansalS17}.
However, mining frequent features in subgraphs only optimizes the index for lookup operations for commonly used queries.
It does not compute a comprehensive graph summary along specified structural features.
Qiao \etal~\cite{DBLP:journals/pvldb/QiaoZC17} propose an approach to compute an index of isomorphic subgraphs in an unlabeled, undirected graph~$G$.
The goal is to find the set of subgraphs in~$G$ that are isomorphic to a given query pattern.
The result is a compression of the original graph that is suitable to answer, \eg cardinality queries regarding subgraphs.
Their work does yet not support changes in the graph.
The algorithm of Fan~\etal~\cite{DBLP:conf/sigmod/FanLLTWW11} can deal with graph changes for the subgraph isomorphism problem.
Their incremental computation of an index for isomorphic subgraphs is closely related to structural graph summarization, but it differs in that the graph pattern~$p$ is an input to the algorithm, not the output.
The goal of structural graph summarization is to compute, based on a summary model that specifies the features, the set of common graph patterns that occur in a graph.
Min \etal~\cite{DBLP:journals/pvldb/MinPPGIH21} propose an algorithm for continuous subgraph matching using an auxiliary data structure, which stores the intermediate results between a query graph and a dynamic data graph.
The authors consider undirected graphs where only vertices are labeled. 
Dynamic graphs are updated through a sequence of edge insertions and edge deletions.
The recent TipTap~\cite{tiptap2021} algorithm computes approximations of the frequent $k$-vertex subgraphs w.r.t. a given threshold in evolving graphs.
TipTap's purpose is to count the occurrences of different subgraphs in large, evolving graphs.
The evolution of graphs is modeled as a stream of updates on an existing graph.
Tesseract~\cite{DBLP:conf/eurosys/BindschaedlerML21} is a distributed framework for executing general graph mining algorithms on evolving graphs implemented using Apache Spark Structured Streaming.
It follows a vertex-centric approach to distribute updates to different workers following the assumption that, in general, changes effect only local graphs. 
Thus, few duplicate updates need to be detected.
Tesseract supports $k$-clique enumeration, graph keyword search, motif counting, and frequent subgraph mining. 

Another area of incremental subgraph indices is to consider an evolving set of queries over a static data graph. 
Duong \etal~\cite{DBLP:journals/pvldb/DuongHYWNA21} propose a streaming algorithm using approximate pattern matching to determine subgraph isomorphisms. 
They employ $k$-bisimulation to determine equivalent subgraphs and store them in an index.
However, this index is computed offline for a static graph only and their algorithm considers a stream of graph queries as input. 

\paragraph{Incremental Schema Discovery.}
Another area related to our work is incremental schema discovery from large datasets in NoSQL databases.
Wang \etal~\cite{DBLP:journals/pvldb/WangHZS15} propose an approach for incremental discovery of attribute-based schemata from JSON documents.
The schema is stored in a tree-like data structure, following the nesting of JSON objects.
The attributes are considered per document only, with no cross-document connections of attributes.
The algorithm incrementally adds information to the schema as more documents are processed.
Baazizi \etal~\cite{DBLP:conf/edbt/BaaziziLCGS17} also compute schemata from JSON objects, with focus on optional and mandatory fields.

In addition to document oriented formats like JSON, schema discovery is also used for graph data.
For example, XStruct~\cite{DBLP:conf/icde/HegewaldNW06} follows a heuristic approach to incrementally extract the XML schema of XML documents.
However, such schema discovery approaches cannot deal with modifications or deletions of nodes in the XML tree.
Other schema discovery approaches focus on generating (probabilistic) dataset descriptions.
Kellou-Menouer and Kedad~\cite{DBLP:conf/er/Kellou-MenouerK15} apply density-based hierarchical clustering on vertex and edge labels in a graph database.
This computes profiles that can be used to visualize the schema of the graph.
The term \enquote{incremental} for the related schema discovery algorithms refers in their work to the concept of incrementally processing large documents, not considering modifications and deletions.
Thus, they are not designed to update the discovered schema for evolving graphs.
Recently, Bouhamoum \etal~\cite{BouhamoumESWC2021} proposed to use density-based clustering to extract schema information from an RDF graph and to incrementally update the schema when new RDF instances arrive.
While the work can deal with additions, the deletion of edges and vertices is not considered.

\paragraph{Bisimulation for Graph Summarization}

Our formal language, FLUID, can be used to define graph summaries based on bisimulation as we have show in \cref{sec:iterative-bisimulation}. 
Generally, bisimulation deals with the computation of equivalence relations on states in labeled transition systems~\cite{Bisimulation:2009}.
States are considered equivalent or bisimilar, if equivalent states have transitions of the same types, \ie to equivalent states.
Labeled transition systems can be interpreted as edge-labeled graphs.
One distinguishes a forward and backward bisimulation, where the former is defined using outgoing edges and the latter is defined using incoming edges~\cite{DBLP:journals/vldb/CebiricGKKMTZ19,DBLP:journals/tcs/BlumeRS21}.
A stratified bisimulation, or $k$-bisimulation relaxes the former notion of bisimulation, such that only the neighbours up to distance~$k$ have to be bisimilar as well.

A notion of $k$-bisimulation w.r.t.\@ graph indices is introduced by seminal works such as the $k$-RO index by Nestorov \etal~\cite{DBLP:conf/icde/NestorovUWC97} and the T-indexes~\cite{DBLP:conf/icdt/MiloS99} by Milo and Suciu.
Furthermore, in the existing research, one can observe that $k$-bisimulation is a popular feature for structural graph summarization~\cite{DBLP:journals/tcs/BlumeRS21}.  
The following works make use of $k$-bisimulation in their approach to structural graph summarization:
Buneman et al. make use of forward $k$-bisimulation in the problem of RDF graph alignment \cite{DBLP:journals/pvldb/BunemanS16}.
Summarizing the union of two consecutive versions $G_\text{union} = G_1 \cup G_2$ of an RDF graph with respect to $k$-bisimulation, puts vertices to be aligned in the same partition.
Additionally to $k$-bisimulation, they use a similarity measure to further refine the initial $k$-bisimulation partition, as it does not capture all vertices to be aligned.
The focus of their work is the optimization of the alignment process, so that every node pair $(v_1, v_2)$, with $v_1 \in G_1$ and $v_2 \in G_2$, which have to be aligned is identified and not the construction of a $k$-bisimulation-based partition of $G$.
Combining forward- and backward-bisimulation, Tran et al. compute a structural index for graphs based on forward-backward $k$-bisimulation~\cite{DBLP:journals/tkde/TranLR13}.
Moreover, they parameterize their notion of bisimulation to a forward-set $L_1$ and a backward-set $L_2$, so that only labels $l \in L_1$ are considered for forward-bisimulation and labels $l \in L_2$ for backward-bisimulation.
However, similar to Buneman \etal, the particular focus of their work is not the actual construction of the structural index, \eg the bisimulation partition.
Rather, they evaluate how one can efficiently optimize query processing on semi-structured data using such an index graph~\cite{DBLP:journals/tkde/TranLR13}.

Schätzle \etal compute a forward $k$-bisimulation on RDF graphs in a sequential and a distributed setting~\cite{DBLP:conf/sigmod/SchatzleNLP13}.
For a small synthetic dataset ($\sim 1M$ RDF-triples) the sequential algorithm slightly outperforms the distributed one, whereas for increasing size of the dataset, the distributed algorithm clearly outperforms the sequential one.
Different to this, Kaushik et al. propose a summarization technique based on backward $k$-bisimulation~\cite{DBLP:conf/icde/KaushikSBG02}.
Their A(k)-Index serves as a framework for computing backward $k$-bisimulation of a graph $G$, which results in an index graph $I(G)$.
Qun et al.~\cite{DBLP:conf/sigmod/QunLO03} have extended the A(k)-Index to a D(k)-Index, which also bases on bisimulation but focuses on query optimization.
To this end, the D(k)-Index dynamically adapts its structure according to the current query load.

Regarding distributed computation on bisimulation on external disk-based memory, 
Luo \etal examine structural graph summarization with respect to forward $k$-bisimulation~\cite{DBLP:conf/cikm/LuoFHWB13}.
Furthermore, they empirically observe, that for values of $k > 5$, the summary graphs partition blocks change little or not at all.
Therefore they state, that for summarizing a graph with respect to $k$-bisimulation, it is sufficient to summarize up to a value of $k = 5$~\cite{DBLP:conf/sigmod/LuoFHBW13}.
Finally, Martens et al.~\cite{DBLP:conf/facs2/0001GHHW21} introduced a parallel algorithm for bisimulation and transferred it run the code on massive parallel devices such as graphics processing units (GPUs).
The motivation is that a GPU resembles a form of parallel random access machine, where computations can scale through large amounts of parallelism. 
The presented approach does not support multi-GPU support and is tested on a single NVIDIA Titan RTX with $24$ GB of GPU RAM. 
This significantly limits the applicability to large datasets.
However, the proposed blocking mechanism could be combined with our vertex-centric programming approach to further improve the performance. 

\section{Conclusion and Future Work}

We have presented a parallel algorithm for structural graph summarization.
We have empirically evaluated this algorithm on several synthetic and real-world graph datasets and analyzed its theoretical runtime and space complexity.
We have shown that this parallel algorithm can server as base for an extension to an incremental graph summarization algorithm for structural graph summaries over evolving graphs.
We analyzed the complexity of the incremental algorithm and empirically evaluated time and space requirements for three typical graph summary models on two benchmark and two real-world datasets.
The incremental summarization algorithm almost always outperforms its batch counterpart, even in cases when from one version to the next version about $50\%$ of the graph database changes.
Furthermore, the incremental summarization algorithm outperforms the batch summarization algorithm even when using fewer cores, despite the batch computation benefiting more from parallelization than the incremental algorithm.
Finally, we analyzed the memory overhead of the incremental algorithm induced by the \vertexHashIndex{}.
Using the real-world DyLDO-core dataset, we found that using four cores, the incremental algorithm is on average $5$ to $44$~times faster, with a memory overhead of only $8\%$ ($\pm 1\%$).
Finally, we have produced an optimized version of the parallel algorithm to efficiently compute $k$-bisimulations on large graphs. 
We have demonstrated the applicability of this algorithm on different synthetic and real-world graphs.
We have shown that it scales to large graphs with tens of millions and up to $1$~billion edges.

In our future work, we aim to develop a combination of the two extensions of the base algorithm for parallel graph summarization.
Thus, we plan to produce an incremental $k$-bisimulation algorithm for evolving graphs.
However, this combination is non-trivial since when using the hash-based messaging model, we cannot easily update the hashes when the graph changes.
Our work shows that incremental bisimulation is possible and we plan to further optimize the data structure of our incremental algorithm to allow efficient incremental bisimulation computation.
The vertex-centric programming model Pregel used in our algorithms is inherently iterative, synchronous, and deterministic~\cite{DBLP:conf/sigmod/MalewiczABDHLC10,DBLP:phd/dnb/Erb20} and achieves parallelism very similar to MapReduce~\cite{DBLP:journals/cacm/DeanG10}.
Selecting an optimal partitioning strategy of the graph for a given summary model to compute the summary in a distributed infrastructure (as distinct from the multi-core approach of our experiments), is also an open research topic.  
Triggered by the recent popularity of deep graph representation learning, we also search for solutions to replace the explicit summarization algorithm by graph neural networks.
Finally, an interesting future work would be to aim to exploit our approach on Multi-GPU systems instead of Multi-CPUs for the graph summarization computation. 
This would extend on the work of Martens et al.~\cite{DBLP:conf/facs2/0001GHHW21} for GPU-based bisimulation, which is limited to a single GPU.

\paragraph{\textbf{Conflicts of interest/Competing interests.}}
The authors declare that they have no conflict of interest.

\paragraph{\textbf{Availability of data and material and Code availability.}}
Our framework, the experimental apparatus, and the results are available on GitHub under an open source license~\cite{our-github}.

\paragraph{\textbf{Authors' contributions.}}
\textbf{Till Blume}: data curation, formal analysis, investigation, methodology, software, validation, visualization, writing –- original draft, and writing -- review \& editing.
\textbf{Jannik Rau}: data curation, formal analysis, investigation, methodology, software, validation, visualization, writing~-- review \& editing.
\textbf{Ansgar Scherp}: conceptualization, investigation, writing -- review \& editing, and supervision.
\textbf{David Richerby}: formal analysis, writing -- review \& editing, and supervision.

\extended{
Statement is based on the Contributor Roles Taxonomy, see: 
\url{http://credit.niso.org/}.

\paragraph{\textbf{Ethics approval.}} Not applicable.

\paragraph{\textbf{Consent to participate.}} Not applicable.

\paragraph{\textbf{Consent for publication.}} Not applicable.
}

\bibliographystyle{spmpsci}      

\bibliography{references-clean}

\end{document}